\documentclass[11pt]{article} % use larger type; default would be 10pt

\usepackage[utf8]{inputenc} % set input encoding (not needed with XeLaTeX)

\usepackage{geometry} % to change the page dimensions
\usepackage{graphicx} % support the \includegraphics command and options
\usepackage{color}

\definecolor{Red}{rgb}{1,0,0}
\definecolor{Blue}{rgb}{0,0,1}
\definecolor{Olive}{rgb}{0.41,0.55,0.13}
\definecolor{Green}{rgb}{0,1,0}
\definecolor{MGreen}{rgb}{0,0.8,0}
\definecolor{DGreen}{rgb}{0,0.55,0}
\definecolor{Yellow}{rgb}{1,1,0}
\definecolor{Cyan}{rgb}{0,1,1}
\definecolor{Magenta}{rgb}{1,0,1}
\definecolor{Orange}{rgb}{1,.5,0}
\definecolor{Violet}{rgb}{.5,0,.5}
\definecolor{Purple}{rgb}{.75,0,.25}
\definecolor{Brown}{rgb}{.75,.5,.25}
\definecolor{Grey}{rgb}{.5,.5,.5}

\usepackage{booktabs} % for much better looking tables
\usepackage{array} % for better arrays (eg matrices) in maths
\usepackage{paralist} % very flexible & customisable lists (eg. enumerate/itemize, etc.)
\usepackage{verbatim} % adds environment for commenting out blocks of text & for better verbatim
\usepackage{subfigure} % make it possible to include more than one captioned figure/table in a single float
\usepackage{amsmath,amssymb,amsthm,mathrsfs}
\usepackage[colorlinks]{hyperref}

\usepackage{fancyhdr} % This should be set AFTER setting up the page geometry
\usepackage{sectsty}
\allsectionsfont{\sffamily\mdseries\upshape} % (See the fntguide.pdf for font help)
\makeatletter
\newtheorem*{rep@theorem}{\rep@title}
\newcommand{\newreptheorem}[2]{%
\newenvironment{rep#1}[1]{%
 \def\rep@title{#2 \ref{##1}}%
 \begin{rep@theorem}}%
 {\end{rep@theorem}}}
\makeatother
\theoremstyle{plain}

\newtheorem{theorem}{Theorem}[section]

\newtheorem{lemma}[theorem]{Lemma}

\newtheorem*{theorem*}{Theorem}   %unnumbered theorem
\newreptheorem{theorem}{Theorem} %repeat theorem
\newreptheorem{lemma}{Lemma}      %repeat lemma
\newreptheorem{corollary}{Corollary} %repeat corollary
\theoremstyle{remark}

\newtheorem{remark}[theorem]{Remark}

\theoremstyle{definition}
\newtheorem{definition}{Definition}
\newtheorem*{definition*}{Definition}   %unnumbered definition

\def\cA{{\cal A}}
\def\cB{{\cal B}}

\def\cD{{\cal D}}

\def\cF{{\cal F}}

\def\cN{{\cal N}}

\def\cS{{\cal S}}

\title{A Geometric Analysis of the AWGN channel with a $(\sigma, \rho)$-Power Constraint}
\author{ Varun  Jog and Venkat Anantharam\\
Department of Electrical Engineering and Computer Sciences, UC Berkeley\\
\texttt{\small varunjog@berkeley.edu, ananth@berkeley.edu}}

\begin{document}
\maketitle

\begin{abstract}
In this paper, we consider the AWGN channel with a power constraint called the $(\sigma, \rho)$-power constraint, which is motivated by energy harvesting communication systems. Given a codeword, the constraint imposes a limit of $\sigma + k \rho$ on the total power of  any $k\geq 1$ consecutive transmitted symbols. Such a channel has infinite memory and evaluating its exact capacity is a difficult task. Consequently, we establish an $n$-letter capacity expression and seek bounds for the same. We obtain a lower bound on capacity by considering the volume of $\cS_n(\sigma, \rho) \subseteq \mathbb{R}^n$, which is the set of all length $n$ sequences satisfying the $(\sigma, \rho)$-power constraints. For a noise power of $\nu$, we obtain an upper bound on capacity by considering the volume of $\cS_n(\sigma, \rho) \oplus B_n(\sqrt{n\nu})$, which is the Minkowski sum of $\cS_n(\sigma, \rho)$ and the $n$-dimensional Euclidean ball of radius $\sqrt{n\nu}$. We analyze this bound using a result from convex geometry known as Steiner's formula, which gives the volume of this Minkowski sum in terms of the intrinsic volumes of $\cS_n(\sigma, \rho)$. We show that as the dimension $n$ increases, the logarithm of the sequence of intrinsic volumes of  $\{\cS_n(\sigma, \rho)\}$ converges to a limit function under an appropriate scaling. The upper bound on capacity is then expressed in terms of this limit function. We derive the asymptotic capacity in the low and high noise regime for the $(\sigma, \rho)$-power constrained AWGN channel, with strengthened results for the special case of $\sigma = 0$, which is the amplitude constrained AWGN channel.\\

\end{abstract}

\section{Introduction}
 
The additive white Gaussian noise (AWGN) channel is one of the most basic channel models studied in information theory. This channel is represented by a sequence of channel inputs denoted by $X_i$, and an input-independent additive noise $Z_i$.  The noise variables $Z_i$ are assumed to be independent and identically distributed as $\cN(0, \nu)$. The channel output $Y_i$ is given by
\begin{equation}
Y_i = X_i + Z_i \text{~~for~~} i\geq 1.
\end{equation}
The Shannon capacity this channel is infinite in case there are no constraints on the channel inputs $X_i$; however, practical considerations always constrain the input in some manner. These input constraints are often defined in terms of the power of the input. For a channel input $(x_1, x_2, \dots, x_n)$, the most common power constraints encountered are:\\
\begin{itemize}
\item[$\mathbf{(AP)}$:] An average power constraint of $P > 0$, which says that
\begin{equation*}
\sum_{i=1}^n x_i^2 \leq nP.
\end{equation*}

\item[$\mathbf{(PP)}$:] A peak power constraint of $A > 0$, which says that
\begin{equation*}
|x_i| \leq A \text{~~for all~~} 1\leq i\leq n.
\end{equation*}

\item[$\mathbf{(APP)}$:] An average and peak power constraint, consisting of $\mathbf{(AP)}$ and $\mathbf{(PP)}$ simultaneously.
\end{itemize}
The AWGN channel with the $\mathbf{(AP)}$ constraint was first analyzed by Shannon~\cite{shannon}. Shannon showed that the capacity $C$ for this constraint is given by
\begin{equation}
C = \sup_{E[X^2] \leq P} I(X; Y) = \frac{1}{2} \log \left(1 + \frac{P}{\nu}\right),
\end{equation}
and the supremum is attained when $X \sim \cN(0, P)$. Here capacity is defined in the usual sense, due to Shannon. See Section \ref{section: capacity} for a precise definition.

Compared to the  $\mathbf{(AP)}$ constraint, fewer results exist about the $\mathbf{(PP)}$ constrained AWGN. The AWGN channel with the $\mathbf{(PP)}$ constraints was first analyzed by Smith \cite{smith1971information}. Smith showed that the channel capacity $C$ in this case is given by
\begin{equation}\label{eq: smith}
C = \sup_{|X| \leq A} I(X; Y).
\end{equation}
Unlike the $\mathbf{(AP)}$ case, the supremum in equation \eqref{eq: smith} does not have a closed form expression. Using tools from complex analysis, Smith established that the optimal input distribution attaining the supremum in  equation \eqref{eq: smith} is discrete, and is supported on a finite number on points in the interval $[-A,A]$. He proposed an algorithm to numerically evaluate this optimal distribution, and thus the capacity. Smith also analyzed the $\mathbf{(APP)}$ constrained AWGN channel and derived similar results. In a related problem, Shamai \& Bar-David \cite{shamai1995capacity} studied the quadrature Gaussian channel with $\mathbf{(APP)}$ constraints, and extended Smith's techniques to establish analogous capacity results for the same.
\medskip

%Shamai \& Bar-David \cite{shamai1995capacity} looked at the quadrature Gaussian channel with  $\mathbf{(PP)}$ and  $\mathbf{(APP)}$ constraints and established results analogous to Smith \cite{smith1971information}.

Our work in this paper is primarily concerned with a  power constraint, which we call a $(\sigma, \rho)$-power constraint, defined as follows:

\begin{definition*}
Let $\sigma, \rho \geq 0$. A codeword $(x_1, x_2, \dots, x_n)$ is said to satisfy a \emph{$(\sigma, \rho)$-power constraint} if
\begin{equation}\label{eq: def}
\sum_{j = k+1}^l x_j^2 \leq \sigma + (l-k)\rho~,~\forall~ 0 \leq k < l \leq n.
\end{equation}
\end{definition*}

These constraints are motivated by \emph{energy harvesting communication systems}, a research area which has seen a surge of interest in recent years. Energy harvesting (EH) is a process by which energy derived from an external source is captured, stored, and harnessed for applications.  For example, harvested energy in the form of solar, thermal, or kinetic energy is converted into electrical energy using photoelectric, thermoelectric, or piezoelectric materials, and is used to power electronic devices. Energy which is harvested is generally present as ambient background and is free. EH devices are efficient, cheap, and require low maintenance, making them an attractive alternative to battery-powered devices. The problem of communicating over a noisy channel using harvested energy is encountered  in a prominent application of EH: wireless sensor networks. Typically, sensor nodes used in such networks are battery-powered and thus have finite lifetimes. Since EH sensor nodes are capable of harvesting energy for their functioning, they have potentially infinite lifetimes and thereby have many advantages over their battery-powered counterparts \cite{sudevalayam2011energy}.
\begin{figure}[hh]
\begin{center}
\includegraphics[scale = 0.4]{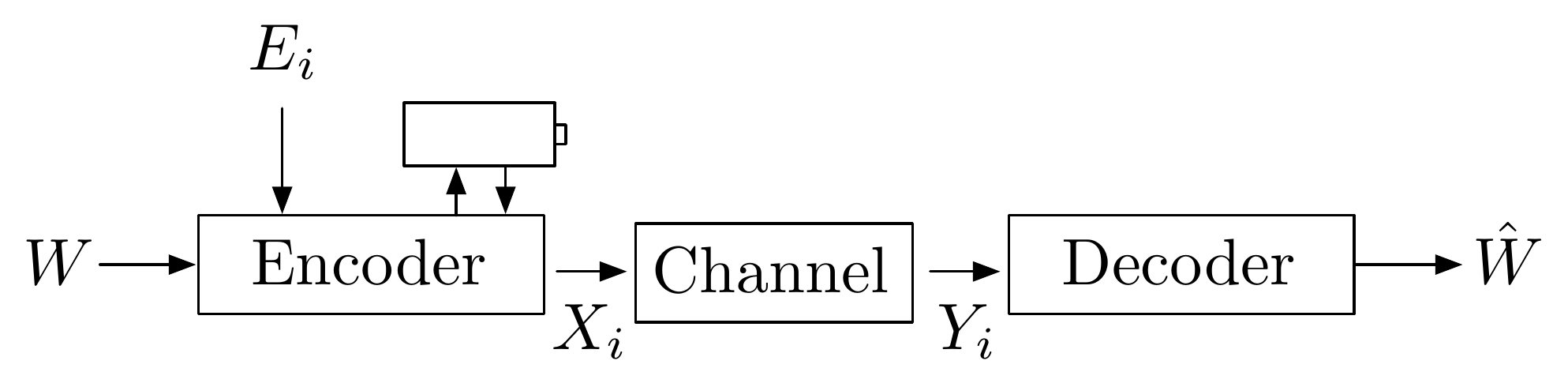}
\caption{Block diagram of a general energy harvesting communication system}\label{fig: generaleh}
\end{center}
 \end{figure}

We can model communication scenarios like the ``EH sensor node" via a general energy harvesting communication system shown in  Figure \ref{fig: generaleh}. Here, the transmitter is capable of harvesting energy, and uses it to transmit a codeword $X^n$, corresponding to a message $W$. The transmitter has a battery to store the excess unutilized energy, which can be used for transmission later. The amount of energy harvested in time slot $i$, denoted by $E_i$, can be modeled as a stochastic process. The process $E_i$, along with the battery capacity, determines the power constraints that the codeword $X^n$ has to satisfy. This codeword is transmitted over a noisy channel, and the receiver decodes $W$ using the channel output $Y^n$. 
A natural channel to study in this setting is the classical additive Gaussian noise (AWGN) channel. Suppose we have a channel model as in Figure \ref{fig: gaussianeh}; namely,  an AWGN channel with an energy harvesting transmitter which harvests a constant $\rho$ amount of energy per time slot, and which has a battery of capacity $\sigma$ attached to it. 

\begin{figure}[h]
\begin{center}
\includegraphics[scale = 0.45]{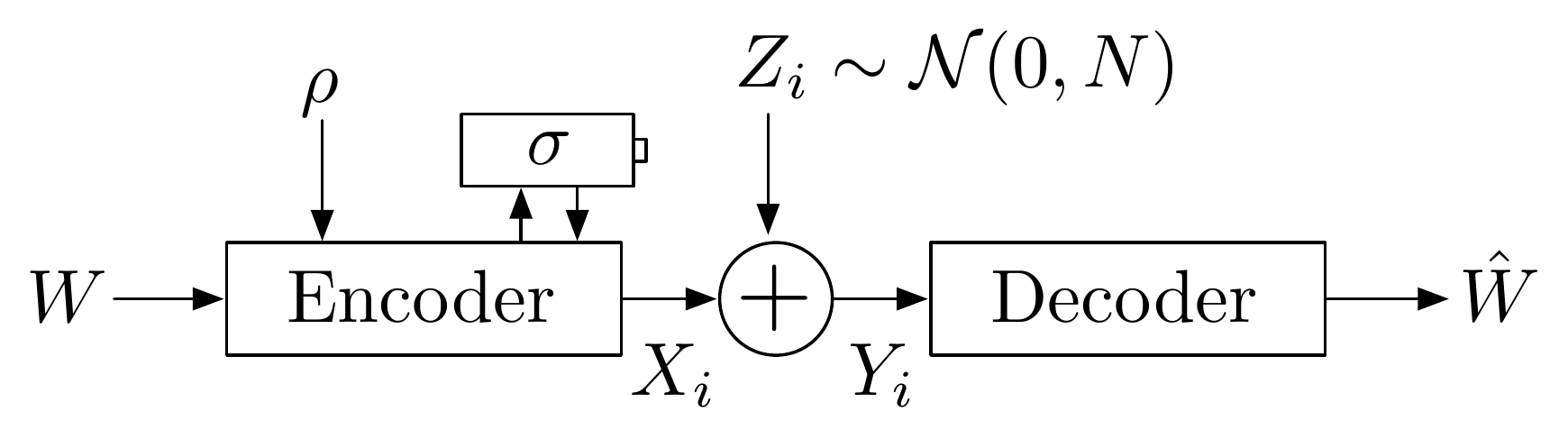}
\caption{$(\sigma, \rho)$-power constrained AWGN channel} \label{fig: gaussianeh}
\end{center}
\end{figure}

To understand the power constraints imposed on a transmitted codeword $(x_1, x_2, \dots, x_n)$ in this scenario, we define a \emph{state} $\sigma_i$, for each $i \geq 0$ as
\begin{align}
\sigma_0 = \sigma, \text{ and }
\sigma_{i+1} = \min( \sigma, \sigma_i + \rho - x_i^2 )~. \label{eq: batterycentered}
\end{align}
From the energy harvesting viewpoint, we can think of the state $\sigma_i$ as the charge in the battery at time $i$ before transmitting $x_i$, assuming the battery started out fully charged at time $0$. Denote by $\cS_n(\sigma, \rho) \subseteq \mathbb{R}^n$ the set
\begin{equation}
\cS_n(\sigma, \rho) = \{ x^n \in \mathbb R^n : \sigma_i \geq 0~,~ \forall~ 0 \leq i \leq n \}.
\end{equation}

In words, the set $\cS_n(\sigma, \rho)$ consists of sequences $(x_1, x_2, \dots, x_n)$ such that at no point during its transmission, is there a need to overdraw the battery. Thus, this set is precisely the set of all possible length $n$ sequences which the transmitter is capable of transmitting. Telescoping the minimum in equation (\ref{eq: batterycentered}), we get that for all $i \geq 0$, 
\begin{equation}
\sigma_{i+1} = \min\left(\sigma,~ \sigma+\rho - x_i^2,~ \cdots,~  \sigma + i\rho - \sum_{j=1}^{i} x_j^2 \right).
\end{equation}
Using the condition $\sigma_i \geq 0$ for all $i$, we obtain another characterization of $\cS_n(\sigma, \rho)$:
\begin{equation} 
\cS_n(\sigma, \rho) = \{ x^n \in \mathbb R^n: \sum_{j = k+1}^l x_j^2 \leq \sigma + (l-k)\rho~,~\forall~ 0 \leq k < l \leq n \},
\end{equation}
which is exactly the $(\sigma, \rho)$-power constraint defined in equation \eqref{eq: def}.  It is interesting to note that such $(\sigma, \rho)$-constraints were originally introduced by Cruz \cite{cruz1991calculus1, cruz1991calculus2} in connection with the study of packet-switched networks. We first look at the $(\sigma, \rho)$-power constraint for the extreme cases; namely, $\sigma = 0$ and $\sigma = \infty$.

\subsection*{No battery:}
Suppose that the battery capacity $\sigma$ is $0$; i.e., unused energy in a time slot cannot be stored for future transmissions. We can easily check that for a  transmitted codeword $(x_1,x_2, ..., x_n)$, the power constraints 
\begin{equation}
x_i^2 \leq \rho \mbox{, for every } 1 \leq i \leq n
\end{equation}
are necessary and sufficient to satisfy the inequalities in \eqref{eq: def}. Thus, the case of $\sigma = 0$ is simply the $\mathbf{(PP)}$ constraint of $\sqrt \rho$. 

\subsection*{Infinite battery:}
Consider the case where the battery capacity is now infinite, so that any unused energy can be saved for future transmissions. We assume that the battery is initially empty, but we can equally well assume it to start with any finite amount of energy in this scenario. The constraints imposed on a transmitted codeword $(x_1, x_2, ..., x_n)$ are
\begin{equation}
\sum_{i=1}^k x_i^2 \leq k\rho  \mbox{, for every } 1 \leq i \leq n.
\end{equation}
It was shown by Ozel \& Ulukus \cite{ozel2012achieving} that the strategy of initially saving energy and then using a Gaussian codebook achieves capacity, which is $\frac{1}{2} \log \left(1 + \frac{\rho}{N} \right)$. In fact,  \cite{ozel2012achieving} considers not just constant $E_i$, but a more general case of i.i.d.\ $E_i$.

\subsection*{Finite battery:}

An examination of equations \eqref{eq: def} and \eqref{eq: batterycentered} reveals that the energy constraint on the $n+1$-th symbol $x_{n+1}$, depends on the entire history of symbols transmitted up to time $n$. This infinite memory makes the exact calculation of  channel capacity under these constraints a difficult task. For some recent work on \emph{discrete} channels with finite batteries, we refer the reader to Tutuncuoglu et. al. \cite{tutuncuoglu2013binary, tutuncuoglu2014improved} and Mao \& Hassibi \cite{mao2013capacity}. An alternative model of an  AWGN channel with a finite battery was also considered by Dong et. al. \cite{dong2014near}, where the authors established approximate capacity results for the same.  
\medskip

In this paper, we will primarily focus on getting bounds on the channel capacity of an AWGN channel with $(\sigma, \rho)$-power constraints. Our work can be broadly divided into two parts; the first part deals with getting a lower bound, and the second part with getting an upper bound. The approach for both these parts relies on analyzing the geometric properties  of the sets $\cS_n(\sigma, \rho)$. In what follows, we briefly describe our results.

\subsection{Lower bound on capacity}
We obtain a lower bound on the channel capacity in terms of the volume of $\cS_n(\sigma, \rho)$. More precisely, we define $v(\sigma, \rho)$ to be the \emph{exponential growth rate of volume} of the family $\{\cS_n(\sigma, \rho)\}$:
\begin{equation}
v(\sigma, \rho) := \lim_{n\to\infty} \frac{1}{n} \log \text{Vol}(\cS_n(\sigma, \rho)),
\end{equation}
where the limit can be shown to exist by subadditivity. 
Our first result is Theorem \ref{thm: bounds} in Section \ref{section: volume}, which contains a lower bound on the channel capacity:
\begin{reptheorem}{thm: bounds}
The capacity $C$ of  an AWGN channel with a $(\sigma, \rho)$-power constraints and noise power $\nu$ satisfies
\begin{equation}\label{eq: rep bounds UB}
\frac{1}{2} \log \left( 1 + \frac{e^{2v(\sigma, \rho)}}{2\pi e\nu} \right) \leq C \leq \frac{1}{2}\log \left( 1 + \frac{\rho}{\nu} \right).
\end{equation}
\end{reptheorem}

Having obtained this lower bound on $C$, it is natural to study the dependence of $v(\sigma, \rho)$ on its arguments. Theorem \ref{thm: range of v} in Section \ref{section: v} establishes the following:

\begin{reptheorem}{thm: range of v}
For a fixed $\rho$, $v(\sigma, \rho)$ is a monotonically increasing, continuous, and concave function of $\sigma$  over $[0, \infty)$, with its range being $[\log 2\sqrt\rho, \frac{1}{2}\log 2\pi e \rho)$.
\end{reptheorem}

In Section \ref{section: finding v}, we describe a numerical method to find $v(\sigma, \rho)$ for any value of the pair $(\sigma, \rho)$. This calculated value can be used to compare the lower and upper bounds in Theorem \ref{thm: bounds} for different values of $\sigma$ for a fixed $\rho$. From the energy-harvesting perspective, this comparison indicates the benefit that a finite battery of capacity $\sigma$ has on the channel capacity. With this we conclude the first part of the paper.

\subsection{Upper bound on capacity}
The upper bound on capacity in \eqref{eq: rep bounds UB} is not satisfactory as it does not depend on $\sigma$. Our approach to deriving an improved upper bound on capacity also involves a volume calculation. However, the improved upper bound is not in terms of the volume of $\cS_n(\sigma, \rho)$, but in terms of the volume of the Minkowski sum of $\cS_n(\sigma, \rho)$ and a ``noise ball." Let $B_n(\sqrt{n\nu})$ be the Euclidean ball of radius $\sqrt{n\nu}$. The Minkowski sum of $\cS_n(\sigma, \rho)$ and $B_n(\sqrt{n\nu})$ (also called the parallel body of $\cS_n(\sigma, \rho)$ at a distance $\sqrt{n\nu}$), is defined by
\begin{equation}
\cS_n(\sigma, \rho) \oplus B_n(\sqrt{n\nu}) = \{x^n + z^n ~|~ x^n \in \cS_n(\sigma, \rho) , z^n \in B_n(\sqrt{n\nu}) \}~.
\end{equation}
In Section \ref{section: outer}, we prove the following upper bound on capacity:
\begin{reptheorem}{thm: voronoi}
The capacity $C$ of an AWGN channel with a $(\sigma, \rho)$-power constraint and noise power $\nu$ satisfies
 \begin{equation} \label{eq: minkowski}
C \leq \lim_{\epsilon \to 0_+} \limsup_{n \to \infty} \frac{1}{n} \log  \frac{\text{Vol}(\cS_n(\sigma, \rho) \oplus B_n(\sqrt{n(\nu+ \epsilon)}~))}{\text{Vol}(B_n(\sqrt{n\nu}))}.
\end{equation}
\end{reptheorem}
 This motivates us to define a function $\ell:[0, \infty) \to \mathbb R$, giving the  growth rate of the volume of the parallel body as follows:
\begin{equation}
\ell(\nu) :=  \limsup_{n \to \infty} \frac{1}{n} \log \text{Vol}(\cS_n(\sigma, \rho)\oplus B_n(\sqrt{n\nu}~)).
\end{equation}
 The upper bound can be restated as
\begin{equation}
C \leq \limsup_{\epsilon \to 0_+} \left[\ell(\nu+\epsilon) - \frac{1}{2}\log 2\pi e\nu\right].
\end{equation}
To study the properties of $\ell(\cdot)$, we use the following result from convex geometry called \emph{Steiner's formula}:
\begin{reptheorem}{thm: steinerformula} 
Let $K_n \subset \mathbb{R}^n$ be a compact convex set and let $B_n \subset \mathbb{R}^n$ be the unit ball. Denote by $\mu_j(K_n)$ the $j$-th intrinsic volume $K_n$, and by $\epsilon_j$ the volume of $B_j$. Then for $t \geq 0$,
\begin{equation}\label{eq: repsteiner}
Vol(K_n \oplus tB_n) = \sum_{j=0}^{n}\mu_{n-j}(K_n)\epsilon_j t^j.
\end{equation}
\end{reptheorem}

Intrinsic volumes are a fundamental concept in convex and integral geometry. They describe the global characteristics of a set, including the volume, surface area, mean width, and the Euler characteristic. For more details, we refer the reader to Schneider \cite{schneider2013convex} and section $14.2$ of Schneider \& Weil \cite{schneider2008stochastic}.
\smallskip

In Section \ref{section: cube}, we focus on the $\sigma = 0$ case for two reasons. Firstly, intrinsic volumes are notoriously hard to compute for arbitrary convex bodies. But when $\sigma = 0$, the set $\cS_n(\sigma, \rho)$ is simply the cube $[-\sqrt \rho , \sqrt \rho]^n$. The intrinsic volumes of a cube are well known in a closed form, which permits an explicit evaluation of $\ell(\nu)$. In his paper, Smith \cite{smith1971information} numerically evaluated and plotted the capacity of a $\mathbf{(PP)}$ constrained AWGN channel. Based on the plots, Smith noted that as $\nu \to 0$, the channel capacity seemed to satisfy
\begin{equation}
C = \log 2A - \frac{1}{2}\log 2\pi e\nu + o(1),
\end{equation}
where the $o(1)$ terms goes to $0$ as $\nu \to 0$. He gave an intuitive explanation for this phenomenon as follows: Let $X$ be the amplitude-constrained input, let $Z \sim \cN(0,\nu)$ be the noise, and let $Y$ be the channel output. Then for a small noise power $\nu$, $h(Y) \approx h(X),$
%\begin{align*}
%h(Y) >> h(Y|X), \text{~~and~~}
%h(X) >> h(X|Y).
%\end{align*}
%Thus, 
%\begin{align*}
%h(X) &= I(X;Y) + h(X|Y)\\
%&= h(Y) -h(Y|X) + h(X|Y)\\
%&\approx h(Y),
%\end{align*}
and
\begin{align*}
C &= \sup_X I(X;Y)\\
&= \sup_X h(Y) - h(Y|X)\\
&\approx \sup_X h(X)- h(Y|X)\\
&= \log 2A - \frac{1}{2}\log 2\pi e\nu.
\end{align*}
Note that the crux of this argument is that when the noise power is small, $\sup_X h(Y) \approx \sup_X h(X) = \log 2A$. This argument can me made rigorous by establishing
\begin{equation} \label{eq: lim of h}
\lim_{\nu \to 0}\left[\sup_X h(X+Z)\right] - \log 2A = 0.
\end{equation}

Recall that our upper bound on capacity is $C \leq \limsup_{\epsilon \to 0_+} \left\{\ell(\nu+\epsilon) - \frac{1}{2}\log 2\pi e\nu\right\}.$ Since $\ell(0) = \log 2A$, the continuity of $\ell$ at $0$ would lead to asymptotic upper bound which agrees with Smith's intuition. The following theorems provide our main result for the case of $\sigma = 0$:
\begin{reptheorem}{thm:  leconte}
The function $\ell(\nu)$ is continuous on $[0, \infty).$ For $\nu > 0$, we can explicitly compute $\ell(\nu)$ via the expression
\begin{equation}
\ell(\nu) = H(\theta^*) + (1-\theta^*)\log 2A +  \frac{\theta^*}{2} \log \frac{2\pi e\nu}{\theta^*},
\end{equation}
where $H$ is the binary entropy function, and $\theta^* \in (0,1)$ satisfies $$\frac{(1-\theta^*)^2}{{\theta^*}^3} = \frac{2A^2}{\pi\nu}.$$
\end{reptheorem}

\begin{reptheorem}{thm: ass cap cube}
The capacity $C$  of an AWGN channel with an amplitude constraint of $A$, and with noise power $\nu$, satisfies the following:
\begin{itemize}
\item[1.] 
When the noise power $\nu \to 0$, capacity $C$ is given by
$$C =  \log 2A - \frac{1}{2}\log 2\pi e\nu + O(\nu^{\frac{1}{3}}).$$
\item[2.]
When the noise power $\nu \to \infty$, capacity $C$ is given by
$$C = \frac{\alpha^2}{2} - \frac{\alpha^4}{4} + \frac{\alpha^6}{6} - \frac{5\alpha^8}{24} + O(\alpha^{10}),$$
where $\alpha = A/\sqrt\nu$.
\end{itemize}
\end{reptheorem}

We also establish a general entropy upper bound, which does not require the noise $Z$ to be Gaussian:
\begin{reptheorem}{thm: AN bound}
Let $A, \nu \geq 0$. Let $X$ and $Z$ be random variables satisfying $|X| \leq A$ a.s.\ and $\text{Var}(Z) \leq \nu$. Then 
\begin{equation}
h(X + Z) \leq \ell(\nu).
\end{equation}
\end{reptheorem}
~\\

Finally, in Section \ref{section: sn} we turn to the case of $\sigma > 0$. Unlike the $\sigma = 0$ case, the intrinsic volumes of $\cS_n(\sigma, \rho)$ are not known in a closed form. For $n \geq 1$, we let $\{\mu_n(0), \cdots, \mu_n(n)\}$ be the intrinsic volumes of $\cS_n(\sigma, \rho)$. The sequence of intrinsic volumes $\{\mu_n(\cdot)\}_{n \geq 1}$ forms a \emph{sub-convolutive} sequence (analyzed in Appendix \ref{appendix:  subc}). Convergence properties of such sequences can be effectively studied using large deviation techniques; in particular, the G\"{a}rtner-Ellis theorem \cite{dembo1998large}. These convergence results for intrinsic volumes can be  used in conjunction with Steiner's formula to establish results about $\ell$ and the asymptotic capacity of a $(\sigma, \rho)$-constrained channel in the low noise regime. Our main results here are:

\begin{reptheorem}{thm:  sn leconte}
Define $\ell(\nu)$ as
\begin{equation}
\ell(\nu) =  \limsup_{n \to \infty} \frac{1}{n} \log \text{Vol}(\cS_n(\sigma, \rho)\oplus B_n(\sqrt{n\nu}~)).
\end{equation}
For $n \geq 1$, define $G_n: \mathbb{R} \to \mathbb{R}$ and $g_n: \mathbb R \to \mathbb R$ as 
\begin{equation}
G_n(t) = \log \sum_{j=0}^n \mu_n(j)e^{jt}, \text{~~and~~} g_n(t) = \frac{G_n(t)}{n}.
\end{equation}
Define $\Lambda$ to be the pointwise limit of the sequence of functions $\{g_n\}$, which we show exists.
Let $\Lambda^*$ be the convex conjugate of $\Lambda$. Then the following hold:
\begin{enumerate}
\item
$\ell(\nu)$ is continuous on $[0, \infty).$
\item
For $\nu > 0$, 
\begin{equation}
\ell(\nu) = \sup_{\theta \in [0,1]} \left[ -\Lambda^* (1-\theta)+ \frac{\theta}{2}\log \frac{2\pi e \nu}{\theta}\right].
\end{equation}
\end{enumerate}
 \end{reptheorem}

 \begin{reptheorem}{thm: ass cap sn}
The capacity $C$ of an AWGN channel with $(\sigma, \rho)$-power constraints and noise power $\nu$ satisfies the following:
\begin{itemize}
\item[1.]
When the noise power $\nu \to 0$, capacity $C$ is given by
$$C =  v(\sigma, \rho) - \frac{1}{2}\log 2\pi e\nu + \epsilon(\nu),$$
where $\epsilon(\cdot)$ is a function such that $\lim_{\nu \to 0}\epsilon(\nu) = 0.$
\item[2.]
When noise power $\nu \to \infty$, capacity $C$ is given by
$$C = \frac{1}{2}\left(\frac{\rho}{\nu}\right)^2 - \frac{1}{4}\left(\frac{\rho}{\nu}\right)^4 +  \frac{1}{6}\left(\frac{\rho}{\nu}\right)^6 +  O\left(\left(\frac{\rho}{\nu}\right)^8\right).$$
\end{itemize}
\end{reptheorem}

%=============================================================================%

\section{Channel Capacity}\label{section: capacity}

We define channel capacity as per the usual convention \cite{cover1991}:
\begin{definition}
A $(2^{nR}, n)$ code for the AWGN channel with a $(\sigma, \rho)$-power constraint consists of the following:
\begin{enumerate}
\item
A set of messages $\{1, 2, \dots, 2^{\lfloor nR \rfloor}\}$
\item
An encoding function $f: \{1, 2, \dots, 2^{\lfloor nR \rfloor}\} \to  \cS_n(\sigma, \rho)$, yielding codewords $f(1),  \dots, f(2^{\lfloor nR \rfloor})$
\item
A decoding function $g : \mathbb R^n \to \{1, 2, \dots, 2^{\lfloor nR \rfloor}\}$
\end{enumerate}
A rate $R$ is said to be achievable if there exists a sequence of $(2^{nR}, n)$ codes such the that probability of decoding error diminishes to $0$ as $n \to \infty$. The capacity of this channel is the supremum of all achievable rates.
\end{definition}
 Shannon's formula for channel capacity
\begin{equation}\label{eq: shannon 1}
C = \sup_X I(X;Y),
\end{equation}
is valid if the channel is \emph{memoryless}. For a channel with memory, one can often generalize this expression to
\begin{equation}\label{eq: shannon 2}
C = \lim_{n \to \infty}\left[ \sup_{X^n} \frac{1}{n} I(X^n; Y^n) \right],
\end{equation}
but this formula does not always hold. Dobrushin \cite{dobrushin1963general} showed that channel  capacity is given by formula \eqref{eq: shannon 2} for a class of channels called \emph{information stable} channels. Checking information stability for specific channels can be quite challenging. Fortunately, in the case of a $(\sigma, \rho)$-power constrained AWGN channel, we can establish formula \eqref{eq: shannon 2} without having to check for information stability. We prove the following theorem:

\begin{theorem}\label{thm: capacity}
For $n \in \mathbb{N}$, let $\cF_n$ be the set of all probability distributions supported on $\cS_n(\sigma, \rho)$. The capacity $C$ of a $(\sigma, \rho)$-power-constrained scalar AWGN channel is given by
\begin{equation}
C = \lim_{n \to \infty} \frac{1}{n}  \sup_{p_{X^n}(x^n) \in \cF_n} I(X^n; Y^n).
\end{equation}
\end{theorem}

\begin{proof}
Let $N$ be a positive integer. Without loss of generality, we can assume that coding is done for block lengths which are multiples  $N$, say $nN$. For codes over such blocks, we relax the $(\sigma, \rho)$ constraints as follows. For every transmitted codeword $(x_1, x_2, \cdots, x_{nN})$, each consecutive block of $N$ symbols has to lie in $\cS_N(\sigma, \rho)$; i.e.,
\begin{equation}\label{EqnBlock}
(x_{kN+1}, x_{kN+2}, ..., x_{(k+1)N}) \in \cS_N(\sigma, \rho), \mbox{ for } 0 \leq k \leq n-1.
\end{equation}
Note that this is indeed a relaxation because a codeword satisfying the constraint \eqref{EqnBlock} is not guaranteed to satisfy the $(\sigma, \rho)$-constraints but any codeword satisfying the $(\sigma, \rho)$-constraints necessarily satisfies the constraint \eqref{EqnBlock}. The capacity of this channel $C_N$ can be written as
\begin{equation}
C_N = \sup_{p_{X^N}(x^N) \in \cF_N} I(X^N; Y^N).
\end{equation}
This capacity provides an upper bound to  $NC$ for any choice of $N$. Thus, we have the bound
\begin{equation}\label{eq: ub1}
C \leq \inf_N \frac{C_N}{N}.
\end{equation}
To show that $\inf_N C_N/N$ is $\lim_N C_N/N$, we first note that
$$I(X_1^{M+N}; Y_1^{M+N}) \leq I(X_1^M; Y_1^M) + I(X_{M+1}^{M+N}; Y_{M+1}^{M+N})~.$$
Taking the supremum on both sides with $p_{X^{M+N}}$ ranging over $\cF_{M+N}$,
\begin{align} \label{eq: C_N}
C_{M+N} &\leq \sup_{p_{X^{M+N}}(x^{M+N}) \in \cF_{M+N}} \left(I(X_1^M; Y_1^M) + I(X_{M+1}^{M+N}; Y_{M+1}^{M+N}) \right) \\
&\stackrel{(a)}\leq \sup_{p_{X^M}(x^M) \in \cF_M} I(X_1^M; Y_1^M) + \sup_{p_{X^N}(x^N) \in \cF_N} I(X_1^N; Y_1^N)\\
&= C_M + C_N.
\end{align}
Here $(a)$ follows due to the containment $\cF_{M+N} \subseteq \cF_M \times \cF_N$. This calculation shows that $\{C_{N}\}$ is a sub-additive sequence. Applying Fekete's lemma \cite{steele1997probability} we conclude that
$\lim_{N} C_N/N$ exists and equals $\inf_N C_N/N$, and thereby establish the upper bound
\begin{equation}\label{eq: ub2}
C \leq \lim_{N \to \infty} \frac{C_N}{N}. 
\end{equation}
We now show that $C$ is lower bounded by $\lim_N C_N/N$. Given any $\mathbf{x}_1, \mathbf{x}_2, \cdots, \mathbf{x}_n \in \cS_N$, the concatenated sequence $\mathbf{x}_1 \cdots \mathbf{x}_n$ need not always satisfy  the $(\sigma, \rho)$ power constraints. However, if we append $k = \lceil \frac{\sigma}{\rho} \rceil$ zeros to each $\mathbf{x}_i$ and then concatenate them, the $n(N+k)$ length string so formed lies in $\cS_{n(N+k)}$. This is because transmitting $\lceil \frac{\sigma}{\rho} \rceil$ zeros  after  each $\mathbf x_i$ ensures that the state, as defined in equation \eqref{eq: batterycentered}, returns to $\sigma$ before the transmission of $\mathbf{x}_{i+1}$ begins. Let us define a new set 
$$\hat{\cS}_N = \{ x^{N+k} : x_1^N \in \cS_N, x_{N+1}^{N+k} = \mathbf{0}\}.$$
The earlier discussion implies that 
\begin{equation}\label{eq: lb0}
\underbrace{\hat \cS_N \times \cdots \hat \cS_N}_{n \text{ times}} \subseteq \cS_{n(N+k)}.
\end{equation}
Equation \eqref{eq: lb0} implies that any block coding scheme which uses symbols from $\hat \cS_N$ is also a valid coding scheme under the $(\sigma, \rho)$ power constraints. The achievable rate for such a scheme can therefore provide a lower bound to $C$. This achievable rate is simply ${C_N}$, as the final $k$ transmissions in each symbol carry no information. Thus the per transmission achievable rate is $\frac{C_N}{N+k}$, and we get that
\begin{equation}\label{eq: lb1}
C \geq \frac{C_N}{N+k},
\end{equation}
for all $N$. Taking the limit as $N \to \infty$, we arrive at the bound
\begin{equation}\label{eq: lb2}
C \geq \lim_{N \to \infty} \frac{C_N}{N}.
\end{equation}
The containment (\ref{eq: ub2}), together with the inequality (\ref{eq: lb2}), completes the proof.
\end{proof}

\section{Lower-bounding capacity}\label{section: volume}
Coding with the $(\sigma,\rho)$ constraints can be thought of as trying to fit the largest number of centers of noise balls in $\mathcal{S}_n$, such that the noise balls are asymptotically approximately disjoint. One might therefore hope to get a packing based upper bound on capacity through the volume of $\cS_n$. We shall show that the volume of $\cS_n$ surprisingly yields a neat lower bound on capacity. 
\smallskip

Let $V_n(\sigma, \rho)$ denote the volume of $\cS_n(\sigma, \rho)$. We look at the exponential growth rate of this volume defined by
\begin{equation}
 v(\sigma, \rho) := \lim_{n \to \infty} \frac{\log V_n(\sigma, \rho)}{n}.
\end{equation}
Our first lemma is to establish the existence of the limit in the definition of $v(\sigma, \rho)$.
\begin{lemma}\label{lemma: limit}
$\lim_{n \to \infty} \frac{\log V_n(\sigma, \rho)}{n} $ exists.
\end{lemma}
\begin{proof}
The containment $\cS_{m+n}(\sigma, \rho) \subseteq \cS_m(\sigma, \rho) \times \cS_n(\sigma, \rho)$ gives
$$V_{m+n}(\sigma, \rho) \leq V_m(\sigma, \rho)V_n(\sigma, \rho),$$
which implies
$$\log V_{m+n}(\sigma, \rho) \leq \log V_m(\sigma, \rho)+ \log V_n(\sigma, \rho).$$
This shows that $\log V_n(\sigma, \rho)$ is a sub-additive sequence, and by Fekete's Lemma, the limit $\lim_{n \to \infty} \frac{\log V_n(\sigma, \rho)}{n}$ exists and is equal to $\inf_n \frac{\log V_n(\sigma, \rho)}{n}$ (which may a priori be $-\infty$).
\end{proof}
\begin{theorem}\label{thm: bounds}
The capacity $C$ of  an AWGN channel with  $(\sigma, \rho)$-power constraints and noise power $\nu$ satisfies
\begin{equation}
\frac{1}{2} \log \left( 1 + \frac{e^{2v(\sigma, \rho)}}{2\pi e\nu} \right) \leq C \leq \frac{1}{2}\log \left( 1 + \frac{\rho}{\nu} \right).
\end{equation}
\end{theorem}
\begin{proof}
%recall capacity expression, maximizing mutual information is same as maximizing output entropy. Instead of best p(x^n) choose uniform distribution on available set. Input entropy is log volume. Lower bound output entropy using EPI, and get lower bound.
Clearly, $C$ is upper bounded by the capacity for the $\sigma = \infty$ case (with zero initial battery condition), which by \cite{ozel2012achieving} is $\frac{1}{2} \log (1 + \frac{\rho}{\nu})$.
\smallskip

Let the noise $Z \sim \cN(0, \nu)$. To prove the lower bound, recall the capacity expression in Theorem \ref{thm: capacity}:
\begin{align}
C &= \lim_{n \to \infty} \frac{1}{n}  \sup_{p_{X^n}(x^n) \in \cF_n} I(X^n; Y^n)\\
&= \lim_{n \to \infty} \frac{1}{n} \sup_{p_{X^n}(x^n) \in \cF_n} h(Y^n) - h(Z^n)\\
&= \lim_{n \to \infty} \frac{1}{n} \sup_{p_{X^n}(x^n) \in \cF_n} h(Y^n) - \frac{1}{2}\log 2 \pi e \nu  \label{eq: mutual}
\end{align}
Thus, calculating capacity requires maximizing the output differential entropy $h(Y^n)$. Using Shannon's entropy power inequality, we have
\begin{equation}\label{eq: epi}
e^{\frac{2h(Y^n)}{n}} \geq e^{\frac{2h(X^n)}{n}} + e^{\frac{2h(Z^n)}{n}}.\\
\end{equation}
Thus,
\begin{align*}
\sup_{p_{X^n}(x^n) \in \cF_n} e^{\frac{2h(Y^n)}{n}} &\geq \sup_{p_{X^n}(x^n) \in \cF_n} e^{2\frac{h(X^n)}{n}} + 2\pi e\nu\\
&= e^{2\frac{\log V_n}{n}} + 2\pi e\nu.
\end{align*}
Taking logarithms on both sides and letting $n$ tend to infinity, we have
\begin{align}
\lim_{n \to \infty} \sup_{p_{X^n}(x^n) \in \cF_n} \frac{h(Y^n)}{n} \geq \frac{1}{2}\log \left( e^{2v(\sigma, \rho)} + 2\pi e\nu \right),
\end{align}
which, combined with equation (\ref{eq: mutual}) concludes the proof.
\end{proof}
%=============================================================================%

\section{Properties of $v(\sigma, \rho)$}\label{section: v}

%\begin{lemma}
%The following bounds hold for $v(\sigma, \rho)$:
%\begin{equation}\label{eq: bounds}
%\log 2\sqrt{\rho} \leq v(\sigma, \rho) \leq \log \sqrt{2\pi e\rho}.
%\end{equation}
%\end{lemma}

%\begin{proof}
We can readily see that $v(\sigma, \rho)$ is monotonically increasing in both of its arguments. With a little more effort, we can also establish the following simple bounds for $v(\sigma, \rho)$:
\begin{equation}\label{eq: simple bounds}
\log 2\sqrt{\rho} \leq v(\sigma, \rho) \leq \log \sqrt{2\pi e\rho}.
\end{equation}
To show the lower bound from inequality \eqref{eq: simple bounds}, observe that if $x^n$ is such that for every $1 \leq i \leq n$, 
$$ |x_i| \leq \sqrt{\rho},$$
then the $(\sigma, \rho)$-constraints are satisfied. Thus, the cube $[-\sqrt{\rho}, \sqrt \rho]^n$ of volume $(2\sqrt{\rho})^n$ lies inside the set $\cS_n(\sigma, \rho)$, giving the lower bound $$v(\sigma, \rho) \geq \log 2\sqrt{\rho}.$$
For the upper bound, we use the ``total power" constraint,
$$x_1^2 + x_2^2 + \ldots + x_n^n \leq \sigma + n\rho,$$ which implies that $\cS_n(\sigma, \rho) \subseteq B_n(\sqrt{\sigma + n\rho})$, where $B_n(\sqrt{\sigma + n\rho})$ is the Euclidean ball of radius $\sqrt{\sigma + n\rho}$. The volume $V_n(\sigma, \rho)$ of the set $\cS_n(\sigma, \rho)$  is bounded above by the volume of $B_n(\sqrt{\sigma + n\rho})$, which gives 
\begin{align*}
v(\sigma, \rho) &\leq \lim_{n \to \infty} \frac{1}{n} \log\left( \frac{\pi^{\frac{n}{2}}}{\Gamma\left(\frac{n}{2} +1 \right)} (\sigma + n\rho)^{\frac{n}{2}}\right)\\
&= \lim_{n \to \infty} \frac{1}{2}\left(\log \pi + \log(\sigma + n\rho) - \log\left(\frac{n}{2e} \right) \right)\\
&= \frac{1}{2} \log 2\pi e \rho.
\end{align*}
%\end{proof}
Note that when $\rho = 0$, then $v(\sigma, 0) = -\infty$ for any value of $\sigma$. Henceforth, we assume $\rho > 0$. When $\sigma = 0$, the set $\cS_n(\sigma, \rho)$ degenerates to the cube $[-\sqrt \rho, \sqrt \rho]^n$, which has the volume growth rate exponent of $\log 2\sqrt{\rho}$. It is clear that when $\sigma > 0$, the set $\cS_n(\sigma, \rho)$ contains the cube $[-\sqrt \rho, \sqrt \rho]^n$, implying that
$$V_n(\sigma, \rho) > (2\sqrt \rho)^n.$$
However, this does not immediately imply that $v(\sigma, \rho) > \log 2 \sqrt \rho$. The following theorem is the main result of this section, where we show that such a strict inequality holds, and also prove some other properties of the function $v(\sigma, \rho)$:
\begin{theorem}\label{thm: range of v}
For a fixed $\rho$, $v(\sigma, \rho)$ is a monotonically increasing, continuous, and concave function of $\sigma \in [0, \infty)$, with its range being $[\log 2\sqrt\rho, \frac{1}{2}\log 2\pi e \rho)$.
\end{theorem}
\begin{proof}[Proof of Theorem \ref{thm: range of v}]
Theorem \ref{thm: range of v} relies on several lemmas. We state the lemmas here and defer their proofs to Appendix \ref{appendix: section: v}. We first show that it is enough to prove the theorem for $\rho = 1$:
 \begin{lemma}[Proof in Appendix \ref{proof: lemma: v1}]\label{lemma: v1}
Let $v_1(\sigma) = v(\sigma, 1)$. Then $v(\sigma, \rho)$ depends on $v_1(\sigma/\rho)$ according to
\begin{equation}\label{eq: scaling}
v(\sigma, \rho) = \log \sqrt{\rho} + v_1(\sigma/\rho).
\end{equation}
\end{lemma}

Thus, a different value of $\rho$ leads to a function $v(\sigma, \rho)$ which is essentially $v_1(\sigma)$ shifted by a constant. Therefore, if $v_1(\sigma)$ is monotonically increasing, continuous, and concave, so is $v(\sigma, \rho)$ for any other value of $\rho > 0$. In Lemmas \ref{lemma: continuity} and \ref{lemma: concavity of v1}, we establish that $v_1(\sigma)$ is a continuous and concave function on $[0, \infty)$:

\begin{lemma}[Proof in Appendix \ref{proof: lemma: continuity}]\label{lemma: continuity}
The function $v_1(\sigma)$ is continuous on $[0, \infty)$.
\end{lemma}

\begin{lemma}[Proof in Appendix \ref{proof: lemma: concavity of v1}]\label{lemma: concavity of v1}
The function $v_1(\sigma)$ is concave on $[0, \infty)$.
\end{lemma}

To finish the proof, we need to show that the limiting value of $v_1(\sigma)$ as $\sigma \to \infty$ is $\frac{1}{2} \log 2\pi e$. It is useful to define a quantity, which we call \emph{burstiness of a sequence}, as follows: Let $\cA_n$ denote the the $n$-dimensional ball of radius $\sqrt{n}$; i.e.,
$$\cA_n := \left\{x^n ~:~\sum_{i=1}^n x_i^2 \leq n \right\}.$$
Fix $x^n \in \cA_n$. We associate a \emph{burstiness} to each such sequence, defined by
\begin{equation}
\sigma(x^n) := \max_{ 0\leq k < l \leq n }\left( \sum_{i=k+1}^{l}x_i^2 - (l-k) \right).
\end{equation}
Let
$$\cA_n(\sigma) = \{x^n \in \cA_n~:~ \sigma(x^n)\leq \sigma\} .$$
Notice that $\cA_n(\sigma) \subseteq \cS_n(\sigma, 1)$. 
We have $\cA_n(0) = [-1,1]^n$ and $\cA_n(n-1) = \cA_n$. As $\sigma$ increases from $0$ to $n-1$, $\cA_n(\sigma)$ increases from the cube to the entire sphere. We have the following lemma:

\begin{lemma}[Proof in Appendix \ref{proof: lemma: Ansigman}]\label{lemma: Ansigman}
If there exists a sequence $\sigma(n)$ such that
\begin{subequations}
\begin{align}
&\lim_{n\to \infty} \frac{1}{n}\log \mbox{Vol}(\cA_n(\sigma(n))) = \frac{1}{2} \log 2\pi e~ \label{eq: cond1}, \quad \text{and}\\
&\lim_{n \to \infty} \frac{\sigma(n)}{n} = 0 \label{eq: cond2},
\end{align}
\end{subequations}
then $\lim_{\sigma \to \infty} v_1(\sigma) = \frac{1}{2} \log 2\pi e$.
\end{lemma}

Note that the natural choice which satisfies condition (\ref{eq: cond1}) is $\sigma(n) = n-1$, but this does not satisfy condition (\ref{eq: cond2}).
To complete the proof, we show that $\sigma(n) = c\sqrt{n}$ for a suitable constant $c$ satisfies both conditions of Lemma \ref{lemma: Ansigman}, and establish the following result:

\begin{lemma}[Proof in Appendix \ref{proof: lemma: sigma_infty}]\label{lemma: sigma_infty}
$\lim_{\sigma \to \infty} v(\sigma, 1) = \frac{1}{2} \log 2\pi e.$
\end{lemma}
This completes the proof of Theorem \ref{thm: range of v}.
\end{proof}

%=============================================================================%

\section{Numerical method to compute $v(\sigma, \rho)$}\label{section: finding v}
In this section, we briefly discuss the numerical evaluation of $v(\sigma, \rho)$. This discussion is nontechnical and for all the technical details justifying the numerical method, we refer the reader to Appendix \ref{appendix: numerical}.\\

 Numerical computation of $v(\sigma, \rho)$ is enabled by exploiting the idea of state as defined in equation \eqref{eq: batterycentered}. However for ease of analysis and implementation, we define the state slightly differently. Given $(x_1, \dots, x_n) \in \cS_{n}(\sigma, 1)$, define
\begin{equation}\label{eq: phi1_rough}
\phi_n = 
\begin{cases}
\sigma_n &\text{ if } \sigma_n < \sigma,\\
\sigma_{n-1} + 1 - x_n^2 &\text{ if } \sigma_n = \sigma.\\
\end{cases}
\end{equation}
The state $\phi_n$ is a sum of two terms: $\sigma_n$, which is the amount of charge in the battery at time $n$, and the amount of energy wasted at time $n$ due to the limited battery capacity. Note that energy is wasted only when $\sigma_n = \sigma$; i.e., when the battery becomes full. Setting $\phi_0 = \sigma$, equation \eqref{eq: phi1_rough} can also be written as
\begin{equation}\label{eq: phi2_rough}
\phi_{n} = 
\begin{cases}
\phi_{n-1} + 1 - x_n^2 &\text{ if } \phi_{n-1} < \sigma,\\
\sigma + 1 - x_n^2 &\text{ if } \phi_{n-1} \geq \sigma.\\
\end{cases}
\end{equation}
Consider the function $\Phi_n: \cS_{n}(\sigma, 1) \to \mathbb R$, defined by $\Phi_n(x_1, \dots, x_n) = \phi_n$. Thus, $\Phi_n$ maps a point in $\cS_{n}(\sigma, 1)$ to its state at time $n$, as defined in equations \eqref{eq: phi1_rough} and \eqref{eq: phi2_rough}. Let $\lambda_n$ be the Lebesgue measure restricted to $\cS_{n}(\sigma, 1)$. The function $\Phi_n$ induces a measure on $\mathbb R$, which we call $\nu_n$. As $0 \leq \phi_n \leq \sigma+1$ for all $x^n \in \cS_n(\sigma, 1)$, we see that the measure $\nu_n$ is supported on $[0,\sigma+1]$, giving 
$$\nu_n([0,\sigma+1]) = \text{Vol}(\cS_n(\sigma, 1)).$$
Suppose $\nu_n$ is absolutely continuous with respect to the Lebesgue measure on $\mathbb R$; this implies existence of a density $f_n$ corresponding to $\nu_n$, which satisfies
\begin{equation}\label{eq: integral_rough}
\text{Vol}(\cS_n(\sigma,1)) = \int_{x=0}^{\sigma+1} f_n(x) dx.
\end{equation}
Given the state $\phi_n < \sigma$, the symbol $x_{n+1}$ is constrained to lie in $[-\sqrt{\phi_n+1}, \sqrt{\phi_n+1}]$. Furthermore, given $\phi_n$, the symbol $x_{n+1}$ has the Lebesgue measure restricted to this set. Similarly, for $\phi_n \geq \sigma$, the conditional measure of $x_{n+1}$ is the Lebesgue measure restricted to $[-\sqrt{\sigma+1}, \sqrt{\sigma+1}]$. Using equation \eqref{eq: phi2_rough}, we can find a relation between the measures $\nu_n$ and $\nu_{n+1}$ as follows: For $\phi \in [0, \sigma+1]$,
\begin{align}
F_{n+1}(\phi) &:= \nu_{n+1}((-\infty, \phi])\\
=& \int_{x = 0}^{\phi - 1} \int_{t = 0}^{x+1}\frac{f_n(x)}{\sqrt t} dt dx \nonumber\\
 +& \int_{x = \phi - 1}^{\sigma} \int_{t = x - (\phi - 1)}^{ x+1}\frac{f_n(x)}{\sqrt t} dt dx \nonumber\\
 +& \int_{x = \sigma}^{\sigma+1} \int_{t = \sigma - (\phi -1)}^{\sigma+1}\frac{f_n(x)}{\sqrt t} dt dx\\
 =& \int_{x=0}^{\phi-1} 2f_n(x)[\sqrt{x+1}]dx \nonumber\\
+& \int_{x=\phi-1}^\sigma 2f_n(x)[\sqrt{x+1}-\sqrt{x-(\phi-1)}]dx \nonumber\\
+& \int_{x = \sigma}^{\sigma+1} 2f_n(x)[\sqrt{\sigma+1} - \sqrt{\sigma-(\phi-1)}]dx.
\end{align}
Differentiating $F_{n+1}$, we obtain 
\begin{align}\label{eq: fn1_rough}
f_{n+1}(\phi) = \int_{\phi-1}^\sigma \frac{f_n(x)}{\sqrt{x - (\phi - 1)}} dx + \int_\sigma^{\sigma+1} \frac{f_n(x)}{\sqrt{\sigma - (\phi -1)}} dx.
\end{align}
Define the integral operator $A$ as follows:
\begin{equation}
A(x,t) = 
\begin{cases}
\frac{1}{\sqrt{x+1-t}} &\text{ if } 0 \leq x < \sigma \text{ and } 0 \leq t \leq  x+1,\\ 
 \frac{1}{\sqrt{\sigma+1-t}} &\text{ if } \sigma \leq x \leq \sigma+1 \text{ and } 0 \leq t \leq \sigma+1,\\
 0 &\text{ otherwise.}
\end{cases} 
\end{equation}
We can express equation \eqref{eq: fn1_rough} in another form,
\begin{equation}\label{eq: fn2_rough}
f_{n+1}(t) = \int A(x,t)f_n(x)dx,
\end{equation}
denoted by $f_{n+1} = A(f_n)$. Iterating this relation, we obtain
\begin{equation}\label{eq: fn3_rough}
f_{n+1} = A^n f_1.
\end{equation}
Our interest is in $v_1(\sigma)$, which by equations \eqref{eq: integral_rough} and \eqref{eq: fn3_rough} is
\begin{equation}
v_1(\sigma) = \lim_{n \to \infty} \frac{1}{n} \log \int_{0}^{\sigma+1} A^{n-1}f_1(x)dx.
\end{equation}
It seems natural to expect this limit to equal the largest eigenvalue of $A$. Our approach to finding the largest eigenvalue is to discretize $A$; let $h_n = \frac{\sigma+1}{n}$, and let $A_n$ be an $(n+1) \times (n+1)$-matrix such that 
$$A_n(i,j) = h_n \times A\big((i-1)h_n, (j-1)h_n\big) \text{~~for~~} 1 \leq i, j \leq n+1.$$
We can approximate the largest eigenvalue of the matrix $A_n$ using standard methods and expect this value to tend to the largest eigenvalue of $A$ as $n$ becomes large.
\smallskip

Figure \ref{fig: graph_v} shows the plot of $v_1(\sigma)$ obtained using the numerical procedure. Note that as $\sigma$ becomes large, $v_1(\sigma)$ tends to the limit $\frac{1}{2} \log 2\pi e$ in a concave manner,  as per Theorem \ref{thm: range of v}. 

\begin{figure}[h]
\begin{center}
\includegraphics[scale = 0.60]{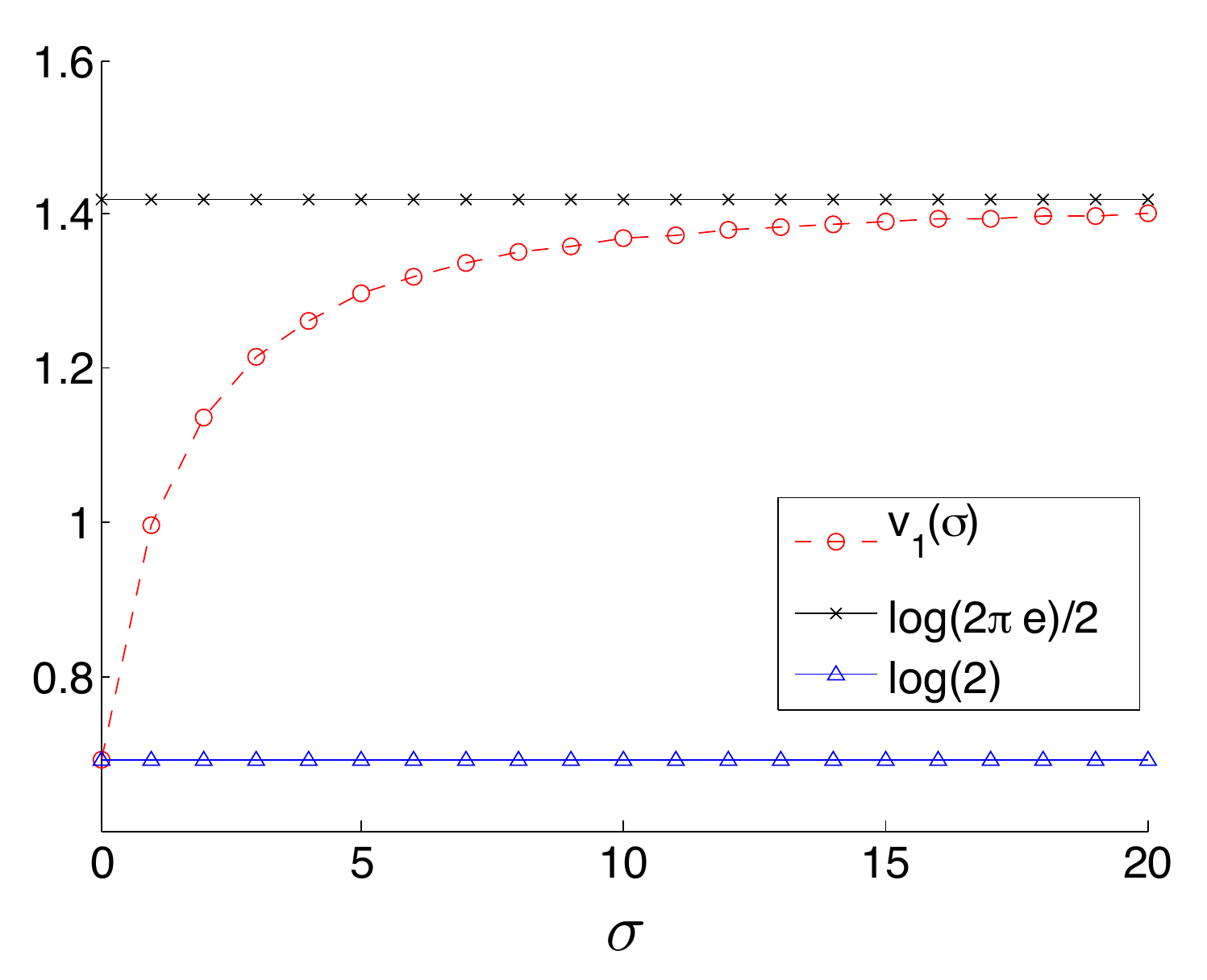}
 \caption{Graph of $v_1(\sigma)$ obtained numerically} \label{fig: graph_v}
\end{center}
\end{figure}

We are now in a position to plot the bounds on  capacity derived in Theorem \ref{thm: bounds}. Figure \ref{fig: bounds} shows a plot of the lower and upper bounds for a fixed value of $\rho$ $( = 1$) and for different values of the noise power $\nu$.
\begin{figure}[h]
\begin{center}
\includegraphics[scale = 0.55]{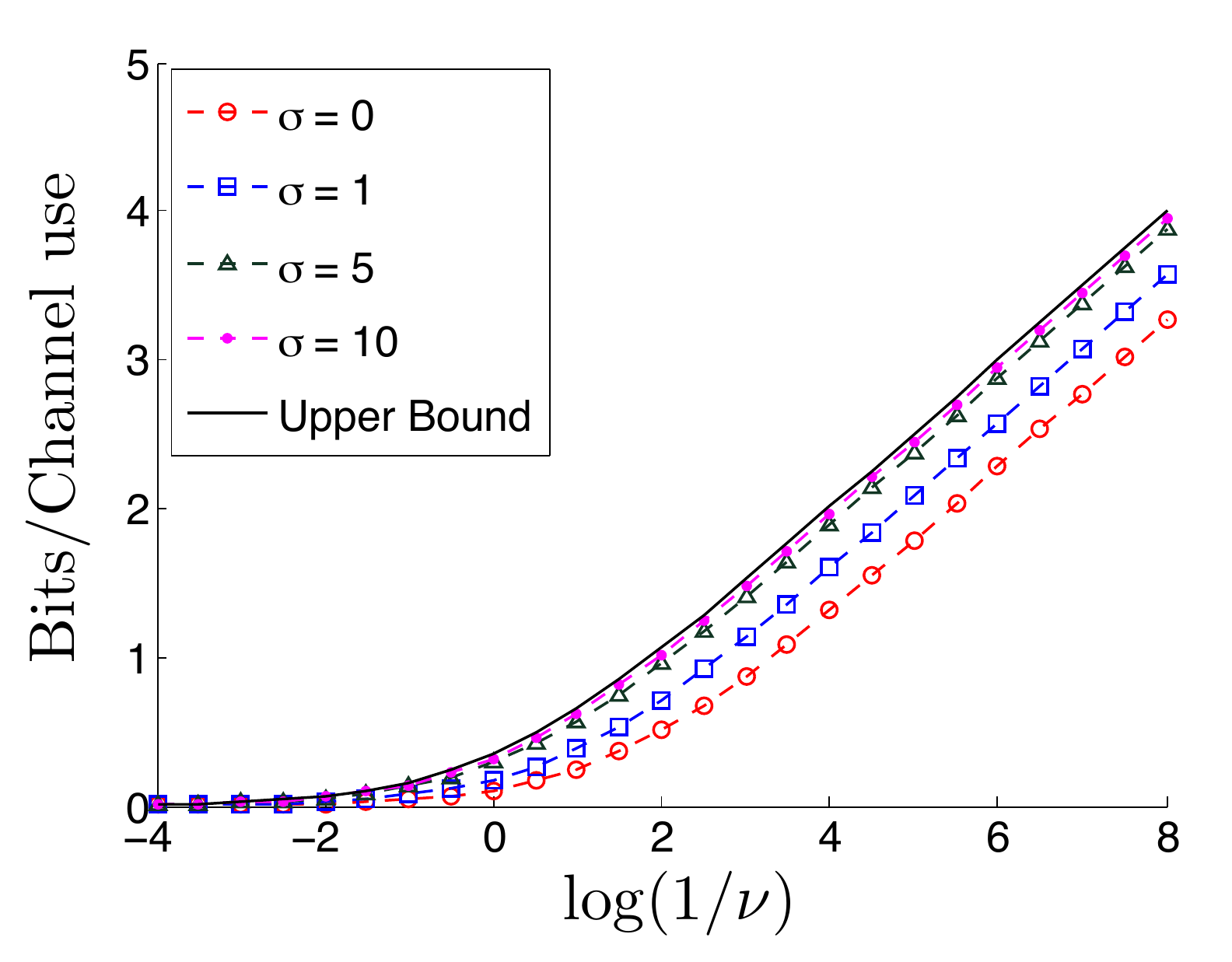}
 \caption{Capacity lower bounds for $\sigma = 0, 1, 5,$ and $10$, and the upper bound, from Theorem \ref{thm: bounds} plotted versus $\log(1/\nu)$} \label{fig: bounds}
\end{center}
\end{figure}
Note that even for relatively small values of $\sigma$, the volume based lower bound on capacity is close to the upper bound, which we recall is the channel capacity when $\sigma = \infty$. Thus, a small battery leads to significant gains in the capacity of a $(\sigma, \rho)$-power constrained AWGN channel.

\section{Upper-bounding capacity}\label{section: outer}

Theorem \ref{thm: bounds} states that $\frac{1}{2} \log \left(1+ \frac{\rho}{\nu}\right)$ upper-bounds the channel capacity. This bound is not entirely satisfactory since it is independent of the value of $\sigma$. Furthermore, Figure \ref{fig: bounds} indicates that the lower bound and the upper bound do not converge asymptotically: as $\nu \to 0$, the lower bound is $v(\sigma, \rho) - \frac{1}{2}\log 2\pi e\nu + O(\nu)$ and the upper bound is $\frac{1}{2}\log \frac{\rho}{\nu} + O(\nu)$, which differ by $O(1)$. This implies that either the upper bound, or the lower bound, or both, are loose in the low-noise regime. It is natural to expect the upper bound to be loose, since it disregards the effects of a finite $\sigma$ on capacity. To obtain some insight on the low-noise capacity, it is useful to think of coding with the $(\sigma,\rho)$-constraints as trying to fit the largest number of centers of noise balls in $\mathcal{S}_n(\sigma, \rho)$, such that the noise balls are asymptotically approximately disjoint. As the noise power $\nu$ decreases, so does the size of the noise balls, and one can imagine a very efficient packing of these small balls so that they occupy almost all the available space. The total number of balls one can pack is then roughly given by
\begin{equation}
\# \text{ of balls } \approx \frac{\text{Vol}(\cS_n(\sigma,\rho))}{\text{Vol(Noise ball)}}, 
\end{equation}
so the capacity is roughly
\begin{align}
\frac{1}{n} \log \# \text{ of balls } &= \frac{1}{n} \log \frac{\text{Vol}(\cS_n(\sigma,\rho)} {\text{Vol(Noise ball)}} \label{eq: voronoi1}\\
&\approx v(\sigma,\rho) - \frac{1}{2}\log 2\pi e \nu \label{eq: asymptote}.
\end{align}
We can make the statement in equation \eqref{eq: asymptote} rigorous, as follows:
%\begin{lemma}\label{lemma: thick sn}
%Consider a $(\sigma, \rho)$-power constrained Gaussian channel with noise power $\nu$. Let $X^n$ be a transmitted codeword satisfying the power constraints, and let $Z^n$ be the noise. Let $Y^n$ denote the received sequence. For any $\epsilon > 0$, we have
%\begin{equation}
%\lim_{n \to \infty} \mathbb{P}(Y^n \in \cS_n(\sigma, \rho) \oplus B_n(\sqrt{n(N + \epsilon)})) = 1~.
%\end{equation}
%\end{lemma}

%\begin{proof}
%By law of large numbers, $\lim_{n \to \infty} \mathbb{P}(Z^n \in B_n(\sqrt{n(N + \epsilon)})) = 1$. Also, $X^n \in \cS_n(\sigma, \rho)$ with probability $1$. The claim follows from these two facts.
%\end{proof}
 
 \begin{theorem}\label{thm: voronoi}
 Let $B_n(\sqrt{n\nu})$ be the $n$-dimensional Euclidean ball of radius $\sqrt{n\nu}$. The Minkowski sum of $\cS_n(\sigma, \rho)$ and $B_n(\sqrt{n\nu})$ is the set
\begin{equation}
\cS_n(\sigma, \rho) \oplus B_n(\sqrt{n\nu}) := \{x^n + z^n ~|~ x^n \in \cS_n(\sigma, \rho) , z^n \in B_n(\sqrt{n\nu}) \}.
\end{equation}
The capacity $C$ of an AWGN channel with a $(\sigma, \rho)$-power constraint and noise power $\nu$ satisfies
\begin{equation} \label{eq: minkowski}
C \leq \lim_{\epsilon \to 0_+} \limsup_{n \to \infty} \frac{1}{n} \log  \text{Vol}(\cS_n(\sigma, \rho) \oplus B_n(\sqrt{n(\nu+ \epsilon)}~)) - \frac{1}{2} \log 2\pi e \nu.
\end{equation}
 \end{theorem}
 
\begin{proof}
For $n \in \mathbb{N}$, let $\cF_n$ be the set of all probability distributions supported on $\cS_n(\sigma, \rho)$. From Theorem \ref{thm: capacity}, we know that the capacity $C$   is given by
\begin{equation}
C = \lim_{n \to \infty} \frac{1}{n}  \sup_{p_{X^n}(x^n) \in \cF_n} I(X^n; Y^n).
\end{equation}

Let $p_{X^n}(x^n) \in \cF_n$. Denote $\cS_n(\sigma, \rho) \oplus B_n(\sqrt{n(\nu+\epsilon)}~)$ by $C_n$, where $\epsilon > 0$, and let  $\delta_n := \mathbb P \left( Y^n \notin C_n \right).$ By the law of large numbers, we have $\delta_n \to 0$. Let $\chi$ be the indicator variable for the event $\{Y^n \in C_n\}$. Then
\begin{align}
h(Y^n) &= H(\delta_n) + \bar \delta_nh(Y^n | \chi =1) + \delta_n h(Y^n | \chi =0) \nonumber \\
&\leq H(\delta_n) + \bar \delta_n \log \text{Vol}(C_n) + \delta_n h(Y^n|\chi = 0) \label{eq: hyn_sn}.
\end{align}
%Let $\hat Y^n \sim p(Y^n | Y^ n \notin C_n)$. 
where $\bar a = 1-a$. Since $\|X^n\|^2 \leq \sigma+n\rho$ with probability $1$, we have the following bound on power of $Y^n$:
\begin{align*}
E[ \|Y^n\|^2 ] = E[ \|X^n\|^2 ] + E[ \|Z^n\|^2 ] \leq \sigma+ n\rho + n\nu.
\end{align*}
This translates to the bound 
\begin{align*}
E[ \|Y^n\|^2 \mid \chi =0] &\leq \frac{n(\rho + \nu + \sigma/n)}{\delta_n},
\end{align*}
so
\begin{align*}
h(Y^n\mid \bar\chi) &\leq \frac{n}{2} \log\frac{ 2\pi e (\rho + \nu + \sigma/n)}{\delta_n}.
\end{align*}
Substituting into inequality \eqref{eq: hyn_sn} and dividing by $n$ gives
\begin{align*}
\frac{h(Y^n)}{n}  \leq \frac{H(\delta_n)}{n} +  \bar\delta_n\frac{ \log \text{Vol}(C_n)}{n} + \frac{\delta_n}{2} \log\frac{ 2\pi e (\rho + \nu + \sigma/n)}{\delta_n}.
\end{align*}
Since this holds for any choice of $p_{X^n} \in \cF_n$, we obtain
\begin{align*}
\sup_{p_{X^n} \in \cF_n} \frac{1}{n} h(Y^n) \leq &\frac{H(\delta_n)}{n}+  \bar\delta_n\frac{ \log \text{Vol}(C_n)}{n} + \frac{\delta_n}{2} \log\frac{ 2\pi e (\rho + \nu + \sigma/n)}{\delta_n}.
\end{align*}
Taking the limsup in $n$, we arrive at
\begin{align*}
\limsup_{n \to \infty}\sup_{p_{X^n} \in \cF_n}\frac{1}{n} h(Y^n) \leq \limsup_{n \to \infty} \frac{\log \text{Vol}(C_n)}{n}
= \limsup_{n \to \infty} \frac{\log \text{Vol}(\cS_n(\sigma, \rho) \oplus B_n(\sqrt{n(\nu+\epsilon)}~))}{n}.
\end{align*}
Taking the limit as $\epsilon \to 0_+$ and noting that capacity is $\lim_{n \to \infty}\sup_{p_{X^n} \in \cF_n}\frac{1}{n} h(Y^n)  - \frac{1}{2}\log 2\pi e \nu$, we arrive at the bound in expression \eqref{eq: minkowski}.
\end{proof}
To simplify notation, define $\ell:[0, \infty) \to \mathbb R$ as
\begin{equation}\label{def: l}
\ell(\nu) :=  \limsup_{n \to \infty} \frac{1}{n} \log \text{Vol}(\cS_n(\sigma, \rho)\oplus B_n(\sqrt{n\nu}~)).
\end{equation}
We can restate the upper bound in Theorem \ref{thm: voronoi} as
\begin{equation}\label{eq: ell+epsilon}
C \leq \lim_{\epsilon \to 0_+} \ell(\nu+\epsilon) - \frac{1}{2} \log 2\pi e \nu.
\end{equation}
If $\ell$ happens to be continuous at $\nu$, we can drop the $\epsilon$ from inequality \eqref{eq: ell+epsilon} to obtain a simplified expression
\begin{align}
C &\leq  \ell(\nu) - \frac{1}{2}\log 2\pi e \nu.
\end{align}
Note that $\ell(0) = v(\sigma,\rho)$. The continuity of $\ell$ at $\nu = 0$ can be used to rigorously establish the asymptotic capacity expression in equation \eqref{eq: asymptote}. These continuity properties will be established later in this paper.

The upper bound expression involves the volume of the Minkowski sum of $\cS_n(\sigma, \rho)$ with a ball. We state here a result from convex geometry called Steiner's formula \cite{klainrota}, which gives an expression for the volume of such a Minkowski sum:
\begin{theorem}[Steiner's formula]\label{thm: steinerformula} 
Let $K_n \subset \mathbb{R}^n$ be a compact convex set and let $B_n \subset \mathbb{R}^n$ be the unit ball. Denote by $\mu_j(K_n)$ the $j$-th intrinsic volume $K_n$, and by $\epsilon_j$ the volume of $B_j$. Then for $t \geq 0$,
\begin{equation}\label{eq: steiner}
Vol(K_n \oplus tB_n) = \sum_{j=0}^{n}\mu_{n-j}(K_n)\epsilon_j t^j.
\end{equation}
\end{theorem}
Steiner's formula states that the volume of $\cS_n(\sigma, \rho) \oplus B_n(\sqrt{n\nu})$ depends not only on the volumes of these sets, but also on the intrinsic volumes of $\cS_n(\sigma, \rho)$.  Intrinsic volumes are notoriously hard to compute even for simple enough sets such as polytopes \cite{klainrota}. So it is optimistic to expect a closed form expression for the intrinsic volumes of $\cS_n(\sigma, \rho)$. Furthermore, the sets $\{\cS_n(\sigma, \rho)\}$ evolve with the dimension $n$, and to compute the volume via Steiner's formula it is necessary to keep track of how the intrinsic volumes of these sets evolve with $n$. 
\smallskip

As mentioned earlier, the case of $\sigma = 0$ is the amplitude-constrained Gaussian noise channel, the capacity of which was numerically evaluated by Smith \cite{smith1971information}. In the following section, we concentrate on evaluating the upper bound for this special case. %Using the numerical technique described by Smith, one can plot this capacity and observe that it does have the asymptote predicted in \eqref{eq: asymptote}. However, to the best of our knowledge, this has not been proved in the literature. In the following section we will prove this fact by explicitly evaluating the upper bound in Theorem \ref{thm: voronoi} via the $\ell$ function for $\sigma = 0$.

%=============================================================================%
\section{The case of $\sigma = 0$}\label{section: cube}
To simplify notation, we denote $A := \sqrt{\rho}$ in this section. We consider the scalar Gaussian noise channel with noise power $\nu$ and an input amplitude constraint of $A$. Let the capacity of this channel be $C$. Recall that the function $\ell(\nu)$ is defined as
\begin{equation}\label{eq: def ell cube}
\ell(\nu) = \limsup_{n\to\infty} \frac{1}{n} \log \text{Vol}([-A,A]^n \oplus B_n(\sqrt{n\nu})),
\end{equation}
and the upper bound on channel capacity is given by
$$C \leq \lim_{\epsilon \to 0_+} \ell(\nu+\epsilon) - \frac{1}{2}\log 2\pi e \nu.$$
The main result of this section is as follows:
\begin{theorem}\label{thm:  leconte}
The function $\ell(\nu)$ is continuous on $[0, \infty).$ For $\nu > 0$, we can explicitly compute $\ell(\nu)$ via the expression
\begin{equation}
\ell(\nu) = H(\theta^*) + (1-\theta^*)\log 2A +  \frac{\theta^*}{2} \log \frac{2\pi e\nu}{\theta^*},
\end{equation}
where $H$ is the binary entropy function, and $\theta^* \in (0,1)$ is the unique solution to $$\frac{(1-\theta^*)^2}{{\theta^*}^3} = \frac{2A^2}{\pi\nu}.$$
\end{theorem}

\begin{proof}[Proof of Theorem \ref{thm: leconte}]
The proof of Theorem \ref{thm: leconte} relies on a number of lemmas. Here we shall merely state the lemmas and defer their proofs to Appendix \ref{proofs: section: cube}.

We first prove a lemma, which makes it possible to replace $\limsup$ by $\lim$ in the expression of $\ell(\nu)$ given in equation \eqref{eq: def ell cube}.
\begin{lemma}[Proof in Appendix \ref{proof: lemma: limit cube}]\label{lemma: limit cube}
For all $\nu \geq 0$, the limit
$$\lim_{n\to\infty} \frac{1}{n} \log \text{Vol}([-A,A]^n \oplus B_n(\sqrt{n\nu}))$$  exists and is finite and equals $\ell(\nu)$, as defined in equation \eqref{eq: def ell cube}.
\end{lemma}

The special case of Steiner's formula \eqref{eq: steiner} when $K_n$ is the cube $[-A, A]^n$ and $t = \sqrt{n\nu}$ is given by
\begin{equation}\label{eq: steiner cube}
\text{Vol}([-A,A]^n \oplus B_n(\sqrt{n\nu})) = \sum_{j=0}^n {n \choose j}(2A)^{n-j}\epsilon_j(\sqrt{n\nu})^j,
\end{equation}
where $\epsilon_j$ is the volume of the $j$-dimensional unit ball. Replacing $\epsilon_j$ in equation (\ref{eq: steiner cube}), 
\begin{align}
\text{Vol}([-A,A]^n \oplus B_n(\sqrt{n\nu})) &= \sum_{j=0}^n {n \choose j}(2A)^{n-j}\frac{\pi^{j/2}}{\Gamma(j/2 + 1)}(\sqrt{n\nu})^j\\
&= \sum_{j=0}^n \frac{\Gamma(n+1)}{\Gamma(n-j+1)\Gamma(j+1)} \frac{\pi^{j/2}}{\Gamma(j/2 + 1)}(2A)^{n-j}(\sqrt{n\nu})^j.
\end{align}
Letting $\theta = \frac{j}{n}$, we rewrite the term inside the summation as
\begin{equation}
\frac{\Gamma(n+1)}{\Gamma(n(1-\theta)+1)\Gamma(n\theta+1)} \frac{\pi^{n\theta/2}}{\Gamma(n\theta/2 + 1)}(2A)^{n(1-\theta)}(\sqrt{n\nu})^{n\theta}.
\end{equation}
For $\nu > 0$, define $f^\nu_n(\theta)$ as follows:
\begin{align}
f^\nu_n(\theta) &= \frac{1}{n} \log \left( \frac{\Gamma(n+1)}{\Gamma(n(1-\theta)+1)\Gamma(n\theta+1)} \frac{\pi^{n\theta/2}}{\Gamma(n\theta/2 + 1)}(2A)^{n(1-\theta)}(\sqrt{n\nu})^{n\theta} \right)\\
&= \frac{1}{n} \log \left( \frac{\Gamma(n+1)n^{n\theta/2}}{\Gamma(n(1-\theta)+1)\Gamma(n\theta+1)\Gamma(n\theta/2 + 1)}\right) + (1-\theta)\log 2A + \theta \log \sqrt \nu + \frac{\theta}{2}\log \pi \label{eq: fntheta}.
\end{align}
Note that $f^\nu_n(\theta)$ is defined for all $n \in \mathbb{N}$, for all $\theta \in [0,1]$, and for all $\nu > 0$. Using this notation, we can rewrite the volume as
\begin{align}
\text{Vol}([-A,A]^n \oplus B_n(\sqrt{n\nu})) &= \sum_{j=0}^n e^{nf^\nu_n(j/n)}. 
\end{align}
We argue that since the volume is a sum of $n+1$ terms, the exponential growth rate of the volume is determined by the growth rate of the largest term amongst these $n+1$ terms. To be precise, we define
\begin{equation}
\hat \theta_n = \arg\max_{j/n} f^\nu_n(j/n),
\end{equation}
and  prove the following lemma:
\begin{lemma}[Proof in Appendix \ref{proof: lemma: thetan}]\label{lemma: thetan}
The limit $\lim_{n \to \infty} f^\nu_n(\hat \theta_n)$ exists and equals $\ell(\nu)$.
\end{lemma}
The next few lemmas aim to identify the limit of $f^\nu_n(\hat\theta_n)$. We first show that the functions $f^\nu_n(\cdot)$ converge uniformly to a limit function $f^\nu(\cdot)$.

\begin{lemma}[Proof in Appendix \ref{proof: lemma: concave fn}]\label{lemma: concave fn}
 The sequence of functions $\{f^\nu_n\}_{n=1}^\infty$ converges uniformly for all $\theta \in [0,1]$ to a function $f^\nu$ given by
\begin{equation}\label{eq: limit function}
f^\nu(\theta) = H(\theta) + (1-\theta)\log 2A + \frac{\theta}{2} \log \frac{2\pi e \nu}{\theta},
\end{equation}
where $H(\theta) = -\theta\log\theta - (1-\theta)\log(1-\theta)$ is the binary entropy function.
\end{lemma}

With this uniform convergence in hand, we show that the limit of $f^\nu_n(\hat \theta_n)$ can be expressed as follows: 
\begin{lemma}[Proof in Appendix \ref{proof: lemma: limfn}]\label{lemma: limfn}
We claim that  
\begin{equation}\label{eq: limfn1}
\lim_{n\to \infty} f^\nu_n(\hat\theta_n) = \max_\theta f^\nu(\theta),
\end{equation}
and therefore
\begin{equation}\label{eq: limfn2}
\ell(\nu) = \max_\theta f^\nu(\theta).
\end{equation}
\end{lemma}

We are now in a position to prove the continuity of $\ell(\nu)$. Fix a $\nu_0 > 0$, and let $\epsilon > 0$ be given. Choose a $\delta > 0$ such that for all $\nu \in (\nu_0 - \delta, \nu_0 + \delta)$,
$$||f^\nu - f^{\nu_0}||_\infty < \epsilon.$$
We can verify from equation \eqref{eq: limit function} that picking such a $\delta$ is indeed possible.
This implies
\begin{align}
|\sup_\theta f^\nu(\theta) - \sup_\theta f^{\nu_0}(\theta)| < \epsilon.
\end{align}
Using Lemma \ref{lemma: limfn}, this implies
\begin{align}
|\ell(\nu) - \ell(\nu_0)| < \epsilon,
\end{align}
which establishes continuity of $\ell$ at all points $\nu_0 > 0$.\\
 
 To show continuity at $0$, we first explicitly evaluate $\ell(\nu)$. Let $\theta^*(\nu) = \arg \max_{\theta} f^\nu(\theta).$ Using Lemma \ref{lemma: limfn}, we have $\ell(\nu) = f^\nu(\theta^*(\nu))$. Recall the expression for $f^\nu(\theta)$:
\begin{align}
f^\nu(\theta) &= H(\theta) + (1-\theta)\log 2A +\frac{\theta}{2} \log \frac{2\pi e\nu}{\theta}.
\end{align}

Differentiating $f^\nu(\theta)$ with respect to $\theta$,
\begin{align}
\frac{d}{d\theta}f^\nu(\theta) = \log \frac{1-\theta}{\theta} + \log \sqrt{\nu} - \log 2A + \frac{1}{2} \log 2\pi e - \frac{\log e}{2} - \frac{1}{2}\log \theta.
\end{align}
Setting the derivative equal to $0$ gives
\begin{align}
\log \frac{1-\theta}{\theta} + \log \sqrt{\nu} - \log 2A + \frac{1}{2} \log 2\pi e - \frac{\log e}{2} - \frac{1}{2}\log \theta = 0.
\end{align}
Simplifying this and removing the logarithms, we arrive at
\begin{align}
\frac{(1-\theta)^2}{\theta^3}  = \frac{2A^2}{\pi\nu} \label{eq: cubic}.
\end{align}
The function $\frac{(1-\theta)^2}{\theta^3}$ tends to $+\infty$ as $\theta \to 0_+$, and equals $0$ when $\theta = 1$. Thus, equation \eqref{eq: cubic} has at least one solution in the interval $(0,1)$. We can easily check that $\frac{(1-\theta)^2}{\theta^3}$ is strictly decreasing in $(0,1)$, and thus this solution must be unique. The optimal $\theta^*(\nu)$ satisfies the cubic equation \eqref{eq: cubic}, and we can see that
\begin{equation}\label{eq: lim0}
\lim_{\nu \to 0} \theta^*(\nu) = 0.
\end{equation}
Using equations \eqref{eq: cubic}  and \eqref{eq: lim0}, we have 
\begin{equation}\label{eq: lim cube}
\lim_{\nu \to 0} \frac{\nu}{\theta^*(\nu)^3} = \frac{2A^2}{\pi}.
\end{equation}
Thus,
\begin{align}
\lim_{\nu \to 0} \ell(\nu) &= \lim_{\nu \to 0} H(\theta^*(\nu)) + (1-\theta^*(\nu)) \log 2A+ \frac{\theta^*(\nu)}{2} \log \frac{2\pi e \nu}{\theta^*(\nu)}\\
&\stackrel{(a)}= \log 2A + \lim_{\nu \to 0}\frac{\theta^*(\nu)}{2} \log \frac{2\pi e \nu}{\theta^*(\nu)}\\
&\stackrel{(b)}= \log 2A\\
&= \ell(0),
\end{align}
where in $(a)$ we used equation \eqref{eq: lim0}, and in $(b)$ we used equation \eqref{eq: lim cube}. This shows that $\ell$ is continuous over $[0, \infty)$, and concludes the proof of Theorem \ref{thm: leconte}.
\end{proof}

The above bound can also be used to prove an asymptotic capacity result. We prove the following theorem:

\begin{theorem}\label{thm: ass cap cube}
The capacity $C$  of an AWGN channel with an amplitude constraint of $A$, and with noise power $\nu$, satisfies the following:
\begin{itemize}
\item[1.] 
When the noise power $\nu \to 0$, capacity $C$ is given by
$$C =  \log 2A - \frac{1}{2}\log 2\pi e\nu + O(\nu^{\frac{1}{3}}).$$
\item[2.]
When the noise power $\nu \to \infty$, capacity $C$ is given by
$$C = \frac{\alpha^2}{2} - \frac{\alpha^4}{4} + \frac{\alpha^6}{6} - \frac{5\alpha^8}{24} + O(\alpha^{10}),$$
where $\alpha = A/\sqrt\nu$.
\end{itemize}
\end{theorem}

\begin{proof}[Proof of Theorem \ref{thm: ass cap cube}]
Note that all the logarithms in this proof are assumed to be to base $e$.
\begin{itemize}
\item[1.]
Using the lower bound in Theorem \ref{thm: bounds},  
\begin{align}
C &\geq \frac{1}{2} \log \left(1 + \frac{(2A)^2}{2\pi e \nu}\right)\\
&= \log 2A - \frac{1}{2}\log 2\pi e\nu + \log\left(1 + \frac{2\pi e \nu}{(2A)^2}\right)\\
&= \log 2A - \frac{1}{2}\log 2\pi e\nu  + O(\nu). \label{eq: ass lb}
\end{align}
For the upper bound, we have
\begin{align}
\lim_{\nu \to 0} \ell(\nu) &= \log 2A + \lim_{\nu \to 0}\left[H(\theta^*)- \theta^*\log 2A +\frac{\theta^*}{2} \log \frac{2\pi e \nu}{\theta^*}\right]\\
&= \log 2A + \lim_{\nu \to 0} -(1-\theta^*)\log(1-\theta^*) + \frac{\theta^*}{2} \log \frac{\nu}{{\theta^*}^3} + \frac{\theta^*}{2} \log \frac{\pi e}{2A^2}.
\end{align}

Let $c = \left(\frac{\pi}{2A^2}\right)^{1/3}$. Using equation \eqref{eq: lim cube}, we can check that as $\nu \to 0$,
\begin{align*}
  -(1-\theta^*)\log(1-\theta^*) &= c\nu^{1/3} + o(\nu^{1/3}),\\
\frac{\theta^*}{2} \log \frac{\nu}{{\theta^*}^3} &= \frac{-3c\log c}{2}\nu^{1/3} + o(\nu^{1/3}),\\
\frac{\theta^*}{2} \log \frac{\pi e}{2A^2} &= \left(\frac{c}{2} + \frac{3c\log c}{2}\right)\nu^{1/3} + o(\nu^{1/3}).
\end{align*}
This gives the following asymptotic upper bound as $\nu \to 0$:
\begin{equation} \label{eq: ass ub}
C \leq \log 2A - \frac{1}{2}\log 2\pi e \nu +  \frac{3c}{2}\nu^{1/3} + o(\nu^{1/3}).
\end{equation}
From equations \eqref{eq: ass lb}  and \eqref{eq: ass ub}, our claim follows.
\item[2.] As noted by Smith \cite{smith1971information}, for large $\nu$ the optimal input distribution is discrete, and is supported equally on the two points $-A$ and $+A$. The output $Y$ is then distributed as
\begin{align}
Y \sim p_Y(y) = \frac{1}{2} \exp\left(\frac{(y - A)^2}{2\nu}\right) + \frac{1}{2} \exp\left(\frac{(y + A)^2}{2\nu}\right)
\end{align}
Capacity is then given by
\begin{align}
C = h(p_Y) - \frac{1}{2}\log 2\pi e \nu.
\end{align}
The entropy term $h(p_Y)$ can be manipulated as in \cite{michalowicz2008calculation} to arrive at
\begin{align}
h(p_Y) = \frac{1}{2} \log 2\pi e \nu + \alpha^2 - \frac{2}{\sqrt{(2\pi)} \alpha}e^{-\alpha^2/2}\int_0^\infty e^{-y^2/2\alpha^2}\cosh(y)\ln \cosh(y) dy,
\end{align}
where $\alpha = A/\sqrt \nu$. 
%%%%%%%
Let $$f(\alpha) = \frac{2}{\sqrt {2\pi}} \int_0^\infty e^{\frac{-y^2}{2\alpha^2}}\cosh(y)\ln\big(\cosh(y)\big) dy.$$ We consider the Taylor series expansion of $\cosh(y)\ln \big(\cosh(y)\big)$ at $y=0$, and arrive at
\begin{align}
\cosh(y)\log\big(\cosh(y)\big)=\frac{y^2}{2}+\frac{y^4}{6}+ \frac{y^6}{720} + \frac{y^8}{630}+ O\left(y^9\right).
\end{align}

Using the following definite integral expression,   
\begin{equation}
\int_0^{\infty} e^{\frac{-y^2}{2\alpha^2}}y^{2k}dy =2^{k-\frac{1}{2}} \left({\alpha ^2}\right)^{k+\frac{1}{2}} \Gamma
   \left(k+\frac{1}{2}\right),
 \end{equation}
 and substituting, we obtain 
 \begin{equation}
 f(\alpha)=\frac{\alpha ^3}{2}+\frac{\alpha ^5}{2} + \frac{\alpha ^7}{48} + \frac{\alpha^9}{6} + O(\alpha^{11}).
 \end{equation}
 Thus,  
\begin{align}
 h(p_Y) &= \frac{1}{2} \log 2\pi e \nu + \alpha^2 \nonumber\\
 &~~~ - \left(\frac{\alpha ^2}{2}+\frac{\alpha ^4}{2} + \frac{\alpha^6}{48} + \frac{\alpha^8}{6} + O(\alpha^{10})\right)\left(1-\frac{\alpha^2}{2} + \frac{\alpha^4}{8} - \frac{\alpha^6}{48} + O(\alpha^8)\right)\\
 &= \frac{1}{2} \log 2\pi e \nu + \frac{\alpha^2}{2} - \frac{\alpha^4}{4} + \frac{\alpha^6}{6} - \frac{5\alpha^8}{24} + O(\alpha^{10}).
 \end{align}
 Capacity is therefore given by
 $$C = \frac{\alpha^2}{2} - \frac{\alpha^4}{4} + \frac{\alpha^6}{6} - \frac{5\alpha^8}{24} + O(\alpha^{10}).$$
This establishes the claim. Shannon \cite{shannon} had proved that capacity at high noise for the peak power constrained (by $A^2$) AWGN channels is essentially the same as that of an average power constrained (by $A^2$) AWGN; i.e.,  
\begin{align*}
C \approx \frac{1}{2} \log (1+ A^2/\nu) &= \frac{1}{2} \log (1+ \alpha^2)\\
&= \frac{\alpha^2}{2} - \frac{\alpha^4}{4} + \frac{\alpha^6}{6} - \frac{\alpha^8}{8} + O(\alpha^{10}).
\end{align*}
It is interesting to note that the first three terms of this approximation agrees with the actual capacity.
%%%%%%%

\end{itemize}
\end{proof}

We can use Theorem \ref{thm: leconte} to numerically evaluate $\theta^*(\nu)$ and plot the corresponding upper bound from Theorem \ref{thm: voronoi}. Figure \ref{fig: cubebounds} shows the resulting plot. Note that the upper bound from Theorem \ref{thm: bounds} is not asymptotically tight in the low-noise regime, but the new upper bound is asymptotically tight.
\begin{figure}[h]
\begin{center}
\includegraphics[scale = 0.55]{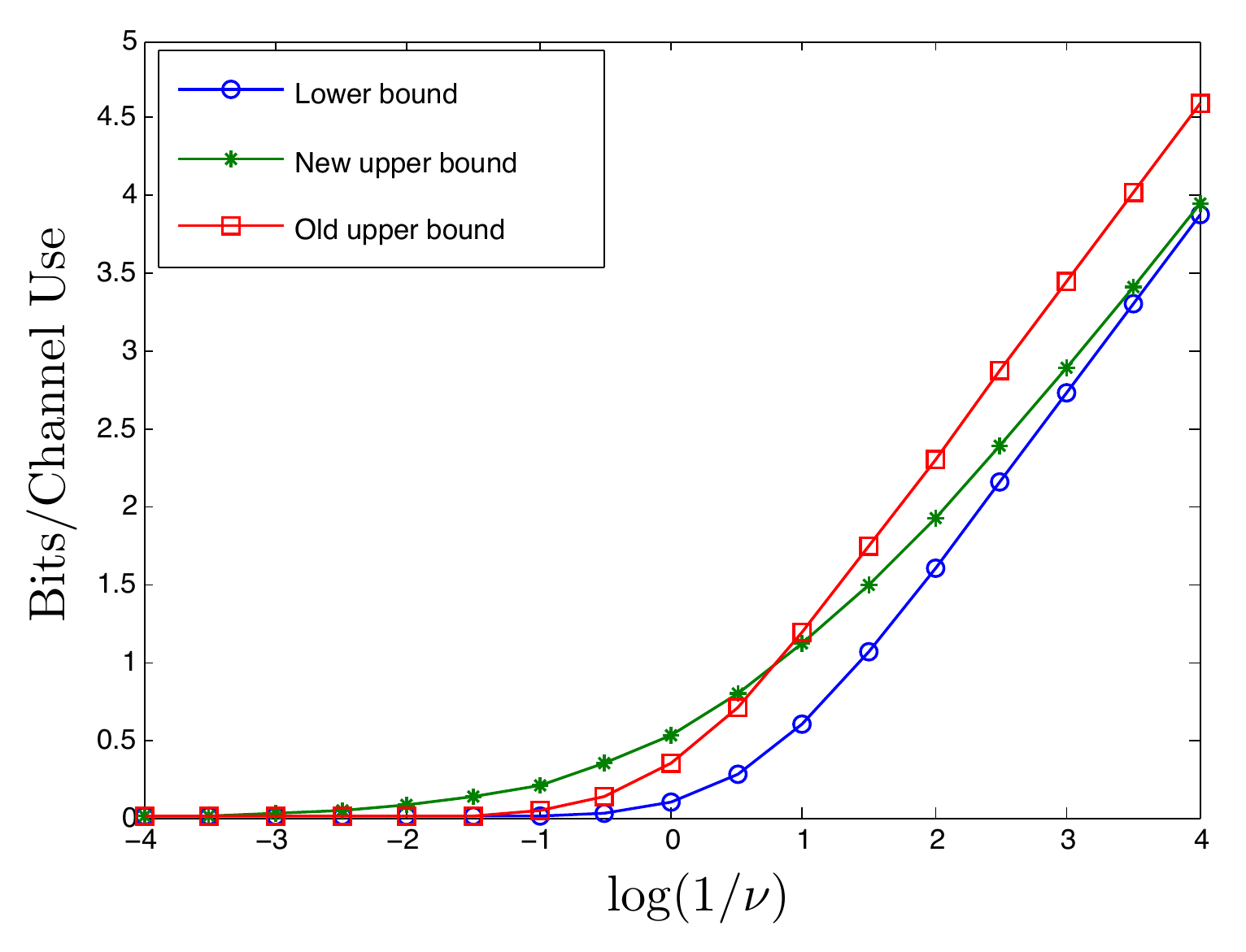}
 \caption{For the AWGN with an amplitude constraint of 1, the new upper bound and the lower bound converge asymptotically as $\nu \to 0$} \label{fig: cubebounds}
\end{center}
\end{figure}

In Theorem \ref{thm: leconte}, we essentially carried out a volume computation which answered the question: How does the volume of the Minkowski sum of a cube and a ball grow? The upper bound on capacity is then a consequence of the following facts:
\begin{enumerate}
\item
The channel capacity depends on the maximum output entropy $h(Y^n)$.
\item
The random variable $Y^n$ is (almost entirely) supported on the sum of a cube and a ball.
\item
The entropy of $Y^n$ is bounded from above by the logarithm of the volume of its (almost) support.
\end{enumerate}
Intuitively, points $2$ and $3$ should not depend on $Z$ being Gaussian, but only on $Z^n$  being almost entirely supported on $B_n(\sqrt{n\nu})$. We make this intuition precise in the following theorem: 

\begin{theorem}[Proof in Appendix \ref{proof: thm: AN bound}] \label{thm: AN bound}
Let $A, \nu \geq 0$. Let $X$ and $Z$ be random variables satisfying $|X| \leq A$ a.s.\ and $\text{Var}(Z) \leq \nu$. Then 
\begin{equation}
h(X + Z) \leq \ell(\nu),
\end{equation}
where $\ell(\nu)$ is as defined in equation \eqref{eq: def ell cube}.
\end{theorem}

 By Theorem \ref{thm: AN bound}, we can assert that the capacity $C$ of any channel with input amplitude constrained by $A$ and with an additive noise $Z$ with power at most $\nu$ is bounded from above according to
\begin{equation}\label{eq: A nu bound 1}
C =  \sup_{|X| \leq A} I(X; X+Z) \leq  \ell(\nu) -h(Z).
\end{equation}
Noting that $\text{Var(Y)} \leq A^2 + \nu$, we also have the upper bound
\begin{equation}\label{eq: A nu bound 2}
C \leq \frac{1}{2} \log 2\pi e(\nu+A^2) - h(Z).
\end{equation}
giving
\begin{equation}
C \leq \min \left(\ell(\nu) -h(Z),  \frac{1}{2} \log 2\pi e(\nu+A^2) - h(Z) \right).
\end{equation}

From Figure \ref{fig: cubebounds}, it is interesting to note that there for large values of $\nu$, the bound in inequality \eqref{eq: A nu bound 2} is better, whereas for small values of $\nu$, the bound in inequality \eqref{eq: A nu bound 1} is better. Both of these bounds are asymptotically tight as $\nu \to \infty$, but only inequality \eqref{eq: A nu bound 1} is tight for $\nu \to 0$.

%=============================================================================%
 
\section{The case of $\sigma > 0$}\label{section: sn}
 In this section, our aim is to parallel the upper-bounding technique used in Section \ref{section: cube} and obtain analogues of Theorem \ref{thm: leconte} and Theorem \ref{thm: ass cap cube}, when $\sigma$ is strictly greater than $0$. When $\sigma > 0$, the set $\cS_n(\sigma, \rho)$ is no longer an easily identifiable set like the $n$-dimensional cube from Section \ref{section: cube}. In particular, the intrinsic volumes of $\cS_n(\sigma, \rho)$ do not have a closed form expression. Despite this difficulty, we shall see that it is still possible to obtain results similar to those in Section \ref{section: cube}.

Our main result in this section is the following:
\begin{theorem}\label{thm:  sn leconte}
Define $\ell(\nu)$ as
\begin{equation}
\ell(\nu) =  \limsup_{n \to \infty} \frac{1}{n} \log \text{Vol}(\cS_n(\sigma, \rho)\oplus B_n(\sqrt{n\nu}~)).
\end{equation}
For $n \geq 1$, denote the intrinsic volumes of $\cS_n(\sigma, \rho)$ by $\mu_n(i)$ for $0 \leq i \leq n$ and define $G_n: \mathbb{R} \to \mathbb{R}$ and $g_n: \mathbb R \to \mathbb R$ as 
\begin{equation}
G_n(t) = \log \sum_{j=0}^n \mu_n(j)e^{jt}, ~~g_n(t) = \frac{G_n(t)}{n}.
\end{equation}
Define $\Lambda$ to be the pointwise limit of the sequence of functions $\{g_n\}$, which we will show exists.
Let $\Lambda^*$ be the convex conjugate of $\Lambda$. Then the following hold:
\begin{enumerate}
\item
$\ell(\nu)$ is continuous on $[0, \infty).$
\item
For $\nu > 0$, 
\begin{equation}
\ell(\nu) = \sup_{\theta \in [0,1]} \left[ -\Lambda^* (1-\theta)+ \frac{\theta}{2}\log \frac{2\pi e \nu}{\theta}\right].
\end{equation}
\end{enumerate}
 \end{theorem}

 \begin{proof}[Proof of Theorem \ref{thm: sn leconte}]
Note that for the statement of Theorem \ref{thm: sn leconte} to make sense, several results need to be established. We establish these in the Lemmas \ref{lemma: sn is convex} and \ref{lemma: gn to lambda}, where we prove the following:
\begin{lemma}[Proof in Appendix \ref{proof: lemma: sn is convex}]\label{lemma: sn is convex}
For all $n \geq 1$, the set $\cS_n(\sigma, \rho)$ is a convex set, and therefore it has well defined intrinsic volumes $\{\mu_n(i)\}_{i=0}^n$.
\end{lemma}

\begin{lemma}[Proof in Appendix \ref{proof: lemma: gn to lambda}]\label{lemma: gn to lambda}
The following results hold:
\begin{enumerate}
\item
The functions $\{g_n\}$ converge pointwise to a function $\Lambda(t): \mathbb R \to \mathbb R$ given by
 \begin{equation}
\Lambda(t) := \lim_{n \to \infty} g_n(t).
\end{equation}
\item
The convex conjugate of $\Lambda$, denoted by $\Lambda^*$, has its domain the set $[0,1]$.
\end{enumerate}
\end{lemma}

By Lemma \ref{lemma: sn is convex}, we can use Steiner's formula for the convex set $\cS_n(\sigma, \rho)$ to get
\begin{equation}\label{eq: steiner sn}
\text{Vol}(\cS_n(\sigma, \rho)\oplus B_n(\sqrt{n\nu}~)) = \sum_{j=0}^{n}\mu_n(n-j)\epsilon_j \sqrt{n\nu}^j.
\end{equation}
Define the functions $a_n(\theta)$ and $b^\nu_n(\theta)$ for $\theta \in [0,1]$ as follows. The function $a_n(\theta)$ is obtained by linearly interpolating the values of $a_n(j/n)$, where the value of $a_n(j/n)$ is given by:
\begin{equation}
a_n\left(\frac{j}{n}\right) = \frac{1}{n} \log \mu_n(n-j)  \text{~~for~~} 0\leq j \leq n.
\end{equation}
The function $b^\nu_n(\theta)$ is given by
\begin{equation}
b^\nu_n(\theta) = \frac{1}{n} \log \frac{\pi^{n\theta/2}}{\Gamma(n\theta/2 + 1)}(n\nu)^{n\theta/2}~~\text{for}~~\theta \in [0,1].
\end{equation}
Define $f^\nu_n:[0,1] \to \mathbb R$ as
\begin{equation}
f^\nu_n(\theta) := f^\nu(n, \theta) = a_n(\theta) + b^\nu_n(\theta).
\end{equation}
With this notation, we can rewrite equation \eqref{eq: steiner sn} as
\begin{align}
\text{Vol}(\cS_n(\sigma, \rho) \oplus B_n(\sqrt{n\nu})) &= \sum_{j=0}^n e^{nf^\nu(n, j/n)}. 
\end{align}
Just as in the proof of Theorem \ref{thm: leconte}, we want to establish the convergence of $f^\nu_n(\cdot)$ to some function $f^\nu(\cdot)$. Proving the convergence of $b^\nu_n(\cdot)$ is not hard, but proving the convergence of $a_n(\cdot)$ requires the application of Lemmas \ref{lemma: alex} and \ref{lemma: sn upper lower bound} given below. In Lemma \ref{lemma: alex} we establish the following:

\begin{lemma}[Proof in Appendix \ref{proof: lemma: alex}]\label{lemma: alex}
For each $n$, the following holds:
\begin{enumerate}
\item
The function $a_n(\cdot)$ is concave.
\item
The function $b^\nu_n(\cdot)$ is concave.
\item
The function $f^\nu_n(\cdot)$ is concave.
\end{enumerate}
\end{lemma}

In Lemma \ref{lemma: sn upper lower bound}, we show that the intrinsic volumes of $\{\cS_n(\sigma, \rho)\}$ satisfy a large deviations-type result, detailed below.
\begin{lemma}[Proof in Appendix \ref{proof: lemma: sn upper lower bound}]\label{lemma: sn upper lower bound}
Define a sequence of measures supported on $[0,1]$ by 
\begin{equation}
\mu_{n/n}\left(\frac{j}{n}\right) := \mu_n(j)  \text{~~for~~}  0 \leq j \leq n.
\end{equation}
The following bounds hold:
\begin{enumerate}
\item
 Let $I \subseteq \mathbb{R}$ be a closed set. The family of measures $\{\mu_{n/n}\}$ satisfies the large deviation upper bound
\begin{equation}
\limsup_{n \to \infty} \frac{1}{n} \log \mu_{n/n}(I) \leq - \inf_{x \in I} \Lambda^*(x).
\end{equation}
\item
Let $F \subseteq \mathbb{R}$ be an open set. The family of measures $\{\mu_{n/n}\}$ satisfies the large deviations lower bound
\begin{equation}
\liminf_{n \to \infty} \frac{1}{n} \log \mu_{n/n}(F) \geq - \inf_{x \in F} \Lambda^*(x).
\end{equation}
\end{enumerate}
\end{lemma}
Using the concavity and large deviations-type convergence from the two previous lemmas, we now prove the convergence of $\{f^\nu_n\}$ in the following lemma.
\begin{lemma}[Proof in Appendix \ref{proof: lemma: sn linearized}]\label{lemma: sn linearized}
The following convergence results hold:
\begin{enumerate}
\item
The sequence of functions $\{a_n\}$ converges uniformly to $-\Lambda^*(1-\theta)$ on $[0,1]$.
\item
The sequence of functions $\{b^\nu_n\}$ converges uniformly to the function $\frac{\theta}{2}\log \frac{2\pi e \nu}{\theta}$ on $[0,1]$.
\item
The sequence of functions $\{f^\nu_n\}$ converges uniformly to a function $f^\nu$ on the interval $[0,1]$, where $f^\nu$ is given by
$$f^\nu(\theta) = -\Lambda^* (1-\theta)+ \frac{\theta}{2}\log \frac{2\pi e \nu}{\theta}.$$
\end{enumerate}
\end{lemma}

We are now in a position to express $\ell(\nu)$ in terms of the limit function $f^\nu$. Let $$\hat \theta_n = \text{arg max}_{j/n} f^\nu_n(j/n).$$ 
In Lemma \ref{lemma: sn limfn} we prove the following:

\begin{lemma}[Proof in Appendix \ref{proof: lemma: sn limfn}]\label{lemma: sn limfn}
The following equality holds:  
\begin{equation}\label{eq: sn limfn1}
\lim_{n\to \infty} f^\nu_n(\hat\theta_n) = \max_\theta f^\nu(\theta).
\end{equation}
\end{lemma}

\begin{lemma}[Proof in Appendix \ref{proof: lemma: sn thetan}]\label{lemma: sn thetan}
The following equality holds:
\begin{equation}
\lim_{n\to \infty} f^\nu_n(\hat\theta_n) = \ell(\nu),
\end{equation}
and therefore
\begin{equation}
\ell(\nu) = \sup_\theta f^\nu(\theta).
\end{equation}
\end{lemma}
Part $2$ of Theorem \ref{thm: sn leconte} follows from Lemma \ref{lemma: sn thetan}. We now concentrate on proving the continuity of $\ell(\nu)$. We first show continuity at all points $\nu \neq 0$. 

Let $\nu_0 > 0$, and let $\epsilon > 0$ be given. Choose a $\delta > 0$ such that for all $\nu \in (\nu_0 - \delta, \nu_0 + \delta)$,
$$||f^\nu - f^{\nu_0}||_\infty < \epsilon.$$
This implies
\begin{align}
|\sup_\theta f^\nu(\theta) - \sup_\theta f^{\nu_0}(\theta)| < \epsilon
\implies |\ell(\nu) - \ell(\nu_0)| < \epsilon,
\end{align}
which establishes continuity of $\ell$ at all points $\nu_0 > 0$. 

Turning towards the $\nu = 0$ case, we define 
\begin{equation}
\theta^*(\nu) = \arg \max_{\theta} f^\nu(\theta).
\end{equation}
Proving the continuity of $\ell$ at $\nu = 0$ is slightly more challenging than the corresponding proof in Theorem \ref{thm: leconte} from Section \ref{section: cube}, since we do not know $\theta^*(\nu)$ explicitly in terms of $\nu$. Despite this, we can still prove the following lemma:

\begin{lemma}[Proof in Appendix \ref{proof: lemma: sn thetastar}]\label{lemma: sn thetastar}
The following equality holds:
\begin{equation}
 \limsup_{\nu \to 0} \theta^*(\nu) = 0.
\end{equation}
\end{lemma}
Now let $\nu_0 = 0$ and let $\epsilon > 0$ be given. Using continuity of $\Lambda^*$, choose an $\eta > 0$ such that 
\begin{equation}\label{eq: iv1}
|-\Lambda^*(1-\theta) - v(\sigma, \rho)| < \epsilon/2 \text{ for all } \theta \in [0, \eta).
\end{equation}
Using Lemma \ref{lemma: sn thetastar}, choose a $\delta_1$ such that 
\begin{equation}\label{eq: iv2}
\theta^*( \nu) < \eta \text{ for all } \nu \in [0, \delta_1).
\end{equation}
For all $\nu \in [0, \delta_1)$, we have
\begin{align}
\ell(\nu) &= \sup_\theta \left[-\Lambda^*(1-\theta) + \frac{\theta}{2} \log \frac{2\pi e\nu}{\theta}\right]\\
&= -\Lambda^*(\theta^*(\nu))+ \frac{\theta^*(\nu)}{2} \log \frac{2\pi e\nu}{\theta^*(\nu)}\\
&\stackrel{(a)}< v(\sigma, \rho) + \frac{\epsilon}{2} + \sup_\theta \frac{\theta}{2} \log \frac{2\pi e\nu}{\theta}\\
&\stackrel{(b)}\leq v(\sigma, \rho) + \frac{\epsilon}{2} + \pi \nu
\end{align}
where $(a)$ follows by inequalities \eqref{eq: iv1}  and \eqref{eq: iv2}, and $(b)$ follows from an evaluation of the supremum in $(a)$. Choose $\delta_2 = \frac{\epsilon}{2\pi}$, and choose $\delta = \min(\delta_1, \delta_2)$. We now have that for all $\nu \in [0, \delta)$,
\begin{equation}
\ell(\nu) < v(\sigma, \rho) + \epsilon.
\end{equation}
This combined with $\ell(\nu) \geq \ell(0) = v(\sigma, \rho)$ gives $|\ell(\nu) - \ell(0)| < \epsilon$, thus establishing continuity at $\nu_0 = 0$.
\end{proof} %proof of main theorem 
Using Theorem \ref{thm: sn leconte}, we  establish the following asymptotic capacity result:

 \begin{theorem}\label{thm: ass cap sn}
The capacity $C$ of an AWGN channel with $(\sigma, \rho)$-power constraints and noise power $\nu$ satisfies the following:
\begin{itemize}
\item[1.]
When the noise power $\nu \to 0$, capacity $C$ is given by
$$C =  v(\sigma, \rho) - \frac{1}{2}\log 2\pi e\nu + \epsilon(\nu),$$
where $\epsilon(\cdot)$ is a function such that $\lim_{\nu \to 0}\epsilon(\nu) = 0.$
\item[2.]
When noise power $\nu \to \infty$, capacity $C$ is given by
$$C = \frac{1}{2}\left(\frac{\rho}{\nu}\right)^2 - \frac{1}{4}\left(\frac{\rho}{\nu}\right)^4 +  \frac{1}{6}\left(\frac{\rho}{\nu}\right)^6 +  O\left(\left(\frac{\rho}{\nu}\right)^8\right).$$
\end{itemize}
\end{theorem}

\begin{proof}[Proof of Theorem \ref{thm: ass cap sn}]
Note that all the logarithms used in this proof are taken to be at base $e$.
\begin{itemize}
\item[1.]
Using the lower bound in Theorem \ref{thm: bounds},  
\begin{align}
C &\geq \frac{1}{2} \log \left(1 + \frac{e^{2v(\sigma, \rho)}}{2\pi e \nu}\right)\\
&= v(\sigma, \rho) - \frac{1}{2}\log 2\pi e\nu + \log\left(1 + \frac{2\pi e \nu}{e^{2v(\sigma, \rho)}}\right)\\
&= v(\sigma, \rho) - \frac{1}{2}\log 2\pi e\nu  + O(\nu). \label{eq: sn ass lb}
\end{align}
By continuity of $\ell$ at $0$, we have that as $\nu \to 0$
\begin{align}
\ell(\nu) &= v(\sigma, \rho) + \epsilon(\nu)
\end{align}
for some $\epsilon(\cdot)$ satisfying $\lim_{\nu \to 0} \epsilon(\nu) = 0$.
This gives the upper bound
\begin{equation} \label{eq: sn ass ub}
C \leq v(\sigma, \rho) - \frac{1}{2}\log 2\pi e\nu + \epsilon(\nu).
\end{equation}
Our claim follows from the inequalities \eqref{eq: sn ass lb}  and \eqref{eq: sn ass ub}.
Unlike the case of $\sigma= 0$, we are unable to give any precise rate at which $\epsilon(\nu)$ goes to $0$. Since we don't know what the intrinsic volumes of $\cS_n(\sigma, \rho)$ are, we can only say that $-\Lambda^*(1-\theta)$ is continuous at $\theta = 0$, while not knowing how fast it approaches $v(\sigma, \rho)$ as $\theta \to 0$.
\item[2.] 
Note that $C$ is bounded from below by the capacity of an AWGN channel with an amplitude constraint of $\sqrt \rho$. Using Theorem \ref{thm: ass cap cube} we obtain for $\nu \to \infty$,
\begin{equation}\label{eq: ass cap sn 1}
C \geq \frac{1}{2}\left(\frac{\rho}{\nu}\right)^2 - \frac{1}{4}\left(\frac{\rho}{\nu}\right)^4 +  \frac{1}{6}\left(\frac{\rho}{\nu}\right)^6 +  O\left(\left(\frac{\rho}{\nu}\right)^8\right).
\end{equation}
In addition, the upper bound from Theorem \ref{thm: bounds} states that
\begin{equation}\label{eq: ass cap sn 2}
C \leq \frac{1}{2}\log \left(1 + \frac{\rho}{\nu}\right) = \frac{1}{2}\left(\frac{\rho}{\nu}\right)^2 - \frac{1}{4}\left(\frac{\rho}{\nu}\right)^4 +  \frac{1}{6}\left(\frac{\rho}{\nu}\right)^6 +  O\left(\left(\frac{\rho}{\nu}\right)^8\right).
\end{equation}
The claim now follows from equations \eqref{eq: ass cap sn 1} and \eqref{eq: ass cap sn 2}.
\end{itemize}
\end{proof}

 %=============================================================================%
\section{Conclusion}
In this paper, we studied in detail an AWGN channel with a power constraint motivated by energy harvesting communication systems, called the $(\sigma, \rho)$-power constraint. Such a power constraint induces an infinite memory in the channel. In general, finding capacity expressions for channels with memory is hard, even if we allow for $n$-letter capacity expressions. However, in this particular case, we are able to exploit the following geometric properties of $\{\cS_n(\sigma, \rho)\}$:
\begin{itemize}
\item[$\mathbf{A:}$]
$\cS_{m+n}(\sigma, \rho) \subseteq \cS_n(\sigma, \rho) \times \cS_m(\sigma, \rho)$,
\item[$\mathbf{B:}$]
$[\cS_{m}(\sigma, \rho) \times \mathbf{0}_k] \times [\cS_{n}(\sigma, \rho) \times \mathbf{0}_k] \subseteq \cS_{m+n+2k} (\sigma, \rho)$, when $k = \lceil \frac{\sigma}{\rho} \rceil$.
\end{itemize}
Property $\mathbf{(A)}$ allowed us to upper-bound channel capacity, and property $\mathbf{(B)}$ allowed us to lower-bound the same. In Section \ref{section: capacity}, we used these two properties to establish an $n$-letter capacity expression.

The main contribution of Section \ref{section: volume} was the EPI based lower bound. To arrive at this lower bound, we used the $n$-letter capacity expression from Section \ref{section: capacity}, and the following property:
\begin{itemize}
\item[$\mathbf{C:}$]
The limit $\lim_{n} \frac{1}{n} \log \text{Vol}(\cS_n(\sigma, \rho))$ exists, and is finite.
\end{itemize}
For most reasonable power constraints, an exponential volume growth rate as defined in property $\mathbf{C}$ can be shown to exist. The case of $(\sigma, \rho)$-constraints was especially interesting, because it was fairly easy to evaluate $v(\sigma, \rho)$ using the numerical method in Section \ref{section: finding v}. We  attribute this ease to the existence of a state $\sigma_n$, which is a single parameter that encapsulates all the relevant information about the history of the sequence. We used the computed value of $v(\sigma, \rho)$ to plot the EPI based lower bound. Our results show that energy harvesting communication systems have significant capacity gains even for a small battery. We then established an upper bound on capacity using the exponential growth rate of volume of the Minkowski sum of $\cS_n(\sigma, \rho)$ and a ball of radius $\sqrt{n\nu}$. For the special case of $\sigma = 0$, which is the peak power constrained AWGN channel, we explicitly evaluated this upper bound. This enabled us to derive new asymptotic capacity results for such a channel. We also established a new upper bound on the entropy $h(X+Z)$, when $X$ is amplitude-constrained, and $Z$ is variance-constrained. The analysis for the case of $\sigma > 0$ was more involved because the intrinsic volumes of $\cS_n(\sigma, \rho)$ are not known in a closed form. Using a new notion of sub-convolutive sequences, we showed that the logarithms of the intrinsic volumes of $\{\cS_n(\sigma, \rho)\}$ when appropriately normalized, converge to a limit function. We then established an asymptotic capacity result in terms of this limit function. Our analysis crucially depended on both, property $\mathbf{(A)}$ and property $\mathbf{(B)}$. It would be interesting to study how our methods can be adapted to study power-constrained channels when the constraint does not satisfy one (or both) of the properties $\mathbf{(A)}$ and property $\mathbf{(B)}$, and we intend to pursue this in the future.

%=============================================================================%
\section*{Appendices}
\begin{appendix}

\section{Proofs for Section \ref{section: v}}\label{appendix: section: v}

\subsection{Proof of Lemma \ref{lemma: v1}}\label{proof: lemma: v1}
If we scale both $\sigma$ and $\rho$ by some $\alpha > 0$, by equation (\ref{eq: def}), $\cS_n(\alpha\sigma, \alpha\rho)$ is a $\sqrt{\alpha}$-scaled version of $\cS_n(\sigma, \rho)$. This means that $V_n(\alpha\sigma, \alpha\rho) = \alpha^{\frac{n}{2}} V_n(\sigma, \rho)$; i.e., $v(\alpha\sigma, \alpha\rho) = \log \sqrt{\alpha} + v(\sigma, \rho)$, which proves the lemma.

\subsection{Proof of Lemma \ref{lemma: continuity}}\label{proof: lemma: continuity}
Let $\sigma \in (0, \infty)$. Let $\epsilon > 0$ be given. We will show that there exists a $\delta ^*> 0$ such that for all $\sigma' \in (\sigma - \delta^*, \sigma + \delta^*)$,
$$|v_1(\sigma') - v_1(\sigma)| < \epsilon.$$
Since $v_1$ is a non-decreasing function, it will be enough to show that 
$$v_1(\sigma+\delta^*) - v_1(\sigma-\delta^*) < \epsilon.$$
Pick any $0 < \delta < \max( \sigma, \frac{1}{2})$. For $z \in (0,1)$, let $\cS_{n}(\sigma+\delta, 1)\times \sqrt{1-z}$ denote the set $ \cS_n(\sigma+\delta, 1)$ scaled by $\sqrt{1-z}$. Fix $z = 2\delta$. We will now show that $\cS_{n}(\sigma+\delta, 1)\times \sqrt{1-z} \subseteq S_n(\sigma-\delta, 1)$. 
\smallskip

Any $(x_1, x_2, \dots, x_n) \in \cS_n(\sigma+\delta, 1)$ satisfies
\begin{equation} \label{eq: outersn}
\sum_{i = k+1}^{l} x_i^2 \leq (l-k) + \sigma+\delta~ \text{for all } 0 \leq k < l \leq n.
\end{equation}
Let the $(\hat x_1, \dots, \hat x_n)$ be $(x_1, \dots, x_n)$ scaled by $\sqrt{1-z}$. If $(x_1, \dots, x_n)$ happens to lie in $\cS_n(\sigma-\delta, 1)$, then so does the scaled version $(\hat x_1, \dots, \hat x_n)$. If $(x_1, \dots, x_n) \in \cS_n(\sigma+\delta, 1) \setminus \cS_n(\sigma-\delta, 1)$, then for each choice of $0 \leq k < l \leq n$ such that
\begin{equation} \label{eq: badpoint}
(l-k) + \sigma + \delta \geq \sum_{i=k+1}^l x_i^2 > (l-k) + \sigma - \delta,
\end{equation}
the point $(\hat x_1, \cdots, \hat x_n)$ satisfies 
\begin{align}
\sum_{i=k+1}^l \hat x_i^2  &= \sum_{i = k+1}^l x_i^2 - z \sum_{i=k+1}^l x_i^2\\
&\leq [(l-k) + \sigma + \delta] - z[(l-k) + \sigma - \delta]\\
&\stackrel{(a)}\leq [(l-k) + \sigma + \delta] - z\\
&\stackrel{(b)}= (l-k) + \sigma - \delta,
\end{align}
where (a) follows since $l-k \geq 1$ and $\sigma -\delta > 0$, implying that $(l-k) + \sigma - \delta \geq 1$, and (b) follows by the choice $z = 2\delta$. Thus, the point $(\hat x_1, \dots, \hat x_n)$ lies in the set $\cS_n(\sigma-\delta, 1)$. The containment
\begin{equation}
\sqrt{1 - 2\delta} \times \cS_n(\sigma + \delta, 1) \subseteq \cS_n(\sigma-\delta, 1) \subseteq \cS_n(\sigma+\delta, 1)
\end{equation}
gives
\begin{equation}
\frac{1}{2} \log (1 - 2\delta) + v_1(\sigma+\delta)  \leq v_1(\sigma-\delta) \leq v_1(\sigma+\delta).
\end{equation}
Hence, we have
\begin{equation}
v_1(\sigma + \delta) - v_1(\sigma - \delta) \leq -\frac{1}{2} \log (1 - 2\delta).
\end{equation}
Picking $\delta^*$ small enough to satisfy 
$$-\frac{1}{2} \log (1 - 2\delta^*) < \epsilon,$$
we establish continuity of $v_1(\sigma)$ in the open set $(0, \infty)$.\\

Now consider the case when $\sigma = 0$.  We will show that there exists a $\delta ^*> 0$ such that for all $\sigma' \in [0, \delta^*)$,
$$|v_1(\sigma') - v_1(0)| < \epsilon.$$
Since $v_1$ is a non-decreasing function, it will be enough to show that 
$$v_1(\delta^*) - v_1(0) < \epsilon.$$
Pick any $\delta < 1$. Using the same strategy as before, we can show that $\cS_{n}(\delta, 1)\times \sqrt{1-\delta} \subseteq S_n(0, 1)$.  This gives
\begin{equation}
v_1(\delta) + \frac{1}{2}\log (1-\delta) \leq v_1(0),
\end{equation}
and thus
\begin{equation}
0 \leq v_1(\delta) - v_1(0) \leq -\frac{1}{2} \log(1-\delta).
\end{equation}
Choosing $\delta^*$ small enough such that $-\frac{1}{2} \log(1-\delta) < \epsilon$, we establish continuity at $\sigma = 0$.

\subsection{Proof of Lemma \ref{lemma: concavity of v1}}\label{proof: lemma: concavity of v1}
For every $n$, define the function $V_n(\sigma) = \frac{\log \text{Vol}(\cS_n(\sigma, 1)}{n}$. We'll first show that $V_n(\sigma)$ is concave. Define the set ${\mathscr S}_{n+1} \subseteq \mathbb{R}^{n+1}$ as follows:
\begin{equation}
{\mathscr S}_{n+1} = \{(x_1, \dots, x_n, \sigma_x) | (x_1, \dots, x_n) \in \cS_n(\sigma_x, 1)\}.
\end{equation}
We claim that ${\mathscr S}_{n+1}$ is convex. Let $\mathbf x = (x_1, \dots, x_n, \sigma_x)$ and $\mathbf y = (y_1, \dots, y_n, \sigma_y)$ be in ${\mathscr S}_{n+1}$. For $\lambda \in [0, 1]$, consider the point $\lambda \mathbf x + (1-\lambda) \mathbf y$. 
For any $0 \leq k < l \leq n$, we have
\begin{align}
\sum_{i=k+1}^l (\lambda x_i + (1-\lambda)y_i)^2 &= \lambda^2 \sum_{i=k+1}^l x_i^2 + (1-\lambda)^2\sum_{i=k+1}^l y_i^2 + 2\lambda(1-\lambda)\sum_{i=k+1}^l x_iy_i\\
&\leq \lambda^2 \sum_{i=k+1}^l x_i^2 + (1-\lambda)^2 \sum_{i=k+1}^l y_i^2 + \lambda(1-\lambda)\sum_{i=k+1}^l (x_i^2 + y_i^2)\\
&= \lambda \sum_{i=k+1}^l x_i^2 + (1-\lambda)\sum_{i=k+1}^l y_i^2 \\
&\leq (\lambda \sigma_x + (1-\lambda)\sigma_y) + (l-k).
\end{align}
Thus, $\lambda \mathbf x + (1-\lambda) \mathbf y \in \mathscr S_{n+1}$, which proves that $\mathscr S_{n+1}$ is a convex set.

Now the $n$-dimensional volume of the intersection of ${\mathscr S}_{n+1}$ with the hyperplane $\sigma_x =  \sigma$ is simply the volume of $\cS_n(\sigma, 1)$. Using the Brunn-Minkowski inequality \cite{schneider2013convex}, we see that $\text{Vol}(\cS_n(\sigma, 1))^\frac{1}{n}$ is concave in $\sigma$, so the logarithm is also concave. This establishes the concavity of $V_n(\sigma)$.

To show that $v_1(\sigma)$ is concave, we simply note that it is the pointwise limit of the sequence of concave functions $\{V_n\}$.

\subsection{Proof of Lemma \ref{lemma: Ansigman}}\label{proof: lemma: Ansigman}
For $x^n \in \cA_n(\sigma(n))$, the state at time $n$ is nonnegative. Suppose that after time $n$, we impose a restriction that the power used per symbol cannot be more than $\frac{1}{2}$. This means that the battery will charge by at least $\frac{1}{2}$ at each timestep, and after $2\sigma(n)$ steps, the battery will be fully charged to $\sigma(n)$. Denote the set of all such $(n+2\sigma(n))$-length sequences obtained by this process as $\hat{\cA}_n(\sigma(n))$. This set is contained in $\cS_{n + 2\sigma(n)}(\sigma(n),1)$, and its volume is
$$\mbox{Vol}(\cA_n(\sigma(n))) \times (\sqrt{2})^{2\sigma(n)}.$$
The key point is to note the containment
$$\hat{\cA}_n(\sigma(n)) \times \cdots \times \hat{\cA}_n(\sigma(n)) \subset \cS_{m(n + 2\sigma(n))}(\sigma(n),1)~,$$
for all $m \geq 1$, where there are $m$ copies in the product on the left hand side.
This holds because we ensure that the battery is fully charged to $\sigma(n)$ after each $(n+2\sigma(n))$-length block.
Taking the limit in $m$ and using Lemma \ref{lemma: limit}, we see that
\begin{align*}
v_1(\sigma(n)) &\geq \frac{1}{n + 2\sigma(n)} \log \left(\mbox{Vol}(\cA_n(\sigma(n))) \times \sqrt{2}^{2\sigma(n)}\right).
\end{align*}
Letting  $n$ tend to infinity and using conditions (\ref{eq: cond1}) and (\ref{eq: cond2}), we arrive at
$$\liminf_{n \to \infty} v_1(\sigma(n)) \geq  \frac{1}{2} \log 2\pi e,$$
which proves the claim.

\subsection{Proof of Lemma \ref{lemma: sigma_infty}}\label{proof: lemma: sigma_infty}
The key to proving Lemma \ref{lemma: sigma_infty} is to examine the distribution of the burstiness $\sigma(X^n)$, when $X^n$ is drawn from a uniform distribution on $\cA_n$. Since a high-dimensional Gaussian closely approximates the uniform distribution on $\cA_n$, it makes sense to look at the burstiness of $X^n$ when each $X_i$ is drawn independently from a standard normal distribution.
\smallskip

Let $X_1, X_2, \dots, X_n$ be i.i.d.\ standard normal random variables. Let $Y_i = X_i^2-1$, for $1\leq i \leq n$. These $Y_i$ are i.i.d.\ with zero mean and variance $2$. Define $S_0 = 0$ and
$$S_m = \sum_{i=1}^m Y_i, \text{ for }1\leq m \leq n.$$
Define $\Sigma_n$, the burstiness of the sequence of $X_i$, by
\begin{equation*}
\Sigma_n = \max_{0 \leq k < l \leq n} Y_{k+1} + Y_{k+2} + ...+ Y_l = \max_{0 \leq k < l \leq n} (S_l - S_k).
\end{equation*}
The following inequality holds:
\begin{equation}
\Sigma_n \leq \max_{0\leq l \leq n} S_l - \min_{0 \leq k \leq n} S_k := \tilde{\Sigma}_n.
\end{equation}
Fix some $\alpha > 0$. Then
\begin{align}
\liminf_n P(\Sigma_n \leq &\alpha \sqrt{n}) \geq \liminf_n P(\tilde{\Sigma}_n \leq \alpha\sqrt{n}) \nonumber \\
&= \liminf_n P( \max \frac{S_l}{\sqrt{2n}} - \min \frac{S_k}{\sqrt{2n}} \leq \frac{\alpha}{\sqrt{2}}) \nonumber \\
&\stackrel{(a)}{=} P(\max_{0 \leq t \leq 1} B_t - \min_{0 \leq t \leq 1} B_t \leq \frac{\alpha}{\sqrt{2}}) \label{eq: donsker}\\
&\geq P(\sup_{0 \leq t \leq 1} |B_t| \leq \frac{\alpha}{2\sqrt{2}}), \nonumber 
\end{align}
where in equation (\ref{eq: donsker}), $B_t$ is the standard Brownian motion and the equality in step (a) follows from Donsker's theorem  \cite{durrett2010probability}. We now choose $\alpha$ large enough so that
\begin{equation}
P(\sup_{0 \leq t \leq 1} |B_t| \leq \frac{\alpha}{2\sqrt{2}}) \geq \frac{3}{4}.
\end{equation}
Since $\lim_n P(X^n \in \cA_n) = 1/2$ by the central limit theorem for 
$Y_1, Y_2, \dots$,  we have
\begin{equation}
\liminf_n P(\cA_n(\alpha \sqrt n)) = \liminf_n P(\Sigma_n \leq \alpha \sqrt{n}, X^n \in \cA_n) \geq \frac{1}{4}~, \nonumber
\end{equation}
where $P(\cA_n(\alpha \sqrt n)) := P(X^n \in \cA_n(\alpha \sqrt n))$.
The volume of $\cA_n(\alpha \sqrt n))$ is upper-bounded by the volume of $\cA_n$. Furthermore, it is lower-bounded by the volume of a ball $\cB_n$ centered at the origin, such that $P(X^n \in \cB_n) =: P(\cB_n) = P(\cA_n(\alpha \sqrt n))$, since a Gaussian distribution decays radially. Using standard concentration bounds on the normal distribution \cite{hopcroftfoundations}, to satisfy $P(\cB_n) \geq1/4$, the radius of $\cB_n$  must be $\sqrt{n} + o(\sqrt{n})$. Thus,
$\lim_n \frac{\text{~Vol}(\cB_n)}{n} = \frac{1}{2} \log 2\pi e.$
 Using $\text{Vol}(\cB_n) \leq \text{Vol}(\cA_n(\alpha \sqrt n)) \leq \text{Vol}(\cA_n)$, we obtain $\lim_n \frac{\text{Vol}(\cA_n(\alpha\sqrt{n}))}{n} = \frac{1}{2}\log 2 \pi e.$
Thus, with $\sigma(n) = \alpha\sqrt n$, the pair $\cA_n(\sigma(n))$ and $\sigma(n)$ satisfy both the conditions in Lemma \ref{lemma: Ansigman}, thereby proving Lemma \ref{lemma: sigma_infty}.

\section{Appendix for Section \ref{section: finding v}} \label{appendix: numerical}
Let $0 < \gamma < 1$. We define a new set $\cS_{n, \gamma}(\sigma, 1)$ to be
\begin{equation}
\cS_{n, \gamma}(\sigma, 1) = \{x^n \in \cS_n(\sigma, \rho)~|~ x_i^2 \geq \gamma \text{ for every } 1 \leq i \leq n\}.
\end{equation}
Using Fekete's Lemma, it is easy to establish that following limit exists:
\begin{equation}
\lim_{n \to \infty} \frac{\text{Vol}(\cS_{n, \gamma}(\sigma, 1))}{n} := v_{1, \gamma}(\sigma).
\end{equation}
Clearly, $v_{1, \gamma}(\sigma) \leq v_1(\sigma)$ as $\cS_{n, \gamma}(\sigma, 1) \subseteq \cS_n(\sigma, 1)$. In Lemma \ref{lemma: use gamma}, we show that it is possible to choose a small enough value of $\gamma$ such that $v_{1,\gamma}(\sigma)$  approximates $v_1(\sigma)$ as closely as desired.

\begin{lemma}\label{lemma: use gamma}
We have $$v_1\left(\frac{\sigma}{1 - \eta}\right) + \frac{1}{2} \log (1 - \eta) \leq v_{1, \gamma}(\sigma) \leq v_1(\sigma),$$
where $\eta = \gamma  + 2(\sqrt{\sigma+1}\sqrt\gamma)$.
\end{lemma}

\begin{proof}
Clearly, $v_{1, \gamma}(\sigma) \leq v_1(\sigma)$, since $\cS_{n, \gamma}(\sigma, 1) \subseteq \cS_n(\sigma, 1)$. 
\smallskip

Now let $x^n \in \cS_n(\sigma, 1-\eta) \cap \mathbb{R}_+^n$. We claim that $$(x_1+\sqrt\gamma, \dots, x_n+\sqrt{\gamma}) \in \cS_{n, \gamma}(\sigma, 1).$$ This would imply that a translated version of $\cS_n(\sigma, 1-\eta) \cap \mathbb R_+^n$ lies inside $\cS_{n,\gamma}(\sigma, 1) \cap \mathbb R_+^n$, which will give us a lower bound on the volume of the latter in terms of the former.
Since each $(x_i+\sqrt \gamma)^2 \geq \gamma$, the only condition we need to check is whether $(x_1+\sqrt\gamma, \dots, x_n+\sqrt{\gamma}) \in \cS_n(\sigma, 1)$. For any $0 \leq k < l \leq n$, we have
\begin{align}
\sum_{i = k+1}^l (x_i + \sqrt{\gamma})^2 &= \sum_{i = k+1}^l x_i^2 + 2\sqrt{\gamma}\sum_{i = k+1}^l x_i + (l-k)\gamma\\
&\leq (l-k)(1 - \eta) + \sigma +  2\sqrt{\gamma}\sum_{i = k+1}^l x_i + (l-k)\gamma\\
&\leq (l-k)(1 - \eta + \gamma) + \sigma + 2\sqrt{\gamma}\sum_{i = k+1}^l \sqrt{\sigma+1}\\
&\leq (l-k)(1 - \eta + \gamma + 2\sqrt \gamma \sqrt{\sigma+1})+\sigma.\\
&= (l-k) + \sigma
\end{align}
This gives us
\begin{align}
\text{Vol}(\cS_n(\sigma, 1 - \eta) \leq \text{Vol}(\cS_{n, \gamma}(\sigma, 1)) \leq \text{Vol}(\cS_n(\sigma, 1)),\\
\end{align}
implying that
\begin{align}
v_1\left(\frac{\sigma}{1 - \eta}\right) + \frac{1}{2} \log (1 - \eta) \leq v_{1, \gamma}(\sigma) \leq v_1(\sigma).
\end{align}
By the continuity of $v_1(\sigma)$,  we see that choosing a $\gamma$ (and consequently an $\eta$) small enough will give a value of $v_{1, \gamma}(\sigma)$ that is as close as desired to $v_1(\sigma)$.
\end{proof}

Lemma \ref{lemma: use gamma} ensures that a numerical method which can closely approximate $v_{1,\gamma}(\sigma)$ can also be used to closely approximate $v_1(\sigma)$ for small values of $\gamma$. Henceforth, we focus our attention on calculating $v_{1,\gamma}(\sigma)$. As noted in Section \ref{section: finding v}, we exploit the idea of battery state. Given $(x_1, \dots, x_n) \in \cS_{n,\gamma}(\sigma, 1)$, define
\begin{equation}\label{eq: phi1}
\phi_n = 
\begin{cases}
\sigma_n &\text{ if } \sigma_n < \sigma,\\
\sigma_{n-1} + 1 - x_n^2 &\text{ if } \sigma_n = \sigma.\\
\end{cases}
\end{equation}
Setting $\phi_0 = \sigma$, equation \eqref{eq: phi1} can also be written as
\begin{equation}\label{eq: phi2}
\phi_{n} = 
\begin{cases}
\phi_{n-1} + 1 - x_n^2 &\text{ if } \phi_{n-1} < \sigma,\\
\sigma + 1 - x_n^2 &\text{ if } \phi_{n-1} \geq \sigma.\\
\end{cases}
\end{equation}

Define the function $\Phi_n: \cS_{n,\gamma}(\sigma, 1) \to \mathbb R$ such that $\Phi_n(x_1, \dots, x_n) = \phi_n$. Let $\lambda_n$ be the Lebesgue measure restricted to $\cS_{n,\gamma}(\sigma, 1)$. Let $\nu_n$ be the measure induced by $\Phi_n$ on $\mathbb R$. In Lemma \ref{lemma: absolute} below, we show the following:

\begin{lemma}\label{lemma: absolute}
The measure $\nu_n$ is absolutely continuous with respect to the Lebesgue measure on $\mathbb R$.
\end{lemma}

\begin{proof}
We first calculate $\nu_1$.  Define 
$$F_1(\phi_1) = \nu_1((-\infty, \phi_1]).$$
We have the relation  $\phi_1 = \sigma + 1 - x_1^2$, where $x_1$ has the Lebesgue measure on $\cS_{1, \gamma}(\sigma, 1)$: $[-\sqrt{\sigma+1}, -\sqrt \gamma] \cup [\sqrt{\sigma+1}, \sqrt \gamma]$. It is easy to see that
\begin{equation}
F_1(\phi) = 
\begin{cases}
0 &\text{ for } \phi < 0 \\
2(\sqrt{\sigma+1} - \sqrt{\sigma+1-\phi}) &\text{ for } 0 \leq \phi \leq \sigma + 1 - \gamma\\
2(\sqrt{\sigma + 1} - \sqrt \gamma) &\text{ for } \sigma + 1 - \gamma < \phi.
\end{cases}
\end{equation}
Observe that $F_1$, being Lipshitz, is an absolutely continuous function. This implies that the measure $\nu_1$ is absolutely continuous with respect to the Lebesgue measure on $\mathbb R$ and possesses a Radon-Nikodym derivative $f_1$, which equals the derivative of $F_1$ almost everywhere. We set $f_1$ as follows:
\begin{equation}
f_1(\phi) = 
\begin{cases}
0 &\text{ for } \phi < 0 \\
\frac{1}{\sqrt{\sigma+1-\phi}} &\text{ for } 0 \leq \phi \leq \sigma + 1 - \gamma\\
0 &\text{ for } \sigma + 1 - \gamma < \phi.
\end{cases}
\end{equation}
We note that $f_1$ is continuous and bounded on the closed interval $[0, \sigma+1-\gamma]$. Our proof now proceeds by induction. We assume that the measure $\nu_n$ admits a density $f_n$, which is continuous and bounded on the closed interval $[0, \sigma+1-\gamma]$, and prove that $\nu_{n+1}$ has a density $f_{n+1}$ which is continuous and bounded on $[0, \sigma+1-\gamma]$.  
\smallskip

Define $$F_{n+1}(\phi) = \nu_{n+1}((-\infty, \phi]).$$ Since $\nu_n$ is supported on $[0, \sigma+1-\gamma]$, we can use the expression in \eqref{eq: phi2} to conclude the same about $\nu_{n+1}$ and restrict our attention to $0 \leq \phi \leq \sigma+1-\gamma$. For $\phi$ in this range, we use relation \eqref{eq: phi2} and express $F_{n+1}$ in terms of $f_n$ as follows:
\begin{align}
F_{n+1}(\phi) =& \int_{x = 0}^{\phi - 1 + \gamma} \int_{t = \gamma}^{x+1}\frac{f_n(x)}{\sqrt t} dt dx \nonumber\\
 +& \int_{x = \phi+1-\gamma}^{\sigma} \int_{t = x - (\phi - 1)}^{ x+1}\frac{f_n(x)}{\sqrt t} dt dx \nonumber\\
 +& \int_{x = \sigma}^{\sigma+1-\gamma} \int_{t = \sigma - (\phi -1)}^{\sigma+1}\frac{f_n(x)}{\sqrt t} dt dx\\
 =& \int_{x=0}^{\phi-1+\gamma} 2f_n(x)[\sqrt{x+1} - \sqrt{\gamma}]dx \nonumber\\
+& \int_{x=\phi-1+\gamma}^\sigma 2f_n(x)[\sqrt{x+1}-\sqrt{x-(\phi-1)}]dx \nonumber\\
+& \int_{x = \sigma}^{\sigma+1-\gamma} 2f_n(x)[\sqrt{\sigma+1} - \sqrt{\sigma-(\phi-1)}]dx.
\end{align}
From the induction assumption of continuity and boundedness of $f_n$, it is easy to check that $F_{n+1}$ is Lipshitz and therefore absolutely continuous. This implies that $\nu_{n+1}$ permits a density, which is equal to the derivative of $F_n$ almost everywhere. We can evaluate this density by differentiating $F_{n+1}$ with respect to $\phi$. This involves differentiating under the integral sign, and the conditions for doing so are seen to be satisfied because of the continuity and boundedness of $f_n$ and the square root function. We then get
\begin{align}\label{eq: fn1}
f_{n+1}(\phi) = \int_{\phi-1+\gamma}^\sigma \frac{f_n(x)}{\sqrt{x - (\phi - 1)}} dx + \int_\sigma^{\sigma+1 -\gamma} \frac{f_n(x)}{\sqrt{\sigma - (\phi -1)}} dx,
\end{align}
which is supported on, and is bounded and continuous on, the interval $[0, \sigma+1-\gamma]$.
\end{proof}

Equation \eqref{eq: fn1} in the proof of Lemma \ref{lemma: absolute} describes the evolution of $f_n$ as the dimension $n$ increases. Let $C([0, \sigma+1-\gamma]$ be the set of continuous functions defined on the interval $[0, \sigma+1-\gamma]$.  Define the integral operator $A: C([0, \sigma+1-\gamma] \to C([0, \sigma+1-\gamma]$ as follows:
\begin{equation}
A(x,t) = 
\begin{cases}
\frac{1}{\sqrt{x+1-t}} &\text{ if } 0 \leq x < \sigma, ~ 0 \leq t \leq  x+1 -\gamma~,\\ 
 \frac{1}{\sqrt{\sigma+1-t}} &\text{ if } \sigma \leq x \leq \sigma+1-\gamma,~ 0 \leq t \leq \sigma+1-\gamma~,\\
 0 &\text{ otherwise. }
\end{cases} 
\end{equation}
We can express equation \eqref{eq: fn1} in another form,
\begin{equation}\label{eq: fn2}
f_{n+1}(t) = \int A(x,t)f_n(x)dx.
\end{equation}
We denote this $f_{n+1} = A(f_n)$. Iterating this relation, we obtain
\begin{equation}\label{eq: fn3}
f_{n+1} = A^n f_1.
\end{equation}
We make three crucial observations. Firstly, the kernel $A$ is bounded and piecewise continuous with the discontinuities confined to a single curve $t = x+1-\gamma$. It is also immediate that the spectral radius of $A$, defined by
$$r(A) = \sup_{||f||=1} ||Af||$$ 
is such that $r(A) > 0$. We use Theorem 2.13 from Anselone \cite{anselone1971} to obtain that such an operator $A$ is compact. In addition, we can apply the Krein Rutman theorem  from Schaefer \cite{schaefer1999} to establish that $r(A)$ is an eigenvalue with a positive eigenvector $u \in C([0, \sigma+1-\gamma] \setminus 0$.

Secondly, we have $$\nu_n([0, \sigma+1-\gamma]) = \int_{x = 0}^{\sigma+1-\gamma} f_n(x) dx = \text{Vol}(\cS_{n,\gamma}(\sigma, 1).$$
Thus we have
\begin{align}
v_{1,\gamma}(\sigma) &= \lim_{n\to\infty} \frac{1}{n} \log \text{Vol}(\cS_{n,\gamma}(\sigma, 1))\\
&= \lim_{n\to\infty} \frac{1}{n} \log \int_{x = 0}^{\sigma+1-\gamma} f_n(x) dx\\
&= \lim_{n\to\infty} \frac{1}{n} \log \int_{x = 0}^{\sigma+1-\gamma}A^{n-1}f_1(x) dx\\
&\stackrel{(a)}= r(A)
\end{align}
where $(a)$ follows because the projection of $f_1$ in the direction of $u$ is nonzero owing to the positivity of both these functions.

Thirdly, define a sequence of operators $\{A_n\}$ as discrete approximations of $A$ as follows. Let $h_n = \frac{\sigma+1-\gamma}{n}$, 
$$A_nf(t) = \sum_{j=0}^n A(jh_n, t)f(jh_n)h_n.$$
Using Theorem 2.13 from Anselone \cite{anselone1971} once more, we conclude that the sequence of operators $\{A_n\}$ is collectively compact and that $||A_n|| \to ||A||$. We can now use existing numerical techniques to find $r(A_n)$, which will provide an approximation to $r(A)$. The spectral radius $r(A)$ equals $v_{1, \gamma}(\sigma)$, which closely approximates $v_1(\sigma)$, and validates the numerical procedure as described in Section \ref{section: finding v}.

%=============================================================================%
\section{Proofs for Section \ref{section: cube}}\label{proofs: section: cube}

\subsection{Proof of Lemma \ref{lemma: limit cube}}\label{proof: lemma: limit cube}
Denote $A_n = [-A,A]^n$, and $B_n = B_n(\sqrt{n\nu})$. Let $C_n =A_n \oplus B_n$. Note that for any $m,n \geq 1$
\begin{align}
B_n(\sqrt{n\nu}) \times B_m(\sqrt{m\nu}) &\subseteq B_{m+n}(\sqrt{(m+n)\nu})\\
[-A,A]^n \times [-A,A]^m &= [-A,A]^{m+n}.
\end{align}
It follows that 
\begin{align}
C_m \times C_n &= (A_m \oplus B_m) \times (A_n \oplus B_n)\\
&=(A_m \times A_n) \oplus (B_m \times B_n)\\
&\subseteq A_{m+n} \oplus B_{m+n}\\
&= C_{m+n}.
\end{align}
This implies
\begin{equation}
\text{Vol}(C_{m+n}) \geq \text{Vol}(C_m)\text{Vol}(C_m),
\end{equation}
which immediately implies existence of the limit $\lim_{n \to \infty} \frac{1}{n} \log \text{Vol}(C_n)$, which equals $\ell(\nu)$ as defined in equation \eqref{eq: def ell cube}. To show this limit is finite, we note that $A_n \subseteq B_n(\sqrt{nA^2})$. Thus $C_n \subseteq B_n(\sqrt{n}(\sqrt{\nu}+A))$, which gives $$\ell(\nu) \leq \frac{1}{2}\log 2\pi e (\sqrt \nu + A)^2 < \infty.$$

\subsection{Proof of Lemma \ref{lemma: thetan}}\label{proof: lemma: thetan}
 We have the trivial bounds
\begin{align}
&\frac{\text{Vol}([-A,A]^n \oplus B_n(\sqrt{n\nu}))}{n+1} \leq e^{nf^\nu_n(\hat \theta_n)} \leq  \text{Vol}([-A,A]^n \oplus B_n(\sqrt{n\nu})),
\end{align}
which implies
\begin{align}
&\frac{1}{n} \log \text{Vol}([-A,A]^n \oplus B_n(\sqrt{n\nu})) - \frac{\log (n+1)}{n} \leq {f^\nu_n(\hat \theta_n)} \leq \frac{1}{n} \log \text{Vol}([-A,A]^n \oplus B_n(\sqrt{n\nu})).
\end{align}
Taking the limit in $n$ and using Lemma \ref{lemma: limit cube} we see that
\begin{equation}
\lim_{n \to \infty} f^\nu_n(\hat \theta_n) = \ell(\nu).
\end{equation}

\subsection{Proof of Lemma \ref{lemma: concave fn}}\label{proof: lemma: concave fn}
We first prove pointwise convergence.
Looking at equation (\ref{eq: fntheta}), we see that all we need to prove is that for all $\theta \in [0,1]$,
\begin{equation}\label{eq: limit for fntheta}
\lim_{n \to \infty} \frac{1}{n} \log \left( \frac{\Gamma(n+1)n^{n\theta/2}}{\Gamma(n(1-\theta)+1)\Gamma(n\theta+1)\Gamma(n\theta/2 + 1)}\right) = H(\theta) + \frac{\theta}{2} \log 2e - \frac{\theta}{2} \log \theta.
\end{equation}

For $\theta = 0$, we can easily check the validity of this statement. Let $\theta > 0$. We use the approximation $$\log \Gamma(z) = z\log z - z + \log \frac{z}{2\pi} + o(z).$$ 
\begin{align}
&\frac{1}{n} \log \left( \frac{\Gamma(n+1)n^{n\theta/2}}{\Gamma(n(1-\theta)+1)\Gamma(n\theta+1)\Gamma(n\theta/2 + 1)}\right) = \nonumber \\
&\frac{1}{n}\bigg((n+1)\log \frac{n+1}{e} + \frac{n\theta}{2} \log n - (n(1-\theta)+1)\log \frac{n(1-\theta)+1}{e} - (n\theta+1)\log \frac{n\theta+1}{e} \nonumber \\
&- (n\theta/2+1) \log \frac{n\theta/2 + 1}{e} + o(n)\bigg)~.
\end{align}
Using $(x+1)\log (x+1) = x \log x + o(x)$, we can simplify the above to get
\begin{align}
&\frac{1}{n}\left( n\log n + \frac{n\theta}{2} \log n - n\bar\theta \log n\bar\theta - n\theta \log n\theta - (n\theta/2) \log n\theta/2e +  o(n) \right)~,\\
&= \frac{1}{n} \left(n H(\theta) - (n\theta/2)\log (\theta/2e) + o(n) \right)~.
\end{align}
Taking the limit as $n \to \infty$, we establish equality (\ref{eq: limit for fntheta}). 

To show uniform convergence, we first observe that the functions $f^\nu_n(\cdot)$ are concave. This concavity is immediately evident from the log-convexity of the $\Gamma$ function and from equation (\ref{eq: fntheta}). Therefore, $\{f^\nu_n\}$ are concave functions converging pointwise to a continuous functions $f^\nu$ on $[0,1]$. Uniform convergence now follows from Lemma \ref{lemma: convex}.

\subsection{Proof of Lemma \ref{lemma: limfn}}\label{proof: lemma: limfn}
 By Lemma \ref{lemma: concave fn}, the sequence of functions $\{f^\nu_n\}$ converges to $f^\nu$ uniformly. This uniform convergence implies that the family of functions $\{f^\nu_n\}$ is equicontinuous \cite{royden2011} (Section 10.1, Theorem 3, pg. 209). Let $\epsilon > 0$. Choose $N$ large such that $|f^\nu_n(x) - f^\nu_n(y)| < \epsilon/2$ if $|x-y| < 1/N$. This implies that for all $n > N$, 
\begin{equation}
\max_\theta f^\nu_n(\theta) \geq f^\nu_n(\hat \theta_n) > \max_\theta f^\nu_n(\theta) - \epsilon/2.
\end{equation}

Using the uniform convergence of $\{f^\nu_n\}$, we choose $M$ large enough such that $||f^\nu-f^\nu_n||_\infty < \epsilon/2$ for all $n > M$. Let $L = \max(M,N)$. For all $n > L$, we have
$$\max_\theta f^\nu(\theta) + \epsilon/2 > \max_\theta f^\nu_n(\theta) \geq f^\nu_n(\hat\theta_n) \geq  \max_\theta f^\nu_n(\theta) - \epsilon/2 \geq \max_\theta f^\nu(\theta) - \epsilon,$$
and thus
$$|f^\nu_n(\hat \theta _n) - \max_\theta f^\nu(\theta)| < \epsilon.$$
This concludes the proof of equation \eqref{eq: limfn1}. By Lemma \ref{lemma: thetan}, we immediately have the equality \eqref{eq: limfn2}.

\subsection{Proof of Theorem \ref{thm: AN bound}}\label{proof: thm: AN bound}

Let $\epsilon > 0$. Let $\{X_i\}_{i=1}^n$ and $\{Z_i\}_{i=1}^n$ be $n$ i.i.d copies of $X$ and $Z$ respectively. Let $\delta_n$ be given by 
\begin{equation}
\delta_n := \mathbb P \left( X^n + Z^n \notin [-A,A]^n \oplus B_n(\sqrt{n(\nu+\epsilon)}~)\right)
\end{equation}
Denote $$C_n := [-A,A]^n \oplus B_n(\sqrt{n(\nu+\epsilon)}~).$$
By the law of large numbers, the probability $\delta_n \to 0$.\\

Let $Y := X+Z$. We have
\begin{align}
nh(Y) = h(Y^n) &= H(\delta_n) + (1-\delta_n)h(Y^n | Y^n \in C_n) + \delta_n h(Y^n | Y^n \notin C_n)\\
&\leq H(\delta_n) + (1- \delta_n) \log \text{Vol}(C_n) + \delta_n h(Y^n | Y^n\notin C_n) \label{eq: hyn}.
\end{align}
Let $\hat Y^n \sim p(Y^n | Y^ n \notin C_n)$. We have following bound on $Y^n$
\begin{align}
E[ ||Y^n||^2 ] \leq n(\nu + A^2).
\end{align}
This translates to a bound on $\hat Y^n$
\begin{align}
E[ || \hat Y^n ||^2 ] &\leq \frac{n(\nu + A^2)}{\delta_n},
\end{align}
which implies
\begin{align}
h(\hat Y^n) &\leq \frac{n}{2} \log\frac{ 2\pi e (\nu+A^2)}{\delta_n}.
\end{align}
Substituting in inequality \eqref{eq: hyn},
\begin{align}
h(Y^n) \leq H(\delta_n) + (1- \delta_n) \log \text{Vol}(C_n) + \delta_n \frac{n}{2} \log\frac{ 2\pi e (\nu+A^2)}{\delta_n}
\end{align}
which implies
\begin{align}
h(Y) \leq \frac{H(\delta_n)}{n} +  (1- \delta_n)\frac{ \log \text{Vol}(C_n)}{n} + \frac{\delta_n}{2} \log\frac{ 2\pi e (\nu+A^2)}{\delta_n}.
\end{align}
Taking the limit in $n$, we get
\begin{align}
h(Y) \leq \ell(\nu + \epsilon).
\end{align}
As this holds for any choice of $\epsilon$, we let $\epsilon$ tend to $0$ and use the continuity from Theorem \ref{thm: leconte} to arrive at
\begin{equation}
h(Y) \leq \ell(\nu).
\end{equation}

%===================================================================================%
\section{Proofs for Section \ref{section: sn}}
\subsection{Proof of Lemma \ref{lemma: sn is convex}}\label{proof: lemma: sn is convex}
Let $x^n, y^n \in \cS_n(\sigma, \rho)$ and let $z^n = \lambda x^n + (1 - \lambda) y^n$. By Jensen's inequality we have for every $1 \leq i \leq n$,
$$z_i^2 \leq \lambda x_i^2 + (1-\lambda)y_i^2.$$
Since both $x^n$ and $y^n$ both satisfy (\ref{eq: def}), the above inequality gives us that $z^n$ does so too; i.e., $z^n \in \cS_n(\sigma, \rho)$. 

\subsection{Proof of Lemma \ref{lemma: gn to lambda}}\label{proof: lemma: gn to lambda}
The sets $\{\cS_n(\sigma, \rho)\}$ satisfy the containment
\begin{equation}
\cS_{m+n} \subseteq \cS_m \times \cS_n  \text{ ~for every~} m,n \geq 1.
\end{equation}
This implies that the family of intrinsic volumes $\{\mu_n(\cdot)\}_{n\geq 1}$, is \emph{sub-convolutive}; i.e., it satisfies the following condition:
\begin{equation}\label{eq: sn define subc1}
\mu_m \star \mu_n \geq \mu_{m+n} \text{ for every } m,n \geq 1.
\end{equation}
Noting that $\mu_n(n)$ is the volume of $\cS_n(\sigma, \rho)$, and $\mu_n(0) = 1$ for all $\cS_n$, we can check that the sequence $\{\mu_n(\cdot)\}$ satisfies the assumptions $\mathbf{(A), (B)}$ and $\mathbf{(C)}$ detailed in Appendix \ref{appendix: subc}; namely,
\begin{align*}
&\mathbf{(A):} ~\alpha := \lim_{n \to \infty} \frac{1}{n} \log \mu_n(n) \text{~~~ is finite.}\\
&\mathbf{(B):} ~\beta := \lim_{n \to \infty} \frac{1}{n} \log \mu_n(0) \text{~~~ is finite.}\\
&\mathbf{(C):} \text{~For all } n,~ \mu_n(n) >0, \mu_n(0)>0.
\end{align*}
Lemma \ref{lemma: gn to lambda} then follows from the results in Appendix \ref{appendix: subc}, in particular Lemma \ref{lemma: lambda sub}.

\subsection{Proof of Lemma \ref{lemma: alex}}\label{proof: lemma: alex}
Note that the claims in points $1$ and $2$ immediately imply $3$, since $f^\nu_n = a_n + b^\nu_n$. 

We shall prove $2$ first. The expression for $b^\nu_n(\theta)$ is given by
\begin{equation}
b^\nu_n(\theta) = \frac{1}{n} \log \frac{\pi^{n\theta/2}}{\Gamma(n\theta/2 + 1)}(n\nu)^{n\theta/2}~~\text{for}~~\theta \in [0,1].
\end{equation}
Since the Gamma function is log-convex \cite{boyd2009convex} (Exercise 3.52), we see that $b^\nu_n(\cdot)$ is a concave function. 

To show $1$, note that all we need to prove is that
\begin{equation}
a_n\left( \frac{j}{n} \right) \geq \frac{ a_n\left(\frac{j-1}{n} \right) + a_n\left( \frac{j+1}{n}\right) }{2} \text{~~for all~~} 1 \leq j \leq n-1,
\end{equation}
as $a_n$ is a linear interpolation of the values at $\frac{j}{n}$. This is equivalent to proving
\begin{equation}
\mu_n(j)^2 \geq \mu_n(j-1)\mu_n(j+1) \text{~~for all~~} 1 \leq j \leq n-1.
\end{equation}
This is an easy application of the Alexandrov-Fenchel inequalities for mixed volumes. For a proof we refer to McMullen \cite{mcmullen1991inequalities}, where in fact the author obtains 
$$\mu_n(j)^2 \geq \frac{j+1}{j}\mu_n(j-1)\mu_n(j+1).$$

\subsection{Proof of Lemma \ref{lemma: sn upper lower bound}}\label{proof: lemma: sn upper lower bound}
As noted in Appendix \ref{proof: lemma: gn to lambda}, the family of intrinsic volumes $\{\mu_n(\cdot)\}_{n\geq 1}$, is sub-convolutive and it satisfies the assumptions $\mathbf{(A)}, \mathbf{(B)}$, and $\mathbf{(C)}$ detailed in Appendix \ref{appendix: subc}. Part $1$ of Lemma \ref{lemma: sn upper lower bound} is now an immediate consequence of Theorem \ref{thm: subc upper bound}.

To prove part $2$, let $F \subseteq \mathbb{R}$ be an open set. We assume that $F \cap [0,1]$ is nonempty, since the otherwise the result is trivial. We will construct a new sequence of functions $\{\hat\mu_n\}$ such that $\mu_n \geq \hat \mu_n$ for all $n$; i.e., $\mu_n$ pointwise dominates $\hat \mu_n$ for all $n$. The large deviations lower bound for the sequence  $\{\hat\mu_n\}$ will then serve as a large deviations lower bound for the sequence $\{\mu_n\}$.

For notational convenience, we write $\cS_n$ for $\cS_n(\sigma, \rho)$ in this proof. Fix an $a \geq 1$. Let $\gamma = \lceil \frac{\sigma}{\rho} \rceil$. Let 
$$\hat \cS_{a+\gamma} = \{x^{a+\gamma} \in \mathbb{R}^{a+\gamma} | x^a \in \cS_a, x_{a+1}^{a+\gamma} = \mathbf{0} \}.$$
For all $k \geq 0$, the $k^\text{th}$ intrinsic volume of a convex body is independent of the ambient dimension \cite{klainrota}. Thus, for $0 \leq k \leq a$, the $k^\text{th}$ intrinsic volume of $\hat \cS_{a+\gamma}$ is exactly  the same as that of $\cS_a$. For $ a+1 \leq k \leq a+\gamma$, the $k^\text{th}$ intrinsic volume of $\hat \cS_{a+\gamma}$ equals $0$. The sequence of intrinsic volumes of $\hat S_{a+\gamma}$ may therefore be considered to be simply $\mu_a$. In addition, note that for all $m \geq 1$,
$$\underbrace{\hat \cS_{a+\gamma} \times \cdots \times \hat \cS_{a+\gamma}}_{m} \subseteq \cS_{m(a+\gamma)},$$
which implies 
$$\underbrace{\mu_a \star \cdots \star \mu_a}_{m} \leq \mu_{m(a+\gamma)}.$$
This leads us to define the new sequence $\hat \mu_n$ as
$$\hat \mu_n = \underbrace{\mu_a \star \cdots \star \mu_a}_{\lfloor \frac{n}{a+\gamma}\rfloor} := \mu_a^{\star \lfloor \frac{n}{a+\gamma}\rfloor}.$$
Clearly $\hat \mu_n \leq \mu_n$. Define $\hat G_n(t)$ as follows,
$$\hat G_n(t) = \log \sum_{j=0}^n \hat\mu_n(j)e^{jt},$$
and consider the limit
\begin{align}\label{eq: sn gnhat}
\lim_{n \to \infty} \frac{1}{n} \hat G_n(t) &= \lim_{n \to \infty} \frac{1}{n}\lfloor \frac{n}{a+\gamma}\rfloor G_a(t)\\
&= \frac{G_a(t)}{a+\gamma}.
\end{align}
Applying  the G\"{a}rtner-Ellis theorem, stated in  Theorem \ref{thm: gartner ellis}, for $\{\hat \mu_n\}$ and noting that $\frac{G_a(t)}{a+\gamma}$ is differentiable, we get the lower bound  
\begin{align*}
&\liminf_{n \to \infty} \frac{1}{n} \log \hat\mu_{n/n}(F) \geq -\inf_{x \in F} \left( \frac{G_a(t)}{a+\gamma}\right)^*(x),\\
\end{align*}
which implies 
\begin{align*} &\liminf_{n \to \infty} \frac{1}{n} \log \mu_{n/n}(F) \geq -\inf_{x \in F} \frac{a}{a+\gamma}g_a^*\left(\frac{a+\gamma}{a}x \right).
\end{align*}
We claim that $\inf_{x \in F}\frac{a}{a+\gamma}g_a^*\left(\frac{a+\gamma}{a}x \right)$ converges to $\inf_{x \in F} \Lambda^*(x)$. Let $\epsilon > 0$. We can rewrite the infimum as
\begin{align*}
\inf_{x \in F}\frac{a}{a+\gamma}g_a^*\left(\frac{a+\gamma}{a}x \right) = \inf_{y \in \frac{a+\gamma}{a} F} \frac{a}{a+\gamma}g_a^*(y). 
\end{align*}
Using Theorem \ref{theorem: subc gnstar}, we know that $\{g_n^*\}$ converges uniformly $\Lambda^*$ over $[0,1]$. By the converse of the Arzela-Ascoli theorem, we have that $g_n^*$ are uniformly bounded and equicontinuous. Let $\delta > 0$ be such that 
\begin{align}
|\Lambda^*(x) - \Lambda^*(y)| < \epsilon/3 \text{ whenever } |x-y| < \delta.
\end{align}
Let $M$ be a uniform bound on $|g_n^*(\cdot)|$. Choose $A_0$ such that for all $a > A_0$, 
\begin{align}\label{eq: A0}
\frac{\gamma}{a+\gamma}M < \epsilon/3.
\end{align}
Choose $A_1$ such that for all $a > A_1$, 
\begin{align}\label{eq: A1}
||g_a^*  - \Lambda^*||_{\infty} < \epsilon/3.
\end{align}
Choose $A_2$ such that for all $a>A_2$,
\begin{align}\label{eq: A2}
\frac{\gamma}{a+\gamma} < \delta.
\end{align}
Choose $A_3$ such that for all $a >A_3$,
\begin{align}\label{eq: A3}
\frac{a+\gamma}{a}F \cap [0,1] \neq \phi.
\end{align}
Now for all $a > \max(A_0, A_1, A_2, A_3)$, 
\begin{align*}
\Bigg|\inf_{y \in \frac{a+\gamma}{a} F} \frac{a}{a+\gamma}g_a^*(y) - \inf_{y \in F} \Lambda^*(y) \Bigg| &\leq \Bigg| \inf_{y \in \frac{a+\gamma}{a} F} \frac{a}{a+\gamma}g_a^*(y)- \inf_{y \in \frac{a+\gamma}{a} F}g_a^*(y)\Bigg|\\
 &+ \Bigg|\inf_{y \in \frac{a+\gamma}{a} F}g_a^*(y) - \inf_{y \in \frac{a+\gamma}{a} F} \Lambda^*(y) \Bigg|\\
 &+ \Bigg| \inf_{y \in \frac{a+\gamma}{a} F} \Lambda^*(y) - \inf_{y \in  F} \Lambda^*(y)\Bigg|\\
 &\stackrel{(a)}< \epsilon/3 + \epsilon/3 + \epsilon/3\\
 &= \epsilon.
\end{align*}
By the relation (\ref{eq: A3}), all the infimums involved in the above sequence of inequalities are finite. In step $(a)$,  the first term is less that $\epsilon/3$ by inequality (\ref{eq: A0}), the second term is less that $\epsilon/3$ by inequality (\ref{eq: A1}), and the last term is less that $\epsilon/3$ by inequality (\ref{eq: A2}). This completes the proof of part $2$ of Lemma \ref{lemma: sn upper lower bound}, and thus completes the proof of Lemma \ref{lemma: sn upper lower bound}.

\subsection{Proof of Lemma \ref{lemma: sn linearized}}\label{proof: lemma: sn linearized}
Note that the claims in points $1$ and $2$ immediately imply $3$, since $f^\nu_n = a_n + b^\nu_n$. 

We'll first prove the claim in point $2$. We start by proving pointwise convergence of $\{b^\nu_n(\cdot)\}$. Recall the expression for $b^\nu_n(\theta)$,
$$b^\nu_n(\theta) = \frac{1}{n} \log \frac{\pi^{n\theta/2}}{\Gamma(n\theta/2 + 1)}(n\nu)^{n\theta/2}.$$
For $\theta = 0$, this convergence is obvious. Let $\theta > 0$. We use the approximation $$\log \Gamma(z) = z\log z - z + O(\log z),$$ 
and get that
\begin{align}
b^\nu_n(n\theta) &= \frac{1}{n}\left[ \frac{n\theta}{2}\log\pi n \nu - \frac{n\theta}{2}\log\frac{n\theta}{2e} + O(\log n\theta)\right]\\
&= \frac{1}{n} \left[ \frac{n\theta}{2} \log \frac{2\pi e \nu}{\theta} + O(\log n\theta)\right]\\
&= \frac{\theta}{2} \log \frac{2\pi e \nu}{\theta} + \frac{O(\log n\theta)}{n}.
\end{align}
Taking the limit as $n \to \infty$, the pointwise convergence of $b^\nu_n$ follows. Concavity of $b^\nu_n$ from point $2$ of Lemma \ref{lemma: alex}, combined with Lemma \ref{lemma: convex} then implies uniform convergence. \\

We shall now prove point $1$. We start by showing the pointwise convergence of $a_n(\theta)$ to $-\Lambda^*(1-\theta)$, or equivalently the convergence of $a_n(1-\theta)$ to $-\Lambda^*(\theta)$. Note that convergence at the boundary points is already known. Let $\theta_0 \in (0,1)$. For ease of notation, we denote 
\begin{align*}
\chi(\theta) &:= -\Lambda^*(\theta)\\
\bar a_n(\theta) &:= a_n(1-\theta).
\end{align*}
Note that $\bar a_n$ is linearly interpolated from its values at $j/n$, where $\bar a_n(j/n) = \frac{1}{n} \log \mu_n(j)$. 
Let $\epsilon > 0$ be given. The function $\chi$, being continuous on the bounded interval $[0,1]$, is uniformly continuous. Choose $\delta >0$ such that
$$|\chi(x) - \chi(y)| < \epsilon, \text{~~whenever~~} |x-y| < \delta.$$
Choose $N_0 > 1/(\delta/3)$, and divide the interval $[0,1]$ into the the $N_0$ intervals $I_j := \left[\frac{j}{N_0}, \frac{j+1}{N_0}\right]$ for $0 \leq j \leq N_0-1$. Note that each interval has length less than $\delta/3$. Without loss of generality, let $\theta_0$ lie in the interior of the $k$-th interval (we can always choose a different value of $N_0$ to make sure $\theta_0$ does not lie on the boundary of any interval). Thus,
$$\frac{k-1}{N_0} < \theta_0 < \frac{k}{N_0}.$$
Lemma \ref{lemma: sn upper lower bound} along with the continuity of $\chi$ imply that
\begin{equation}
\lim_{n \to \infty} \frac{1}{n} \log \mu_n(I_j) = \sup_{\theta \in I_j} \chi(\theta).
\end{equation}
For $n > 2/\min \left( \theta_0 - \frac{k-1}{N_0}, \frac{k}{N_0} - \theta_0\right)$, there exists an $i$ such that
\begin{equation}
\frac{k-1}{N_0} < \frac{i}{n} < \theta_0 < \frac{i+1}{n} < \frac{k}{N_0}.
\end{equation}
Thus for some $\lambda > 0$, we can write
\begin{equation}
\bar a_n(\theta_0) = \lambda\frac{1}{n}\log \mu_{n/n}(i/n) + (1-\lambda)\frac{1}{n} \log\mu_{n/n}((i+1)/n),
\end{equation}
and obtain the inequality 
\begin{align}
\bar a_n(\theta_0) &= \lambda\frac{1}{n}\log \mu_{n/n}(i/n) + (1-\lambda)\frac{1}{n} \log\mu_{n/n}((i+1)/n)\\
&\leq \max\left(\frac{1}{n}\log \mu_{n/n}(i/n), \frac{1}{n} \log\mu_{n/n}((i+1)/n)\right)\\
&\leq \frac{1}{n} \log \mu_{n/n}(I_k).
\end{align}
Thus we have the upper bound
\begin{align}\label{eq: an ub}
\limsup_n \bar a_n(\theta_0) &\leq \lim_{n\to \infty} \frac{1}{n} \log \mu_{n/n}(I_k)\\
 &= \sup_{\theta \in I_k} \chi(\theta)\\
 &\stackrel{(a)}\leq \chi(\theta_0) + \epsilon
\end{align}\\
where $(a)$ follows from the choice of $N_0$ and uniform continuity of $\chi$.\\

Define $$\hat \theta_n(j) = \arg \sup_{\frac{i}{n} \text{~s.t.~} \frac{i}{n} \in I_j} \mu_{n/n}\left(\frac{i}{n}\right).$$ As 
$$\mu_{n/n}(\hat \theta_n(j)) \leq \mu_{n/n}(I_j) \leq \left(\frac{n}{N_0}+2\right)\mu_{n/n}(\hat \theta_n(j)) \leq n\mu_{n/n}(\hat \theta_n(j)),$$ it is easy to see that 
\begin{align}
\lim_{n\to\infty}\frac{1}{n} \log \mu_{n/n}(\hat \theta_n(j)) &= \lim_{n \to \infty} \frac{1}{n} \log \mu_{n/n}(I_j)\\
&= \sup_{\theta \in I_j} \chi(\theta).
\end{align}
Note that 
$$\sup_{\theta \in I_j} \bar a_n(\theta) \geq \frac{1}{n} \log \mu_{n/n}(\hat\theta_n(j)).$$
This implies that for the intervals $I_{k-1}$ and $I_{k+1}$,
\begin{align}
\liminf_{n\to \infty} \left[\sup_{\theta \in I_{k-1}} \bar a_n(\theta)\right] \geq \sup_{\theta \in I_{k-1}} \chi(\theta) \geq \chi(\theta_0) - \epsilon\\
\liminf_{n\to \infty} \left[\sup_{\theta \in I_{k+1}}\bar a_n(\theta)\right] \geq \sup_{\theta \in I_{k+1}} \chi(\theta) \geq \chi(\theta_0) - \epsilon.
\end{align}
Since $\bar a_n(\theta)$ is concave, this implies 
\begin{align}
\bar a_n(\theta_0) &\geq \min (\sup_{\theta \in I_{k-1}} \bar a_n(\theta), \sup_{\theta \in I_{k+1}} \bar a_n(\theta)).
\end{align}
Taking the $\liminf$ on both sides,
\begin{align}\label{eq: an lb}
\liminf_{n \to \infty} \bar a_n(\theta_0) \geq \chi(\theta_0) - \epsilon.
\end{align}
Inequalities \eqref{eq: an ub} and \eqref{eq: an lb} prove the pointwise convergence of $\bar a_n(\theta_0)$ to $\chi(\theta_0)$.
Concavity of $a_n$ from point $2$ of Lemma \ref{lemma: alex}, combined with Lemma \ref{lemma: convex} then implies uniform convergence. 

\subsection{Proof of Lemma \ref{lemma: sn limfn}}\label{proof: lemma: sn limfn}
By Lemma  \ref{lemma: sn linearized}, the sequence of functions $\{f^\nu_n\}$ converges to $f^\nu$ uniformly. Using the converse of the Arzela-Ascoli theorem, this implies that the family of functions $\{f^\nu_n\}$ is equicontinuous. Let $\epsilon > 0$ be given. Choose $N$ large such that $|f^\nu_n(x) - f^\nu_n(y)| < \epsilon/2$ if $|x-y| < 1/N$. This implies that for all $n > N$, 
\begin{equation}
\max_\theta f^\nu_n(\theta) \geq f^\nu_n(\hat \theta_n) > \max_\theta f^\nu_n(\theta) - \epsilon/2.
\end{equation}

Using the uniform convergence of $\{f^\nu_n\}$, we choose $M$ large enough such that $||f^\nu-f^\nu_n||_\infty < \epsilon/2$ for all $n > M$. Let $L = \max(M,N)$. For all $n > L$, we have
$$\max_\theta f^\nu(\theta) + \epsilon/2 > \max_\theta f^\nu_n(\theta) \geq f^\nu_n(\hat\theta_n) \geq  \max_\theta f^\nu_n(\theta) - \epsilon/2 \geq \max_\theta f^\nu(\theta) - \epsilon,$$
and thus
$$|f^\nu_n(\hat \theta _n) - \max_\theta f^\nu(\theta)| < \epsilon.$$
This concludes the proof.

\subsection{Proof of Lemma \ref{lemma: sn thetan}}\label{proof: lemma: sn thetan}
Recall that
$$\text{Vol}(\cS_n(\sigma, \rho)^n \oplus B_n(\sqrt{n\nu})) = \sum_{j=0}^n e^{nf^\nu_n(j/n)}.$$
We have the trivial bounds
\begin{align}
e^{nf^\nu_n(\hat\theta_n)} \leq \text{Vol}(\cS_n(\sigma, \rho) \oplus B_n(\sqrt{n\nu})) \leq (n+1)e^{nf^\nu_n(\hat\theta_n)}
\end{align}
implying
\begin{align}
 f^\nu_n(\hat\theta_n) \leq \frac{1}{n} \log\text{Vol}(\cS_n(\sigma, \rho) \oplus B_n(\sqrt{n\nu})) \leq \frac{\log (n+1)}{n} + f^\nu_n(\hat\theta_n).
\end{align}
Taking the limit in $n$, we obtain
\begin{equation}
\lim_{n\to \infty} f^\nu_n(\hat\theta_n) = \ell(\nu).
\end{equation} 
An application of Lemma \ref{lemma: sn limfn} gives
\begin{equation}
\ell(\nu) = \sup_\theta f^\nu(\theta).
\end{equation}

\subsection{Proof of Lemma \ref{lemma: sn thetastar}}\label{proof: lemma: sn thetastar}

Recall the expression of $f^\nu(\theta)$:
\begin{align}
f^\nu(\theta) &= -\Lambda^* (1-\theta)+ \frac{\theta}{2}\log \frac{2\pi e \nu}{\theta}.
\end{align}
Suppose $\limsup_{\nu \to 0} \theta^*(\nu) = \eta > 0$. Choose a sequence $\{\nu_n\}$ such that 
\begin{align}
\lim_{n \to \infty} \nu_n &= 0\\
\theta^*(\nu_n) &> \frac{\eta}{2} \text{ for all } n \geq 1.
\end{align}
We have that for all $\nu > 0$, 
$$\ell(\nu) = \sup_\theta f^\nu(\theta) \geq f^\nu(0) = -\Lambda^*(1) = v(\sigma, \rho).$$
Thus,
\begin{align}
v(\sigma, \rho) &\leq \ell(\nu_n)\\
&= f^{\nu_n}(\theta^*(\nu_n))\\
&=  -\Lambda^*(\theta^*(\nu_n))+ \frac{\theta^*(\nu_n)}{2} \log \frac{2\pi e\nu_n}{\theta^*(\nu_n)}\\
&\leq \sup_{\theta}\left[-\Lambda^*(1-\theta) + \frac{\theta}{2} \log \frac{2\pi e}{\theta}\right] + \frac{\theta^*(\nu_n)}{2}\log \nu_n\\
&\stackrel{(a)} \leq C + \frac{\theta^*(\nu_n)}{2}\log \nu_n\\
&\stackrel{(b)}\leq C + \frac{\eta}{4} \log \nu_n
\end{align}
where in $(a)$, $C$ is a constant and in $(b)$ we assume $\log \nu_n < 0$. Taking the limit as $n \to \infty$, we get that
\begin{equation}
v(\sigma, \rho) \leq \lim_{n \to \infty} C + \frac{\eta}{4} \log \nu_n = -\infty,
\end{equation}
which is a contradiction. Thus, it must be that $\limsup_{\nu \to 0} \theta^*(\nu) = 0$.
%===================================================================================%
\section{A convergence result for convex functions}\label{appendix: convex}
\begin{lemma}\label{lemma: convex}
Let $\{f_n\}$ be a sequence of continuous convex functions which converge point wise to a continuous function $f$ on an interval $[a, b]$. Then $f_n$ converge to $f$ uniformly.
\end{lemma}
\begin{proof}
Let $\epsilon > 0$. We'll show that there exists a large enough $N$ such that for all $n > N$, $||f_n-f||_\infty < \epsilon$.
\smallskip

The function $f$ is continuous on a compact set, and therefore is uniformly continuous. Choose a $\delta > 0$ such that $|f(x)- f(y)| < \epsilon/10$ for $|x-y| < \delta$. Let $M$ be such that $(b-a)/M < \delta$. We divide the interval $[a,b]$ into $M$ intervals, whose endpoints are equidistant. We denote them by $a = \alpha_0 < \alpha_1 < \cdots < \alpha_M = b$. Since $f_n(\alpha_i) \to f(\alpha_i)$, there exists a $N_i$ such that for all $n > N_i$, $|f_n(\alpha_i) - f(\alpha_i)| < \epsilon/10$. Choose $N = \max(M, N_0, \cdots, N_M)$.
\smallskip

Consider an $x \in (\alpha_i, \alpha_{i+1})$ for some $0 \leq i < M$, and let $n>N$. Using uniform continuity of $f$, we have
\begin{equation}\label{eq: bound on fx}
f(\alpha_i) -\epsilon/10 < f(x) < f(\alpha_i) + \epsilon/10.
\end{equation}
Further, we also have
\begin{align*}
f_n(\alpha_i) &\leq f(\alpha_i) + \epsilon/10~,(\text{by pointwise convergence at } \alpha_i)\\
f_n(\alpha_{i+1}) &\leq f(\alpha_{i+1}) + \epsilon/10~,(\text{by pointwise convergence at } \alpha_{i+1})\\
&\leq f(\alpha_i) + 2\epsilon/10~. (\text{by uniform continuity of } f)
\end{align*}
Convexity of $f_n$ implies
\begin{equation}\label{eq: lower bound on fnx}
f_n(x) < \max(f_n(\alpha_i), f_n(\alpha_{i+1})) < f(\alpha_i) + 2\epsilon/10.
\end{equation}
Combining part of equation (\ref{eq: bound on fx}) and equation (\ref{eq: lower bound on fnx}), we obtain
\begin{equation}\label{eq: first half}
f_n(x) - f(x) < 3\epsilon/10.
\end{equation}
We'll now try to upper bound $f_n(x)$. First consider the case when $i \geq 1$. In this case we have 
$$\alpha_{i-1} < \alpha_i < x < \alpha_{i+1}.$$
We write $\alpha_i$ as a linear combination of $x$ and $\alpha_{i-1}$, and use the convexity of $f_n$ to arrive at
\begin{align*}
&f_n(\alpha_i) \leq \frac{\alpha_i - \alpha_{i-1}}{x - \alpha_{i-1}}f_n(x) + \frac{x-\alpha_i}{x-\alpha_{i-1}}f_n(\alpha_{i-1}).
\end{align*}
This implies
\begin{align*}
 \frac{x-\alpha_{i-1}}{\alpha_i - \alpha_{i-1}}f_n(\alpha_i) - \frac{x-\alpha_i}{\alpha_i - \alpha_{i-1}}f_n(\alpha_{i-1}) \leq f_n(x).
\end{align*}
Taking the infimum of the left side, we get
\begin{align*}
\inf_{x \in (\alpha_i, \alpha_{i+1})} \frac{x-\alpha_{i-1}}{\alpha_i - \alpha_{i-1}}f_n(\alpha_i) - \frac{x-\alpha_i}{\alpha_i - \alpha_{i-1}}f_n(\alpha_{i-1}) \leq f_n(x).
\end{align*}
Note that since the LHS is linear in $x$, the infimum occurs at one of the endpoints of the interval, $\alpha_i$ or $\alpha_{i+1}$. Substituting, we get
\begin{align}
f_n(x) &\geq \min\left( f_n(\alpha_i), 2f_n(\alpha_i) - f_n(\alpha_{i-1})\right) \nonumber\\
&\geq \min(f(\alpha_i)-\epsilon/10, 2(f(\alpha_i)-\epsilon/10) - f(\alpha_{i-1}) - \epsilon/10)~\nonumber \\
&\geq \min(f(\alpha_i)-\epsilon/10, 2f(\alpha_i) - f(\alpha_{i-1}) - 3\epsilon/10)~\nonumber \\
&\geq \min(f(\alpha_i)-\epsilon/10, 2f(\alpha_i) - f(\alpha_i) - \epsilon/10 - 3\epsilon/10)\nonumber \\
&= f(\alpha_i) - 4\epsilon/10. \label{eq: upper bound on fnx}
\end{align}
Combining inequality (\ref{eq: upper bound on fnx}) with a part of inequality (\ref{eq: bound on fx}), we have
\begin{equation}\label{eq: second half}
f_n(x) - f(x) > -5\epsilon/10.
\end{equation}
Combining (\ref{eq: first half}) and (\ref{eq: second half}) we conclude that for all $x \in (\alpha_1, \alpha_M)$, and for all $n >N$, 
\begin{equation}
|f_n(x) - f(x)| < \epsilon/2.
\end{equation}
Now let $x \in (\alpha_0, \alpha_1)$. We can establish inequality \eqref{eq: first half} for $x \in (\alpha_0, \alpha_1)$ using the same steps as above. We express $\alpha_1$ as a linear combination of $x$ and $\alpha_2$ and follows the steps as above to establish \eqref{eq: second half} for $x \in (\alpha_0, \alpha_1)$. This shows that for all $x \in [a, b]$, $||f_n(x) - f(x)|| < \epsilon/2$ for all $n > N$, and concludes the proof.
\end{proof}

%=============================================================================%
\section{Convergence properties of sub-convolutive sequences}\label{appendix: subc}

Consider a sequence of functions $\{\mu_n(\cdot)\}_{n\geq 1}$, such that for every $n$, $\mu_n: \mathbb Z_+ \to \mathbb{R_+}$ with $\mu_n(j) = 0$ for all $j \geq n+1$. We call such a sequence of functions a \emph{sub-convolutive} sequence if for all $m, n \geq 1$ the convolution $\mu_m \star \mu_n$ pointwise dominates $\mu_{m+n}$; i.e., 

\begin{equation}\label{eq: define subc}
\mu_m \star \mu_n (i) \geq \mu_{m+n} (i)  \text{ for all } i \geq 0 \text{, and for all } m,n \geq 1.
\end{equation}

For our results on sub-convolutive  sequences, we make the following assumptions:
\begin{align*}
&\mathbf{(A):} ~\alpha := \lim_{n \to \infty} \frac{1}{n} \log \mu_n(n) \text{~~~ is finite.}\\
&\mathbf{(B):} ~\beta := \lim_{n \to \infty} \frac{1}{n} \log \mu_n(0) \text{~~~ is finite.}\\
&\mathbf{(C):} \text{~For all } n,~ \mu_n(n) >0, \mu_n(0)>0.
\end{align*}
Note that  $\mu_m \star \mu_n(m+n) = \mu_n(n)\mu_m(m)$ and $\mu_m \star \mu_n(0) = \mu_n(0)\mu_m(0)$. Thus, the existence of the limits in assumptions $\mathbf{(A)}$ and $\mathbf{(B)}$ is guaranteed by Fekete's Lemma, and we have
\begin{align}
\alpha &= \inf_n \frac{1}{n}\log \mu_n(n)\\
\beta &= \inf_n \frac{1}{n} \log \mu_n(0).
\end{align}

For $n \geq 1$, define $G_n: \mathbb{R} \to \mathbb{R}$ as 
\begin{equation}
G_n(t) = \log \sum_{j=0}^n \mu_n(j)e^{jt}.
\end{equation}
Condition (\ref{eq: define subc}) implies that the functions $G_n$ satisfy the inequality,
\begin{equation}\label{eq: subc condition on G}
G_m(t) + G_n(t) \geq G_{m+n}(t) \text{ for every } m,n \geq 1 \text{ and for every } t.
\end{equation}
Thus for each $t$, the sequence $\{G_n(t)\}$ is sub additive, and by Fekete's lemma the limit $\lim_n \frac{G_n(t)}{n}$ exists. 
To simply notation a bit, define $g_n := \frac{G_n}{n}$ and let $\Lambda$ be defined as the pointwise limit of $g_n$'s; i.e.,
\begin{equation}
\Lambda(t) =  \lim_n g_n(t).
\end{equation}
\begin{lemma}\label{lemma: lambda sub}
The function $\Lambda$ satisfies the following properties:
\begin{enumerate}
\item
For all $t$,
\begin{equation}
\max(\beta, t +\alpha) \leq \Lambda(t) \leq g_1(t)
\end{equation}
\item
$\Lambda$ is convex and monotonically increasing.
\item
Let $\Lambda^*$ be the convex conjugate of $\Lambda$. The domain of $\Lambda^*$ is $[0,1]$.
\end{enumerate}
\end{lemma}
\begin{proof}
\begin{enumerate}
\item
The inequality \eqref{eq: subc condition on G} immediately gives that for all $t$, and all $n \geq 1$,
\begin{align}\label{eq: subc g1gn}
nG_1(t) \geq G_n(t), \text{~~which implies~~} g_1(t) \geq g_n(t).
\end{align}
Taking the limit in $n$, it follows that $\Lambda(t) \leq g_1(t)$ for all $t$. \\
For all $n$, the functions $g_n$ are monotonically increasing, and for all $t$ they satisfy
\begin{equation}
g_n(t) \geq \lim_{t \to -\infty} g_n(t) = \frac{1}{n} \log \mu_n(0).
\end{equation}
In addition, we also know that $$\inf_n \frac{1}{n} \log \mu_n(0) = \beta.$$
This gives us that
\begin{equation}
g_n(t) \geq \beta.
\end{equation}
Taking the limit in $n$, we conclude that for all $t$,
\begin{equation}\label{eq: lambda beta}
\Lambda(t) \geq \beta.
\end{equation}
For all $n$, we have the lower bound on $g_n$ given by
\begin{align}
g_n(t) &= \frac{1}{n} \log \sum_{j=0}^n \mu_n(j)e^{jt}\\
&\geq \frac{1}{n} \log \mu_n(n)e^{nt}\\
&= t + \frac{1}{n} \log \mu_n(n)\\
&\stackrel{(a)}\geq t + \alpha
\end{align}
where $(a)$ follows as 
$$\inf_n \frac{1}{n} \log \mu_n(n) = \alpha.$$
Taking the limit in $n$, we conclude that
\begin{equation}\label{eq: lambda alpha}
\Lambda(t) \geq t+\alpha.
\end{equation}
Equations \eqref{eq: lambda beta} and \eqref{eq: lambda alpha} establish  $$\Lambda(t) \geq \max(\beta, t + \alpha).$$

\item

The functions $\{g_n\}$ are convex and monotonically increasing. Since $\Lambda$ is the pointwise limit of these functions, $\Lambda$ is also convex and monotonically increasing.

\item
Note that the convex conjugates of the functions $g_1(t)$ and $\max(\beta, t+\alpha)$ are both supported on $[0,1]$. Since $\Lambda$ is trapped between these two functions, it is clear that $\Lambda^*$ is also supported on $[0,1]$.
\end{enumerate}
\end{proof}

Theorem \ref{thm: subc upper bound} requires an application of the  G\"{a}rtner-Ellis theorem \cite{dembo1998large}, which we state here for reference:
\begin{theorem}[G\"{a}rtner-Ellis theorem]\label{thm: gartner ellis}
Consider a sequence of random vectors $Z_n \in \mathbb R^d$, where $Z_n$ possess the law $\mu_n$ and the logarithmic moment generating function
$$\Lambda_n(\lambda) := \log E\left[\exp \langle \lambda, Z_n\rangle\right].$$ We assume the following:
\begin{itemize}
\item[$(\star)$:]
 For each $\lambda\in \mathbb R^d$, the logarithmic moment generating function, defined as the limit
$$\Lambda(\lambda) := \lim_{n \to \infty} \frac{1}{n} \Lambda_n(n\lambda)$$
exists as an extended a real number. Further the origin belongs to the interior $\cD_\Lambda := \{\lambda \in \mathbb R^d~|~\Lambda(\lambda) < \infty\}$.
\end{itemize}
 Let $\Lambda^*$ be the convex conjugate of $\lambda$ with $\cD_{\Lambda^*} = \{x \in \mathbb R^d~|~ \Lambda^*(x) < \infty\}$. When assumption $(\star)$ holds, the following are satisfied:
 \begin{itemize}
 \item[1.] For any closed set $I$,
 $$\limsup_{n \to \infty} \frac{1}{n} \log \mu_n(I) \leq -\inf_{x \in I} \Lambda^*(x).$$
 \item[2.] For any open set $F$,
 $$\liminf_{n \to \infty} \frac{1}{n} \log \mu_n(F) \geq  -\inf_{x \in F \cap \cF} \Lambda^*(x),$$
 where $\cF$ is the set of exposed points of $\Lambda^*$ whose exposing hyperplane belongs to the interior of $\cD_{\Lambda}$.
 \item[3.]
 If $\Lambda$ is an essentially smooth, lower semicontinuous function, then the large deviations principle holds with a good rate function $\Lambda^*$.
 \end{itemize}
\end{theorem}
\begin{remark}
For definitions of \emph{exposed points}, \emph{essentially smooth functions}, \emph{good rate function}, and the \emph{large deviations principle} we refer to Section $2.3$ of \cite{dembo1998large}. For our purpose, it is enough to know that if $\Lambda$ is differentiable on $\cD_{\Lambda} = \mathbb R^d$, then it is essentially smooth and $\Lambda^*$ satisfies the large deviation principle.
\end{remark}

\begin{theorem}\label{thm: subc upper bound}
Let $\{\mu_n\}_{n \geq 1}$ be a sequence of sub-convolutive of functions as defined in equation \eqref{eq: define subc}, satisfying assumptions $\mathbf{(A),  (B)}$ and $\mathbf{(C)}$. Define a sequence of measures supported on $[0,1]$ by 
$$\mu_{n/n}\left(\frac{j}{n}\right) := \mu_n(j)  \text{~~for~~}  j \geq 0.$$
 Let $I \subseteq \mathbb{R}$ be a closed set. The family of measures $\{\mu_{n/n}\}$ satisfies the large deviation upper bound
\begin{equation}
\limsup_{n \to \infty} \frac{1}{n} \log \mu_{n/n}(I) \leq - \inf_{x \in I} \Lambda^*(x).
\end{equation}
\end{theorem}

\begin{proof}
Let $\sum_j \mu_n(j) =  s_n$. We first normalize $\mu_{n/n}$ to define the probability measure
$$p_n := \frac{\mu_{n/n}}{s_n}.$$
The log moment generating function of $p_n$, which we call $P_n$, is given by
\begin{align*}
P_n(t) &= \log \sum_{j=0}^n p_n(j/n)e^{jt/n}\\
&= \log \frac{1}{s_n}\sum_{j=0}^n \mu_n(j)e^{jt/n}\\
&= G_n(t/n) - \log s_n.
\end{align*}
Thus, 
\begin{align*}
\lim_{n \to \infty} \frac{1}{n} P_n(nt) &= \lim_{n \to \infty}\left(\frac{G_n(t)}{n} - \frac{\log s_n}{n}\right)\\
&= \Lambda(t) - \Lambda(0).
\end{align*}
Note also that by Lemma \ref{lemma: lambda sub}, the function $\Lambda$ is finite on all of $\mathbb R$, and thus 0 lies in the interior $\cD(\Lambda)$. Thus, the sequence of probability measures $\{p_n\}$ satisfies the condition $(\star)$ required in the G\"{a}rtner-Ellis theorem. A direct application of this theorem gives the bound
\begin{align*}
\limsup_{n \to \infty} \frac{1}{n} \log p_n(I) &\leq -\inf_{x \in I} (\Lambda(x) - \Lambda(0))^*\\
&= -\inf_{x \in I} \Lambda^*(x) - \Lambda(0),
\end{align*}
which immediately gives 
\begin{align*}
\limsup_{n \to \infty} \frac{1}{n} \log \mu_{n/n}(I) \leq -\inf_{x \in I} \Lambda^*(x).
\end{align*}
\end{proof}
\begin{remark}
If $\Lambda(t)$ is differentiable, we can apply the G\"{a}rtner-Ellis theorem to get a lower bound of the form
$$\liminf_{n \to \infty} \frac{1}{n} \log \mu_{n/n}(F) \geq - \inf_{x \in F} \Lambda^*(x),$$
for every open set $F$. However, it is easy to construct sub-convolutive sequences such that $\Lambda(t)$ is not differentiable. One example is the sequence $\{\mu_n\}$ such that for each $n$,
$$\mu_n(j) = 
\begin{cases}
1 &\text{ if } j = 0 \text{ or } n,\\
0 &\text{ otherwise. }
\end{cases}
$$
\end{remark}
\begin{theorem}\label{theorem: subc gnstar}
The functions $\{g_n^*\}$ converge uniformly to $\Lambda^*$ on $[0,1]$.
\end{theorem}
\begin{proof}
We'll show that $\{g_n^*\}$ converge pointwise to $\Lambda^*$ on $[0,1]$. Since $g_n^*$ and $\Lambda^*$ are all continuous convex functions on a compact set, Lemma \ref{lemma: convex} implies that this pointwise convergence implies uniform convergence.\\

Recall that $\alpha = \inf_n \frac{1}{n} \log \mu_n(n)$, $\beta = \inf_n \frac{1}{n} \log \mu_n(0)$, and $\inf_n g_n(0) = \Lambda(0)$. 
Fix an $x \in (0,1)$, and define 
$$\arg \max_t xt - g_n(t) := t_n.$$
Clearly, $g_n^*(x) = xt_n - g_n(t_n)$. Note that
\begin{align}
g_n^*(x) &\geq xt - g_n(t) \Big|_{t = 0}\\
& = -g_n(0)\\
&\stackrel{(a)} \geq -g_1(0)
\end{align}
where $(a)$ follows by inequality \eqref{eq: subc g1gn}.

 If $t > \frac{g_1(0)-\alpha}{1-x}$, then we have
 \begin{align*}
 xt - g_n(t) &< xt - (t + \alpha)\\
 &= -(1-x)t - \alpha \\
 &< -(g_1(0) - \alpha) - \alpha\\
 &= -g_1(0).
 \end{align*}
 This gives us that $t_n \leq \frac{g_1(0)-\alpha}{1-x}$. Similarly, if $t < \frac{\beta - g_1(0)}{x}$, then
 \begin{align}
 xt - g_n(t) &< xt - \beta\\
 &<(\beta - g_1(0)) - \beta\\
 &= -g_1(0).
 \end{align}
 This gives us that $t_n \geq \frac{\beta - g_1(0)}{x}$. We can thus conclude that for all $n$,
\begin{equation}\label{eq: subc bounds on tn}
 t_n \in \left[ \frac{\beta - g_1(0)}{x}, \frac{g_1(0)-\alpha}{1-x} \right] := I_x.
\end{equation}
Note that all we used to prove relation \eqref{eq: subc bounds on tn} is that $g_n$ is trapped between $g_1$ and $\max(\beta, t+\alpha)$. Since $\Lambda$ also satisfies this, we have
\begin{equation}\label{eq: subc lambda tn}
\arg\max_{t} xt - \Lambda(t) \in I_x.
\end{equation}
We now restrict our attention to the compact interval $I_x$. Let $\hat g_n$ be $g_n$ restricted to $I_x$. The convex functions $\hat g_n$ converge pointwise to a continuous limit $\hat \Lambda$, where $\hat \Lambda$ is $\Lambda$ restricted to $I_x$. This convergence must therefore be uniform, which implies convergence of $\hat g_n^*(x)$ to $\hat \Lambda^*(x)$. Furthermore, relation (\ref{eq: subc bounds on tn}) implies $\hat g_n^*(x)$ equals $g_n^*(x)$, and relation \eqref{eq: subc lambda tn} gives $\hat \Lambda^*(x)$ equals $\Lambda^*(x)$. Thus, $g_n^*(\cdot)$ converges pointwise to $\Lambda^*(\cdot)$ on $(0,1)$.

We'll now consider convergence at the boundary points. Let $\epsilon > 0$ be given. Choose a subsequence $\{g_{n_k}\}$ where $n_k = 2^k$. Using the condition in (\ref{eq: subc condition on G}), it is clear that $\{g_{n_k}\}$ decrease monotonically and converge pointwise to $\Lambda$. Choose $K_0$ large enough such that for all $k > K_0$,
\begin{align}
 \frac{1}{n_k} \log\mu_{n_k}(0) - \beta  &  < \epsilon/2.
\end{align}
Note that the left hand side is non-negative, and we need not use absolute values. Choose a $T_0$ such that for all $t < T_0$,
\begin{align}
g_{n_{K_0}}(t) - \frac{1}{n_{K_0}}\log \mu_{n_{K_0}}(0) < \epsilon/2.
\end{align}
Now for all $k > K_0$ and all $t < T_0$, the following holds:
\begin{align}
g_{n_k}(t) - \beta \leq g_{n_{K_0}}(t) - \beta < \frac{1}{n_{K_0}}\log \mu_{n_{K_0}}(0) + \epsilon/2 - \beta < \epsilon.
\end{align}
Taking the limit in $k$, we get that for all $t < T_0$,
\begin{align}
\Lambda(t) - \beta \leq \epsilon,
\end{align}
this along with the lower bound $\Lambda(t) \geq \beta$ gives that for all $t < T_0$,
\begin{align*}
0 \leq \Lambda(t) - \beta \leq \epsilon.
\end{align*}
We also have 
\begin{align}
\Lambda^*(0) &= \sup_t -\Lambda(t)\\
 &= - \lim_{t \to -\infty} \Lambda(t),
\end{align}
which must equal $-\beta$. Since the limit of $g_n^*(0)$ is also $-\beta$, we have shown convergence of $g_n^*$ to $\Lambda$ at $t = 0$.\\

To show convergence at $t=1$, we follow a similar strategy. Let $\{g_{n_k}\}$ be as before, and let $\epsilon > 0$ be given. We choose a $K_1$ such that for all $k > K_1$, 
\begin{align}
 \frac{1}{n_k} \log\mu_{n_k}(n_k) - \alpha  &  < \epsilon/2.
\end{align}
Note that the left hand side is non-negative, and we need not use absolute values. We now choose a $T_1$ such that for all $t > T_1$,
\begin{align}
g_{n_{K_1}}(t) - \left( t+ \frac{1}{n_{K_1}}\log \mu_{n_{K_1}}(n_{K_1}) \right)< \epsilon/2.
\end{align}
Now for all $k > K_1$ and all $t > T_1$,
\begin{align}
g_{n_k}(t) - (t+\alpha) \leq g_{n_{K_1}}(t) - (t+\alpha) < \frac{1}{n_{K_1}}\log \mu_{n_{K_1}}(n_{K_1})  + \epsilon/2 -\alpha < \epsilon.
\end{align}
Taking the limit in $k$, we get that for all $t > T_1$,
\begin{align*}
\Lambda(t) - (t+\alpha) \leq \epsilon,
\end{align*}
this along with the lower bound $\Lambda(t) \geq t+\alpha$ gives that for all $t > T_1$
\begin{align*}
0 \leq \Lambda(t) - (t+\alpha) \leq \epsilon.
\end{align*}
From this, we conclude that $\Lambda^*(1) = \sup_t t-\Lambda(t) =  \lim_{t \to +\infty} t-\Lambda(t)$, must equal $-\alpha$. Since the limit of $g_n^*(1)$ is also $-\alpha$, we have shown convergence of $g_n^*$ to $\Lambda$ at $t = 1$.\\

This shows that $\{g_n^*\}$ converges pointwise to $\Lambda^*$ on the compact interval $[0,1]$. As all the functions involved are continuous and convex, by Lemma \ref{lemma: convex} this convergence must also be uniform. This concludes the proof.
\end{proof}

%%%%%%%%%%%%%%%%%%%%%%%%%%%%%%%%%%%%%%%%%%%%%%%%%%%%%%%%%%%%%%%%%%%

\end{appendix}

%=============================================================================%

\bibliographystyle{ieeetr}	% (uses file "plain.bst")
\bibliography{myrefs}		% expects file "myrefs.bib"

\begin{thebibliography}{10}

\bibitem{shannon}
C.~Shannon, ``A mathematical theory of communication, {I} and {II},'' {\em Bell
  Syst. Tech. J}, vol.~27, pp.~379--423, 1948.

\bibitem{smith1971information}
J.~G. Smith, ``The information capacity of amplitude-and variance-constrained
  scalar gaussian channels,'' {\em Information and Control}, vol.~18, no.~3,
  pp.~203--219, 1971.

\bibitem{shamai1995capacity}
S.~Shamai and I.~Bar-David, ``The capacity of average and peak-power-limited
  quadrature {G}aussian channels,'' {\em IEEE Transactions on Information
  Theory}, vol.~41, no.~4, pp.~1060--1071, 1995.

\bibitem{sudevalayam2011energy}
S.~Sudevalayam and P.~Kulkarni, ``Energy harvesting sensor nodes: Survey and
  implications,'' {\em Communications Surveys \& Tutorials, IEEE}, vol.~13,
  no.~3, pp.~443--461, 2011.

\bibitem{cruz1991calculus1}
R.~L. Cruz, ``A calculus for network delay {P}art {I}: Network elements in
  isolation,'' {\em IEEE Transactions on Information Theory}, vol.~37, no.~1,
  pp.~114--131, 1991.

\bibitem{cruz1991calculus2}
R.~L. Cruz, ``A calculus of delay {P}art {I}{I}: Network analysis,'' {\em IEEE
  Transactions on Information Theory}, vol.~37, no.~1, pp.~132--141, 1991.

\bibitem{ozel2012achieving}
O.~Ozel and S.~Ulukus, ``Achieving {AWGN} capacity under stochastic energy
  harvesting,'' {\em IEEE Transactions on Information Theory}, vol.~58, no.~10,
  pp.~6471--6483, 2012.

\bibitem{tutuncuoglu2013binary}
K.~Tutuncuoglu, O.~Ozel, A.~Yener, and S.~Ulukus, ``Binary energy harvesting
  channel with finite energy storage,'' in {\em Proceedings of the 2013
  International Symposium on Information Theory (ISIT)}, pp.~1591--1595, IEEE,
  2013.

\bibitem{tutuncuoglu2014improved}
K.~Tutuncuoglu, O.~Ozel, A.~Yener, and S.~Ulukus, ``Improved capacity bounds
  for the binary energy harvesting channel,'' in {\em Proceedings of the 2014
  International Symposium on Information Theory (ISIT)}, p.~976–980, IEEE,
  2014.

\bibitem{mao2013capacity}
W.~Mao and B.~Hassibi, ``On the capacity of a communication system with energy
  harvesting and a limited battery,'' in {\em Proceedings of the 2013
  International Symposium on Information Theory (ISIT)}, pp.~1789--1793, IEEE,
  2013.

\bibitem{dong2014near}
Y.~Dong, F.~Farnia, and A.~{\"O}zg{\"u}r, ``Near optimal energy control and
  approximate capacity of energy harvesting communication,'' {\em arXiv
  preprint arXiv:1405.1156}, 2014.

\bibitem{schneider2013convex}
R.~Schneider, {\em Convex bodies: the {B}runn—{M}inkowski theory}, vol.~151.
\newblock Cambridge University Press, 2013.

\bibitem{schneider2008stochastic}
R.~Schneider and W.~Weil, {\em Stochastic and integral geometry}.
\newblock Springer, 2008.

\bibitem{dembo1998large}
A.~Dembo and O.~Zeitouni, {\em Large deviations techniques and applications},
  vol.~2.
\newblock Springer, 1998.

\bibitem{cover1991}
T.~Cover, J.~Thomas, J.~Wiley, {\em et~al.}, {\em Elements of information
  theory}, vol.~6.
\newblock Wiley Online Library, 1991.

\bibitem{dobrushin1963general}
R.~Dobrushin, ``General formulation of shannon’s main theorem in information
  theory,'' {\em Amer. Math. Soc. Trans}, vol.~33, pp.~323--438, 1963.

\bibitem{steele1997probability}
J.~M. Steele, {\em Probability theory and combinatorial optimization}, vol.~69.
\newblock SIAM, 1997.

\bibitem{klainrota}
D.~A. Klain and G.-C. Rota, {\em Introduction to geometric probability}.
\newblock Cambridge University Press, 1997.

\bibitem{michalowicz2008calculation}
J.~V. Michalowicz, J.~M. Nichols, and F.~Bucholtz, ``Calculation of
  differential entropy for a mixed gaussian distribution,'' {\em Entropy},
  vol.~10, no.~3, pp.~200--206, 2008.

\bibitem{durrett2010probability}
R.~Durrett, {\em Probability: theory and examples}, vol.~3.
\newblock Cambridge University Press, 2010.

\bibitem{hopcroftfoundations}
J.~Hopcroft and R.~Kannan, ``Foundations of data science,'' {\em Available
  online at
  http://research.microsoft.com/en-us/people/kannan/book-dec-30-2013.pdf}.

\bibitem{anselone1971}
P.~M. Anselone and J.~Davis, {\em Collectively compact operator approximation
  theory and applications to integral equations}, vol.~1971.
\newblock Prentice-Hall Englewood Cliffs, NJ, 1971.

\bibitem{schaefer1999}
H.~Schaefer and M.~Wolff, {\em Topological Vector Spaces}.
\newblock Graduate Texts in Mathematics, Springer New York, 1999.

\bibitem{royden2011}
H.~L. Royden and P.~Fitzpatrick, {\em Real analysis, 4th edition}.
\newblock Pearson, 2011.

\bibitem{boyd2009convex}
S.~Boyd and L.~Vandenberghe, {\em Convex optimization}.
\newblock Cambridge university press, 2009.

\bibitem{mcmullen1991inequalities}
P.~McMullen, ``Inequalities between intrinsic volumes,'' {\em Monatshefte
  f{\"u}r Mathematik}, vol.~111, no.~1, pp.~47--53, 1991.

\end{thebibliography}

\end{document}